\documentclass{article}
\usepackage[utf8]{inputenc}
\usepackage{amsfonts, amsmath, amssymb, mathtools, bm, amsthm}
\usepackage[colorlinks, citecolor = blue]{hyperref}
\usepackage{multirow, makecell}
\usepackage{geometry}
\usepackage{diagbox}
\RequirePackage{authblk}
\usepackage{natbib}
\usepackage[perpage]{footmisc}

\title{Determining the Number of Communities in Sparse and Imbalanced Settings}
\author[1]{Zhixuan Shao\thanks{zhxshao@ucdavis.edu}}
\author[1]{Can M. Le\thanks{canle@ucdavis.edu}}
\affil[1]{Department of Statistics, University of California, Davis}
\date{}
\newtheorem{theorem}{Theorem}
\newtheorem{proposition}{Proposition}
\newtheorem{conjecture}{Conjecture}
\newtheorem{lemma}{Lemma}

\newcommand{\E}{\mathbb{E}}
\newcommand{\Prob}{\mathbb{P}}
\newcommand{\Var}{\mathrm{Var}}
\newcommand{\Phat}{\widehat{P}}
\newcommand{\tg}{\textsl{g}}
\newcommand{\Hc}{\underline{H}}
\newcommand{\Dc}{\underline{D}}
\newcommand{\dc}{\underline{d}}
\newcommand{\Bc}{\underline{B}}
\newcommand{\Ac}{\underline{A}}
\newcommand{\vc}{\underline{v}}
\newcommand{\lambdac}{\underline{\lambda}}
\newcommand{\Hce}{\widehat{\underline{H}}}
\newcommand{\Dce}{\widehat{\underline{D}}}
\newcommand{\dce}{\widehat{\underline{d}}}

\newcommand{\Ace}{\widehat{\underline{A}}}
\newcommand{\An}{\widetilde{A}}

\newcommand{\Hn}{\widetilde{H}}

\newcommand{\Ane}{\widehat{\widetilde{A}}}

\newcommand{\In}{\bm{I}_n}
\begin{document}

\maketitle

\begin{abstract}
    Community structures represent a crucial aspect of network analysis, and various methods have been developed to identify these communities.
    However, a common hurdle lies in determining the number of communities $K$, a parameter that often requires estimation in practice.
    Existing approaches for estimating $K$ face two notable challenges: the weak community signal present in sparse networks and the imbalance in community sizes or edge densities that result in unequal per-community expected degree.
    We propose a spectral method based on a novel network operator whose spectral properties effectively overcome both challenges.
    This operator is a refined version of the non-backtracking operator, adapted from a ``centered" adjacency matrix.
    % Similar to the conventional non-backtracking matrix, it empirically exhibits weak Ramanujan property $\lambda_1=2\sqrt{d} + o(1)$ under sparse settings with constant average degree $d$.
    Its leading eigenvalues are more concentrated than those of the adjacency matrix for sparse networks, while they also demonstrate enhanced signal under imbalance scenarios, a benefit attributed to the centering step.
    % This is partially justified under the null model $K=1$ in dense and ultra-sparse settings.
    This is justified, either theoretically or numerically,
    under the null model $K=1$, in both dense and ultra-sparse settings.
    A goodness-of-fit test based on the leading eigenvalue can be applied to determine the number of communities $K$.  
    
    % We analyze the asymptotic null distribution of the largest eigenvalue of the proposed ``centered" non-backtracking matrix under $H_0: K=1$.
    % We claim its convergence to the Tracy-Widom distribution in the dense regime, supported by empirical evidence.
    % In addition, we provide an explanation for its efficacy for ultra-sparse networks.
    % A goodness-of-fit test based on this operator can then be applied sequentially or recursively to determine the number of communities $K$.
\end{abstract}

\section{Introduction}
Network data have wide-ranging applications in various real-world problems, including social networks (Facebook friendship, LinkedIn following, etc.), biological networks (gene network, gene-protein network), information networks (email network, World Wide Web), among others. 
One key feature in many of these networks is the presence of communities, which naturally partition the network into clusters of nodes with high internal connectivity.
These community structures offer valuable insights for network analysis.
% Consider social networks, for instance; communities often emerge based on shared interests and geographical locations.
It is, therefore, of utmost interest to identify these communities.

% Numerous heuristic algorithms have been proposed to identify these communities.
Typically, a generative model is essential for approaching the problem theoretically.
Due to its simplicity, the stochastic block model (SBM) \citep{holland1983stochastic} is arguably the most widely used framework for modeling community structures. In this model, each node is assigned a latent block label, and the connecting probabilities are determined by their block memberships.
% In practice, this model occasionally oversimplifies the complexity in real-world networks by assuming all nodes within the same communities have the same degree distribution.
% To address this limitation, 
An important extension of this model is
the degree-corrected SBM (DCSBM) \citep{karrer2011stochastic},
which introduces a degree parameter for each node, 
allowing for the modeling of degree heterogeneity.

Much research effort has been devoted to the problems of estimating the latent
block memberships in an SBM or DCSBM. There are mainly two types of methods. The first type is based on maximizing likelihood or modularities \citep{newman2006modularity, bickel2009nonparametric, karrer2011stochastic, zhao2012consistency}. 
Since the problem of optimizing over all possible label assignments is computationally infeasible, many variants are proposed, including 
pseudo-likelihood \citep{amini2013pseudo}, 
convex optimization \citep{abbe2015exact, chen2018convexified, amini2018semidefinite, li2021convex},
and variational inference \citep{latouche2012variational, celisse2012consistency, bickel2013asymptotic}, to list a few.

The second type is referred to as spectral methods.
It applies a standard clustering method (e.g., $k$-means) on relevant eigenvectors of a matrix (also called a graph operator) built from the graph.
The most common choices use the adjacency matrix and the Laplacian matrix
\citep{mcsherry2001spectral, fishkind2013consistent, lei2015consistency, jin2015fast}. 
However, it is noticed that the adjacency matrix or the Laplacian matrix does not concentrate well under sparsity \citep{le2017concentration}, particularly in the ultra-sparse regime when the average expected degree is of constant order. 
This observation has led to the development of regularized versions of these matrices \citep{chaudhuri2012spectral, rohe2011spectral, sarkar2015role, le2015sparse}.
On the other hand, the spectrum of two different matrices, the non-backtracking matrix
\citep{krzakala2013spectral, bordenave2015non, gulikers2017non} and the Bethe-Hessian matrix \citep{saade2014spectral, dall2019revisiting}, is much better behaved under this regime.
In particular, the non-backtracking matrix arises in connection to belief propagation \citep{watanabe2009graph, decelle2011asymptotic}, whose spectrum is shown to detect communities all the way down to the theoretical limit in the symmetric case \citep{krzakala2013spectral, bordenave2015non, mossel2018proof, massoulie2014community}.

However, most of these methods assume that the number of communities $K$ is given \textit{a priori} as an input, which is often unknown in practice.
Several methods are proposed to estimate $K$, which can be considered a model selection problem.
Likelihood-based methods \citep{saldana2017many, wang2017likelihood, hu2020corrected} use BIC-type criteria, 
which requires optimizing the maximum likelihood and are computationally challenging.
% , done using either MCMC or the variational method, which are computationally challenging for large networks.
% Since the rate of convergence is slow for sparse networks, a bootstrap correction procedure was employed, which also leads to a high computational cost.
\citet{yan2018provable} proposed a semi-definite programming approach which estimates $K$ and recovers the memberships simultaneously.
A Bayesian approach with a new prior and an efficient sampling scheme was proposed by \citet{riolo2017efficient}.
More recently, \citet{ma2021determining} combined spectral clustering with binary segmentation to derive a new estimate of $K$ based on a pseudo likelihood ratio.
Cross-validation approaches were proposed by \citet{chen2018network, li2020network}.
While they have shown consistency guarantee for SBM, they require estimating communities on many random network splits, which is also computationally costly.
Alternatively, spectral methods are typically simple yet computationally efficient. 
Of these, the USVT method \citep{chatterjee2015matrix} estimates $K$ by simply thresholding the eigenvalues of the adjacency matrix appropriately.
A more recent rank inference method was proposed by \citet{han2023universal} via residual sampling.
It is worth mentioning that many of these methods are essentially goodness-of-fit tests, which can be applied sequentially or recursively to determine the number of communities $K$.
In particular, a goodness-of-fit test based on the leading eigenvalue of a normalized adjacency matrix was proposed separately by \citet{bickel2016hypothesis} and \citet{lei2016goodness}, which we will introduce in more detail in Section~\ref{section:normalized_adjacency}.
More recently, \citet{zhang2023adjusted} proposed a goodness-of-fit test for DCSBM based on an adjusted chi-square statistic, and another step-wise goodness-of-fit test was proposed by \citet{jin2023optimal} based on SCORE \citep{jin2015fast}.
 
Sparsity is one major challenge that confronts all the above methods. Specifically, most of them require the growth of the average degree to be no slower than $\log n$.
Inspired by the favorable spectral properties of the non-backtracking matrix and the Bethe-Hessian matrix,
% \cite{krzakala2013spectral, bordenave2015non, saade2014spectral}, 
the approach that directly counts the ``informative" eigenvalues of these matrices has been proposed by \citet{le2022estimating, dall2021unified, hwang2023estimation}.
For the method of counting the informative eigenvalues of the non-backtracking matrix, its consistency in the ultra-sparse regime with constant average degree is implied by \citet{bordenave2015non} for SBM and extended to DCSBM by \citet{gulikers2017non}.
Furthermore, its consistency in the semi-dense regime with average degree growing faster than $\log n$ has been shown in \citet{le2022estimating}. 
Additionally, within the same semi-dense regime, \citet{le2022estimating} have also demonstrated the consistency of counting the eigenvalues of the Bethe-Hessian matrix with certain heuristic-based choices of the scale parameter $r$.
\citet{dall2021unified} further explored the ultra-sparse regime with degree heterogeneity (DCSBM), proposing an intuitive choice of $r$ such that the largest isolated eigenvalue of the Bethe-Hessian matrix is zero.

However, most of these counting approaches assume an equal expected degree in each community, a setting, from a theoretical point of view, considered most challenging as the degrees alone do not contain information about the latent structure.
However, as noted in \citet{le2022estimating}, a limitation emerges in practice when dealing with unbalanced communities in terms of size or edge density—--the spectrum of these matrices cannot effectively distinguish informative eigenvalues from the bulk, leading to a severe underestimate of $K$.
As demonstrated in Figure~\ref{fig:compare_nonbacktracking_normalized_adjacency_under_community_size_imbalance} and Figure~\ref{fig:compare_nonbacktracking_normalized_adjacency_under_P_imbalance}, where $K=2$ and we have unequal connecting probabilities in two communities, 
the second largest eigenvalue of the non-backtracking matrix is swamped within the bulk.
However, in the meantime, the normalized adjacency matrix, as proposed in \citet{bickel2016hypothesis}  and \citet{lei2016goodness}, clearly shows the existence of an informative eigenvalue despite the imbalance, signaling a second community.
We argue in Section~\ref{section:centering_enhances_signal} that the success of the latter under imbalance can be attributed to the centering process. 

Until very recently, no spectral method has emerged that effectively addresses both challenges---sparsity and imbalance---simultaneously.
\citet{hwang2023estimation} considered counting the eigenvalues of the Bethe-Hessian matrix in the sub-logarithmic sparse regime without requiring equal per-community expected degree.
They propose to choose the parameter $r$ from an oracle interval and show its consistency for estimating $K$.
They also propose an empirical estimate of $r$, which lies in the aforementioned interval with high probability.

Despite the recent literature focusing on the Bethe-Hessian matrix, we maintain that the non-backtracking matrix retains several advantages over it.
Firstly, it is guaranteed effective even under the sparser regime with constant order average degree.
Secondly, and perhaps more crucially for practical applications, it is tuning-free, without the need to choose the scale parameter $r$ or any additional hyperparameters for tuning $r$ as seen in \citet{hwang2023estimation}.

Previous considerations have naturally inspired us to explore the integration of the distinct strengths of two existing methods, intuitively driven by their two key ingredients: non-backtracking and centering/normalization.
We propose a refined version of the non-backtracking operator, adapted from a ``centered" adjacency matrix.
Similar to the conventional non-backtracking matrix, its more concentrated spectrum better captures the signal under the ultra-sparse regime with a bounded average degree.
On the other hand, it demonstrates enhanced signal under imbalance scenarios, a benefit attributed to the centering process.

The rest of the article is organized as follows.
In Section~\ref{section:preliminaries}, we introduce the model and notations and cover the non-backtracking matrix and normalized adjacency matrix in more detail.
In Section~\ref{section:proposed_operator:nonbacktracking_with_centering}, we define the proposed spectrum operator and introduce a goodness-of-fit test based on this operator. We will analyze its null distribution and asymptotic power and delve deeper into the centering mechanism. 
In Section~\ref{section:numerical_experiment}, we provide an extensive simulation study on various settings and compare our proposed test with existing ones.
In Section~\ref{section:determine_K_in_practice} we discuss how to apply this goodness-of-fit test to determine $K$ in practice and present a real-data example.
Finally, we make some comments and discussions in Section~\ref{section:discussions}.

\section{Preliminaries}
\label{section:preliminaries}
We denote by $n$ the number of nodes in the network.
Recall that $A$ is the $n \times n$ symmetric network adjacency matrix with no self-loops.
Let $d_i = \sum_{j=1}^n A_{ij}$ be the degree of node $i$.
Let $P \coloneqq \E[A] $ represent the edge connection probabilities. That is, we assume the edges $A_{ij}$ are independent Bernoulli variables with probabilities given by the entries $P_{ij}$. 

An SBM on $n$ nodes with $K$ communities is parameterized by a membership vector $\bm{g}\in \{1, ..., K\}^n$ and a symmetric block-wise edge probability matrix $Q \in [0,1]^{K \times K}$. 
For this paper, we assume that the node memberships $\bm{g}$ are fixed and unknown.
The edges $A_{ij}$ are then independent Bernoulli random variables with parameters determined by the node memberships,
\begin{equation*}
     \Prob (A_{ij} = 1) = P_{ij} = Q_{g_i g_j}, \quad \forall i \ne j.
\end{equation*}
We assume the model is assortative, that is, $Q_{ii} \ge Q_{ij}$ when $j\ne i$.
% Do we assume Assortative?

For a vector $\bm{m} = (m_1, m_2, ..., m_n)^\top$, let $\mathrm{diag}(\bm{m})$ denote the $n\times n$ diagonal matrix with $i$th diagonal entry being $m_i$. For any $n\times n$ symmetric matrix $U$, let $\lambda_j(U)$ denote its $j$th largest real eigenvalue, arranged in descending order: $\lambda_1(U) \ge \lambda_2(U) \ge \cdots \ge \lambda_n(U)$.
For any $n \times n$ matrix $V$, not necessarily symmetric, we denote its eigenvalues by $\{\mu_i(V)\}_{i=1}^n$, ordered by their real parts in descending orders: $\mathrm{Re}\{\mu_1(V)\} \ge \mathrm{Re}\{\mu_2(V)\} \ge \cdots \ge \mathrm{Re}\{\mu_n(V)\}$.
For a vector $v \in \mathbb{R}^d$, we denote its $\ell_2$ norm as $\|v\|$, and its $\ell_1$ norm as $\|v\|_1$.
For a matrix $A$, we denote its operator norm as $\|A\|$, defined as $\|A\| = \max_{\|v \| =1} \|Av\|$.

\subsection{Non-backtracking matrix fails in imbalanced settings}
Let $m$ be the number of edges in an undirected network.
To construct the non-backtracking matrix, we represent the edge between node $i$ and node $j$ by two directed edges, one from $i$ to $j$ and the other from $j$ to $i$.
The $2m \times 2m$ non-backtracking matrix $B$, indexed by these directed edges, is defined by
\begin{equation}
    B_{i \to j,\, k \to l} = \left\{ \begin{array}{ll}
        1 & \textrm{if } j=k \textrm{ and } i \ne l \\
        0 & \textrm{otherwise.}
    \end{array}
    \right. \label{eq:conventional_nonbacktracking}
\end{equation}
It has been shown, e.g., in \citet{angel2015non, krzakala2013spectral}, that the spectrum of $B$ is the set $\{\pm 1\} \cup \{\mu: \det(\mu^2 I - \mu A + D - I) = 0\}$, or equivalently, the set $\{\pm 1\} \cup \{\textrm{eigenvalues of } H\}$, where
\begin{equation}
    H = \left(\begin{array}{cc}
       A  & \In - D \\
       \In  & \bm{0}
    \end{array} \right). \label{eq:definition_of_H}
\end{equation}
We call this $2n \times 2n$ matrix $H$ the \textit{non-backtracking spectrum operator} for $A$.
Here 
% $\bm{0}$ is the $n \times n$ matrix of all zeros, $\In$ is the $n \times n$ identity matrix, and 
$D = \mathrm{diag}(\bm{d})$ is $n \times n$ diagonal matrix with degrees $d_i$ on its diagonal.
% The eigenvector of $H$ corresponding to eigenvalue $\mu$ is of the form $\left(\begin{array}{c}
%      y \\
%      \mu^{-1} y  
% \end{array}\right)$ for some $y \in \mathbb{C}^n$,
% so that
% \begin{equation*}
%     A y +  \mu ^{-1} (\In - D) y = \mu y
% \end{equation*}
% Therefore, the eigenvalue $\mu$ must satisfy
% \begin{equation*}
%     \det(\mu^{2} \In - \mu A - (\In-D)) = 0,
% \end{equation*}
% or equivalently when $\mu^{-1}$ is a root of the Ihara-Bass formula for the graph zeta function \citep{bass1992ihara, hashimoto1989zeta}.

Since $B$, or equivalently, $H$, is not symmetric, its eigenvalues are generally non-real.
Moreover, they are different from its singular values---while the nontrivial singular values of $B$, $\{d_i - 1\}_{i \in \mathcal{V}}$, are controlled by the node degrees, its eigenvalues are not \citep{krzakala2013spectral}.
This asymmetry is precisely why its spectral properties are superior under sparsity.

From the graph theory perspective, unlike the adjacency matrix $A$, the spectrum of $B$ is not sensitive to high-degree nodes because a walk starting at some node $i$ cannot turn around and return to itself immediately.
It is of interest to study the powers of these matrices, motivated by the method of moment.
In the case of the adjacency matrix, powers are counting walks from one node to another, and these get multiplied around high-degree nodes since the walk can come in and out of such nodes in multiple ways.
Instead, by the construction of the non-backtracking matrix, taking powers forces a directed edge to leave to another directed edge that does not backtrack, preventing such amplifications around high-degree nodes.
So, the non-backtracking gives a way to mitigate the degree-variations and to avoid \textit{localized} eigenvector under sparsity.
% , arguably more efficiently than trimming which removes information from the graph. 

Formally, 
% consider the Erd\H{o}s–R\'{e}nyi model with constant-order degree, $G(n, \frac{d}{n})$ for some $d = O(1)$ and $d>1$. 
under the constant degree regime with $\max_{i} \E[d_i] = O(1)$,
the leading eigenvalues of $A$ (or $\Ac$ below) are dictated by the nodes of the highest degrees, and the corresponding eigenvectors are \textit{localized} around these nodes \citep{benaych2019largest, hiesmayr2023spectral}.
In particular, the adjacency matrix no longer concentrates under this ultra-sparse regime \citep{krivelevich2003largest, le2017concentration}. With high probability,
\begin{equation}
    \| A \| = (1+o(1))\sqrt{d_{\max}} = (1 + o(1)) \sqrt{\frac{\log n}{\log \log n}}.
    % \quad \textrm{if} \quad d = O(1). 
    \label{eq:adjacency_not_concentrate}
\end{equation}
Intuitively, for SBM with constant-order expected degrees, while the amount of signal is held constant, the amount of noise grows as $n$ increases. As a result, the bulk of uninformative eigenvalues will swamp the signal related to community structure, making its existence difficult to detect.

On the other hand, the spectrum of $B$ behaves much better.
It satisfies the \textit{weak Ramanujan property} \citep{bordenave2015non}, which roughly says that the leading eigenvalues of $B$ concentrate around non-zero eigenvalues of $\E[A]$ and the bulk is contained in a circle of radius $\sqrt{\|B\|} \approx \sqrt{\lambda_1(\E[A])}$.
% \begin{equation*}
%     \mu_1(B) = d + o(1) \quad \textrm{ and } \quad 
%     |\mu_k(B)| \le \sqrt{d} + o(1),
% \end{equation*}
% for all non-informative eigenvalues $\mu_k$.
% This property is reflected in the spectrum of the non-backtracking matrix, which has for largest eigenvalue (in magnitude) $\lambda_1$ (which is real positive). 
Then, the signal can be detected as long as the informative eigenvalues are ``visible" and separate from the bulk.
% of non-informative eigenvalues. 
% with radius $\sqrt{\|B\|}$.
It is shown that the method of counting informative eigenvalues of $B$ succeeds all the way down to the \textit{detectability threshold} (defined in e.g., \citet{massoulie2014community, abbe2018community}) in the symmetric case. 
% We will formally define this threshold and discuss more about it in Section~\ref{section:centering_enhances_signal}.

Despite its favorable spectral property under sparsity, there is a notable drawback of the non-backtracking matrix in practice. 
As is noticed in \citet{le2022estimating}, when the communities are unbalanced in terms of size or edge density, it will be significantly more challenging to distinguish informative eigenvalues from the bulk, leading to a severe underestimate of $K$ by counting its leading eigenvalues.
To illustrate, in Figure~\ref{fig:compare_nonbacktracking_normalized_adjacency_under_community_size_imbalance} and Figure~\ref{fig:compare_nonbacktracking_normalized_adjacency_under_P_imbalance},
where we have unequal connecting probabilities in two communities, 
the second largest eigenvalue of the non-backtracking matrix is swamped within the bulk.
However, in the meantime, the \textit{normalized adjacency matrix}, as proposed in \citet{bickel2016hypothesis} and \citet{lei2016goodness}, clearly suggests the existence of an informative eigenvalue despite the imbalance, signaling a second community.
Its effectiveness under imbalance can be attributed to centering, as will be discussed in Section~\ref{section:centering_enhances_signal}. 

We will next take a closer look at the normalized adjacency matrix.

\subsection{Normalized adjacency matrix and Tracy-Widom distribution fail in sparse settings}
\label{section:normalized_adjacency}
Let us introduce the normalized adjacency matrix $\An$ considered in \citet{bickel2016hypothesis} and \citet{lei2016goodness}:
\begin{equation*}
    \An_{ij} = \frac{A_{ij} - P_{ij}}{\sqrt{(n-1) P_{ij} (1-P_{ij})}},\quad i\ne j \quad \textrm{and } \quad 
    \An_{ii} = 0, \forall i. 
    % \label{eq:definition_of_An}
\end{equation*}
If the model $P$ is correctly specified, i.e., $P = \E[A]$, then $\An$ is a \textit{generalized Wigner matrix} \citep{erdos2011universality}, satisfying $\E[\An] = \bm{0}$ and $\sum_j \Var (\An_{ij}) = 1$ for all $i$.
% Its eigenvalues $\lambda_1(\An) \ge ... \ge \lambda_n(\An)$ are all real.
It is a well-known result \citep{erdHos2013spectral,benaych2016lectures} that the empirical distribution of its eigenvalues converges weakly to the \textit{semicircle law} when $n\min_{i,j}\{P_{ij}\}
 \gg \log n$,
\begin{equation*}
    \rho_{sc}(x) = \frac{1}{2\pi} \sqrt{[4 - x^2]_+}.
\end{equation*}
Furthermore, the asymptotic distribution of the extreme eigenvalues of $\An$ has also been well studied.
Based on results in \citet{erdHos2012spectral, lee2014necessary}, when $n\min_{i,j}\{P_{ij}\} \gg n^{2/3}$, we have
\begin{equation}
    n^{2/3}\{\lambda_1(\An) - 2\} \overset{d}{\to} TW_1,
    % \quad \textrm{and} \quad n^{2/3}\{-\lambda_n(\An) - 2\} \overset{d}{\to} TW_1,
    \label{eq:Tracy-Widom_limit_of_lambda_An}
\end{equation}
where $TW_1$ denotes the Tracy–Widom distribution with index 1.
When $A$ is generated from an Erd\H{o}s–R\'{e}nyi model $G(n,p)$ with $P_{ij} = p, \, \forall i \ne j$, 
it is further shown by \citet{lee2018local} that the Tracy-Widom limit still holds when $np \gg n^{1/3}$, except with a deterministic shift of $(np)^{-1}$,
\begin{equation}
    n^{2/3}\{\lambda_1(\An) - (2 + (np)^{-1})\} \overset{d}{\to} TW_1.
    % \quad \textrm{and} \quad n^{2/3}\{-\lambda_n(\An) - (2 + (np)^{-1})\} \overset{d}{\to} TW_1
    \label{eq:Tracy-Widom_limit_with_deterministic_shift}
\end{equation}
This deterministic shift is quite noticeable for finite $n$ that is not large enough.
See Figure~\ref{fig:null_distributions_with_scaling} where $n \lesssim 10^3$.

In practice, $P$ is unknown, so a test statistic cannot be directly constructed from $\An$.
A natural substitute is a plug-in version with $P$ replaced by some estimator $\widehat{P}$.
The empirical normalized adjacency matrix is defined as
\begin{equation*}
    \Ane_{ij} = \frac{A_{ij} - \Phat_{ij}}{\sqrt{(n-1) \Phat_{ij} (1-\Phat_{ij})}},\quad i\ne j \quad \textrm{and } \quad 
    \Ane_{ii} = 0, \forall i.
\end{equation*}
It is proposed by \citet{lei2016goodness} and \citet{bickel2016hypothesis} to use $\lambda_1(\Ane)$ as the test statistic for the goodness-of-fit test for an SBM with $K_0$ communities.

Under the null hypothesis, where we assume $P=\E[A]$ is an SBM with $K_0$ communities, it is shown by \citet{lei2016goodness} that when the community label estimate $\widehat{\bm{g}}$ is consistent and we use a plug-in estimator $\widehat{P}$, the leading eigenvalue of $\Ane$ is asymptotically equivalent to the leading eigenvalue of $\An$. Namely,
\begin{equation*}
    \lambda_1(\Ane) = \lambda_1(\An) + o_{\Prob}(n^{-2/3}), \quad \textrm{under } H_0:K=K_0.
\end{equation*}
Therefore, $\lambda_1(\Ane)$ also has the Tracy-Widom limit like in~\eqref{eq:Tracy-Widom_limit_of_lambda_An}.
This validates the Type-I error control of the Tracy-Widom test based on $\Ane$.
Furthermore, \citet{lei2016goodness} provides asymptotic power guarantee when $H_0$ underestimates $K$.
That is, if the network is generated by the alternative model $\mathrm{SBM}(\bm{g}^{(n)}, Q^{(n)})$ with $K$ blocks (assuming $K$ fixed), while one estimates $P$ by $\widehat{P}$ assuming some $K_0 < K$, then there exists a lower bound for the growth rate of $\lambda_1(\Ane)$,
% \begin{equation}
%     \lambda_1(\Ane) = \Omega(\delta_n n), \quad \textrm{under } H_1:K>K_0, 
%     \label{eq:asymptotic_order_lambda_1_Ane}
% \end{equation}
% where $\delta_n$ is the smallest $\ell_{\infty}$ distance among all pairs of distinct rows of $Q^{(n)}$.
which holds regardless of the structure of $Q^{(n)}$ and the particular method used to estimate the membership $\bm{g}^{(n)}$.

Towards the end of this section, we highlight that the normalized adjacency $\An$ still faces concentration issues \eqref{eq:adjacency_not_concentrate} under sparsity, much like its predecessor, the adjacency matrix $A$.
That means spectral algorithms based on $\An$ also fail significantly above the detectability threshold.
In Figure~\ref{fig:compare_nonbacktracking_normalized_adjacency_under_sparsity}, where the network is generated from a sparse balanced SBM with $K=2$, the largest eigenvalue of $\An$ fails to detect the signal, while the spectrum of the non-backtracking matrix distinctly reveals the presence of the second community.

\section{Novel non-backtracking matrix with centering}
\label{section:proposed_operator:nonbacktracking_with_centering}
Inspired by the benefits brought by centering and non-backtracking, respectively, we consider a variant of the non-backtracking operator based on the following centered adjacency matrix $\Ac$:
\begin{equation*}
    \Ac_{ij} = A_{ij} - P_{ij}, \quad i\ne j, \quad \textrm{and } \quad 
    \Ac_{ii} = 0, \quad \forall i.
\end{equation*}
Define
\begin{equation}
    \Hc = \left(\begin{array}{cc}
       \Ac  & \In - \Dc \\
        \In & \bm{0}
    \end{array} \right),
    \label{eq:definition_of_Hc}
\end{equation}
where $\Dc = \mathrm{diag}(\bm{\dc})$ is an $n\times n$ diagonal matrix with $\bm{\dc} \in \mathbb{R}^n$ and $\dc_i \coloneqq \sum_{j=1}^n \Ac_{ij}$ being its $i$th diagonal entry.
We have $\E[\Dc]=\bm{0}$ when $P_{ij} = \E[A_{ij}], \, \forall i\ne j$.

For practical use, we still rely on an empirical version of $\Hc$ with $P$ replaced by some $\widehat{P}$ that can be calculated from the observed network $A$, 
% where $\widehat{P}$ is an estimator of $P$,
% assuming $K=K_0$ for some pre-specified $K_0$,
\begin{equation*}
    \Hce = \left(\begin{array}{cc}
       \Ace  & \In - \Dce \\
        \In & \bm{0}
    \end{array} \right),
    % \label{eq:definition_of_Hce}
\end{equation*}
with $\Ace_{ij} = A_{ij} - \widehat{P}_{ij}$, $\Dce = \mathrm{diag}(\bm{\dce})$, and $\dce_i \coloneqq \sum_{j=1}^n \Ace_{ij}$ similarly defined.
A natural choice of $\widehat{P}$ will be some estimator of $P$ assuming $K=K_0$ for some pre-specified $K_0$.
We will sometimes denote $\Hce$ by $\Hce_{(K_0)}$ when it is necessary to explicitly show that the $\widehat{P}$ in $\Hce$ is estimated assuming $K=K_0$.

We propose to use the leading eigenvalue of $\Hce$ as the test statistic for a goodness-of-fit test of an SBM with $K_0$ communities.
When testing $H_0: K = K_0$ vs. $H_1: K>K_0$, 
% we consider two choices: $\mu_1(\Hce_{(K_0)})$ and $\mu_{K_0}(\Hce_{(1)})$, with intuitions as follows.
we consider the following decision rule:
\begin{equation*}
    \textrm{Reject } H_0 \textrm{ if } \mu_1(\Hce_{(K_0)}) > t_{n, K_0},
    % \label{eq:decision_rule}
\end{equation*}
where $t_{n,K_0}$ is some threshold to be determined.
The intuition is as follows.
Under the null hypothesis $H_0:K = K_0$, if the estimator $\widehat{P}$ closely approximates $P$, then $\Ace$ is a ``noise" matrix with almost mean-zero entries, and the largest eigenvalue of $\Ace$ or $\Hce$ should be on the edge of the bulk.
% Meanwhile, if $\widehat{P}$ is estimated under $K=1$, reducing to an edge density estimator, then $\Ace$ is a matrix whose signal (expectation) has rank $(K_0 - 1)$. 
% In this case, we also expect the $K_0$-th largest eigenvalue of $\Ace$ or $\Hce$ to be on the edge of the bulk.
On the other hand, under the alternative hypothesis $H_1:K>K_0$, we expect the corresponding eigenvalue to become informative and therefore separated from the bulk. 

We provide theoretical justification for our proposed operator for the rest of this section.
Our analysis will focus on the simple task of testing $H_0: K = 1$ vs. $H_1: K>1$,
which essentially boils down to a goodness-of-fit test for an Erd\H{o}s–R\'{e}nyi model $G(n,p)$ with no community structure.
As discussed later in Section~\ref{section:sequential_or_recursive_testing}, one way to estimate $K$ is by recursively applying this simple test and bi-partitioning the network \citep{li2022hierarchical}. 
% Since $\mu_1(\Hce_{(1)})$ becomes the unique choice for the test statistic in this case, we omit the subscript.
Under $H_0: K = 1$, both $P$ and $\widehat{P}$ reduce to a matrix with off-diagonal entries being the population and empirical edge density, respectively, 
\begin{equation}
    \Ac_{ij} = A_{ij} - p,\quad \textrm{ and } \quad \Ace_{ij} = A_{ij} - \widehat{p}, \, \quad \forall i \ne j, \label{eq:definition_of_Ac_and_Ace}
\end{equation}
where $\widehat{p} = \frac{1}{n(n-1)}\sum_{i,j} A_{ij}$ estimates $p$ very accurately as $|\widehat{p} - p| = O_{\Prob}(n^{-1}p^{1/2})$.

Here, we briefly review some related work on goodness-of-fit tests of the Erd\H{o}s–R\'{e}nyi model against SBM alternatives, including but not limited to the following studies.
\citet{wang2017likelihood, hu2020corrected} proposed likelihood-based approaches and BIC-type criteria. 
As already introduced in Section~\ref{section:normalized_adjacency}, \citet{bickel2016hypothesis} and \citet{lei2016goodness} considered the spectral approach that uses the leading eigenvalue of $\Ane$. 
\citet{gao2017testing1, gao2017testing2} proposed a test based on frequencies of three-node subgraphs. 
\citet{banerjee2017optimal} introduced a linear spectral statistic to test $H_0: K = 1$ vs. $H_1 : K = 2$ under the SBM.
Many of these methods can also be extended to DCSBM alternatives. 
However, for simplicity, we will focus on SBM alternatives in this study and leave its potential as a goodness-of-fit test for DCSBM for future work.
% We will compare our proposed test with these methods in Section~\ref{section:numerical_experiment} through numerical experiments.

\subsection{Null distribution under $H_0: K=1$}
Let the Erd\H{o}s–R\'{e}nyi model $G(n, p)$ be the null model with no community structure.
Inspired by the idea in \citet{wang2023limiting}, 
we consider the rescaled conjugation of $\Hc$ which has the same eigenvalue with $\Hc/\sqrt{\alpha}$, where $\alpha = (n-1)p + 1$.
Note that $\Dc = D - (n-1)p \In = D - (\alpha - 1)\In$.
Define
\begin{eqnarray*}
    \widetilde{H} & \coloneqq& \frac{1}{\sqrt{\alpha}} 
    \left(\begin{array}{cc}
       \frac{1}{\sqrt{\alpha}} \In  &  \bm{0} \\
       \bm{0}  &  \In
    \end{array} \right)
    \left(\begin{array}{cc}
       \Ac  &  \In - \Dc \\
       \In  &  \bm{0}
    \end{array} \right)
    \left(\begin{array}{cc}
       \sqrt{\alpha}  \In  &  \bm{0} \\
       \bm{0}  &  \In
    \end{array} \right) \nonumber\\
    & = &\left(\begin{array}{cc}
       \frac{1}{\sqrt{\alpha}}  \Ac  &   \alpha^{-1} (\In - \Dc) \\
       \In  &  \bm{0}
    \end{array} \right) = \left(\begin{array}{cc}
       \frac{1}{\sqrt{\alpha}}  \Ac  &  \In - \alpha^{-1} D \\
       \In  &  \bm{0}
    \end{array} \right) . 
    % \label{eq:definition_of_tildeH}
\end{eqnarray*}
We have the decomposition 
\begin{equation*}
    \widetilde{H} = \widetilde{H}_0 + E,
\end{equation*}
where
\begin{equation*}
    \widetilde{H}_0 = \left(\begin{array}{cc}
        \frac{1}{\sqrt{\alpha}}  \Ac  &  \bm{0} \\
       \In  &  \bm{0}
    \end{array}\right)
\quad \textrm{ and } \quad 
    E = \left(\begin{array}{cc}
        \bm{0} & \alpha^{-1} (\In - \Dc) \\
        \bm{0} & \bm{0}
    \end{array}\right) =  \left(\begin{array}{cc}
       \bm{0}  &  \In - \alpha^{-1} D \\
       \bm{0}   &  \bm{0}
    \end{array} \right).
\end{equation*}
The nonzero eigenvalues of $\widetilde{H}_0$ are simply given by the eigenvalues of $\frac{1}{\sqrt{\alpha}} \Ac$.

\subsubsection{Dense regime}
When the network is sufficiently dense, note that $\E [\alpha^{-1} (\In - \Dc)] = \alpha^{-1} \In$, and $\Var (\alpha^{-1} (\In - \Dc)) = \alpha^{-2} \Var(D) \approx \alpha^{-1} \In$. 
That means when $\alpha \to \infty$ fast enough, the matrix $E$ becomes negligible and can be treated as a perturbation.
% Namely,
% \begin{equation*}
%     E = 
%      \left(\begin{array}{cc}
%        \bm{0}  &  \alpha^{-1}(\In - \Dc)\\
%        \bm{0}  &  \bm{0}
%     \end{array} \right) \xrightarrow{\Prob} \bm{0},
%     % \quad \textrm{if } \alpha \to \infty,
% \end{equation*}
In particular, it can be shown that $\|E\| = o_{\Prob}(1)$ when $\alpha /\log n \to \infty$.
As a result, the leading eigenvalue of $\widetilde{H}$ converges to that of $\frac{1}{\sqrt{\alpha}}\Ac$.
Moreover, when the average degree grows faster than $n^{2/3}$, 
we anticipate that $\mu_1(\widetilde{H})$ converges to $\lambda_1(\frac{1}{\sqrt{\alpha}}\Ac)$ fast enough so that it also has the Tracy-Widom limit.
However, demonstrating this stronger assertion relies on the following conjecture.
\begin{conjecture}[Concentration of $\vc_1^\top \Dc \ \vc_1$]
\label{conjecture:concentration_of_v1Dv1}
    Suppose a network $A$ is generated from an Erd\H{o}s-R\'{e}nyi model $G(n, p)$.
    Denote $\vc_1$ as the eigenvector of $\Ac$ corresponding to the largest eigenvalue.
    % Assume $\alpha = (n-1)p + 1 = \Omega (n^{\delta})$ for some $\delta \in (0,\frac{2}{3}]$.
    Assume $\alpha = (n-1)p + 1 = \Omega (n^{2/3})$. %for some $\delta \in (0,\frac{2}{3}]$.
    Then
    \begin{equation}
        |\vc_1^\top \Dc \ \vc_1| = o_{\Prob}(n^{\epsilon}), \label{eq:v1Dv1_is_op_n^eps}
    \end{equation}
    for some arbitrarily small $\epsilon > 0$.
\end{conjecture}
The proof of this conjecture turns out to be highly nontrivial due to the challenge that $\Dc$ is dependent on $\Ac$. Recent developments in random matrix theory, particularly in the study of \textit{Quantum Unique Ergodicity} (QUE) for Wigner matrices \citep{cipolloni2021eigenstate, adhikari2023eigenstate}, show that $v_i^\top B v_i$ converges to $\frac{1}{n} \mathrm{tr}(B)$ with high probability for any \textit{deterministic} matrix $B$. 
However, extending these results to our specific scenario is highly technical and beyond the scope of this paper.
% currently beyond our capabilities.
% we leave it for future works.

\begin{figure}[ht]
\centering
\begin{minipage}{0.49\textwidth}
\centering
\includegraphics[width=1.0\textwidth]{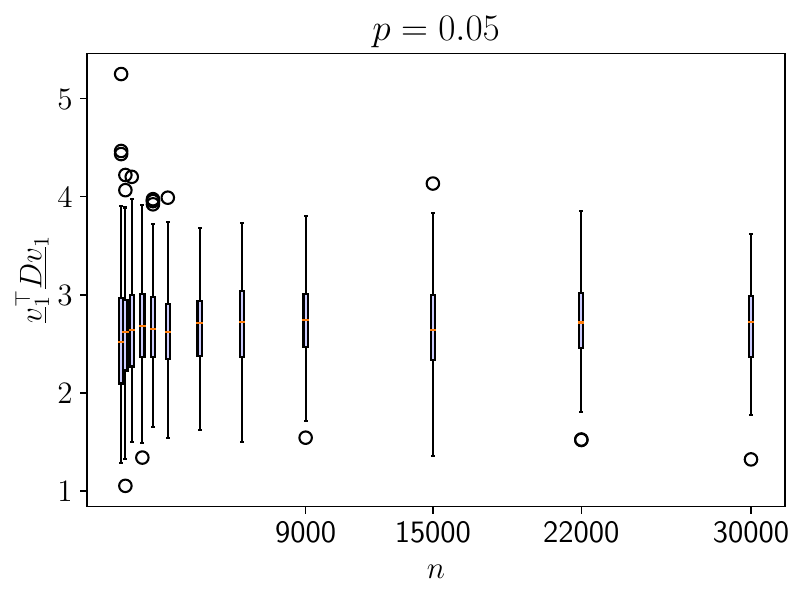}
\end{minipage}
\begin{minipage}{0.49\textwidth}
\centering
\includegraphics[width=1.0\textwidth]{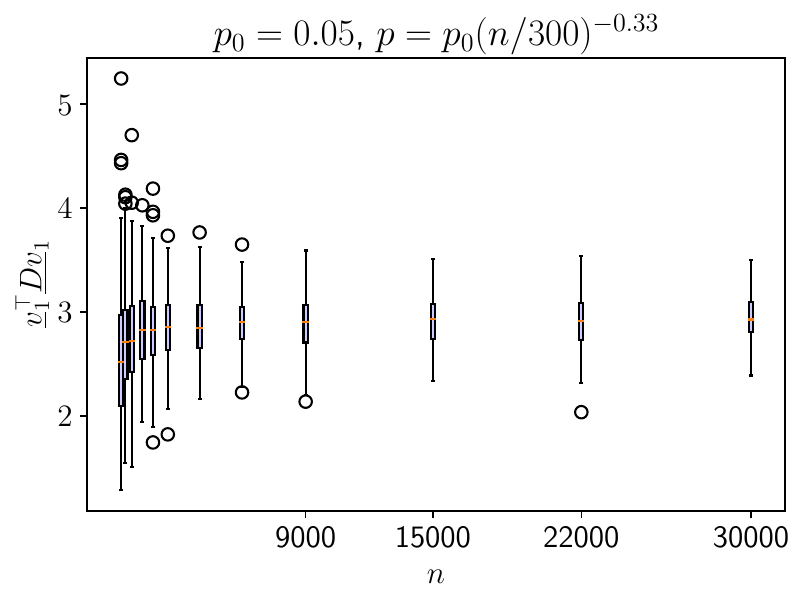}
\end{minipage}
\begin{minipage}{0.49\textwidth}
\centering
\includegraphics[width=1.0\textwidth]{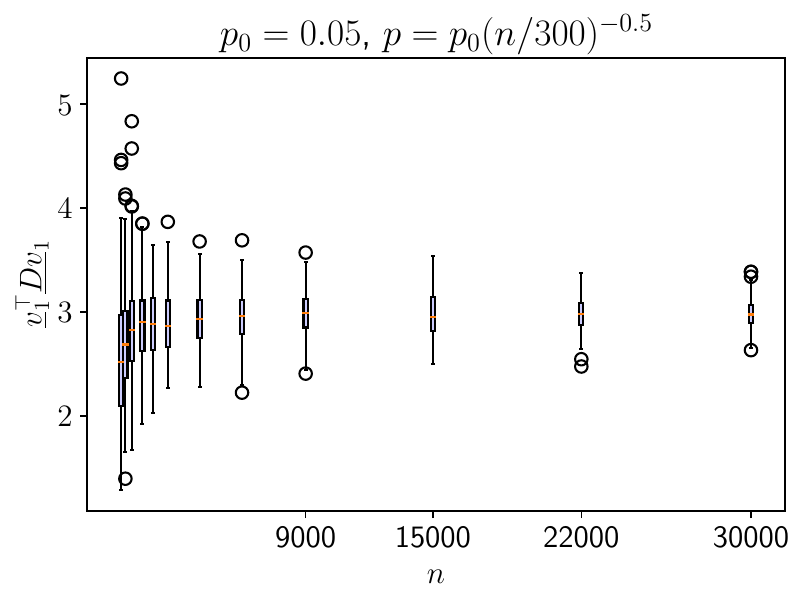}
\end{minipage}
\begin{minipage}{0.49\textwidth}
\centering
\includegraphics[width=1.0\textwidth]{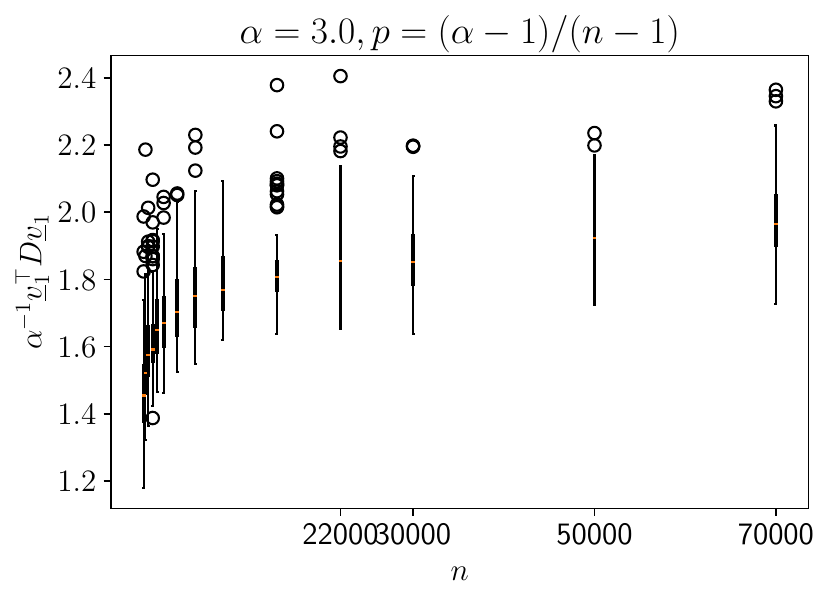}
\end{minipage}
\caption{Growth rate of $\vc_1^\top \Dc \ \vc_1$ under different $p(n)$. Here we take $p\asymp 1$, $p\asymp n^{-1/3}$, $p\asymp n^{-1/2}$ and $p\asymp n^{-1}$, respectively.
All of them appear to align with the bound \eqref{eq:v1Dv1_is_op_n^eps} in Conjecture~\ref{conjecture:concentration_of_v1Dv1}.}
\label{fig:growth_rate_of_v1Dv1}
\end{figure}

Nevertheless, numerical simulations in Figure~\ref{fig:growth_rate_of_v1Dv1} strongly support the claim of Conjecture~\ref{conjecture:concentration_of_v1Dv1}.
Moreover, the condition $\alpha=\Omega(n^{2/3})$ may be stronger than necessary.
Across different values of $\gamma$ for which $\alpha\asymp n^{\gamma}$,  the growth rate of $|\vc_1^\top \Dc \ \vc_1|$ consistently appears to be bounded by \eqref{eq:v1Dv1_is_op_n^eps},
even including sparser cases with $\alpha \prec n^{2/3}$.
In Section~\ref{section:constant_degree_regime_null_distribution}, we will show that when $\alpha = O(1)$, i.e., under the constant degree regime, we have $|\vc_1^\top \Dc \ \vc_1| \asymp |\vc_1^\top D \vc_1| \asymp \frac{\log n}{\log \log n} $.

The following proposition formally describes the convergence of $\mu_1(\Hc)$ to $\lambda_1(\Ac)$.
\begin{proposition}[Tracy-Widom limit of $\mu_1(\Hc)$]
\label{proposition:convergence_of_H_to_TW1}
    Suppose a network $A$ is generated from an Erd\H{o}s-R\'{e}nyi model $G(n, p)$.
    Assume $\alpha / \log n \to \infty$.
    We have
    \begin{equation}
         \alpha^{-1/2} |\mu_1(\Hc)- \lambda_1(\Ac)| = |\mu_1(\widetilde{H})- \mu_1(\widetilde{H}_0) | 
          = O_{\Prob}(\sqrt{\log n} \alpha^{-1/2})
          =o_{\Prob}(1),
        \label{eq:difference_between_mu_1_H_and_lambda_1_A_general}
    \end{equation}
    and
    \begin{equation}
         \widehat{\alpha}^{-1/2} |\mu_1(\Hce)- \lambda_1(\Ace)|
          = O_{\Prob}(\sqrt{\log n} \alpha^{-1/2})
          =o_{\Prob}(1).
          \label{eq:difference_between_mu_1_Hce_and_lambda_1_Ace_general}
    \end{equation}
    Moreover, if we assume $\alpha = \Omega (n^{2/3 + \epsilon})$ for some $\epsilon>0$, 
    Conjecture~\ref{conjecture:concentration_of_v1Dv1} implies that
    \begin{equation}
        \alpha^{-1/2} |\mu_1(\Hc)- \lambda_1(\Ac) | = o_{\Prob}(n^{-2/3}),
        \label{eq:difference_between_mu_1_H_and_lambda_1_A}
    \end{equation}
    and
    \begin{equation}
        \widehat{\alpha}^{-1/2} |\mu_1(\Hce)- \lambda_1(\Ace) | = o_{\Prob}(n^{-2/3}).
         \label{eq:difference_between_mu_1_Hce_and_lambda_1_Ace}
    \end{equation}
    % and 
    % \begin{equation}
    %     \alpha^{-1/2} \left(\mu_1(\widehat{\Hc})- \lambda_1(\Ace)\right)  = o_{\Prob}(n^{-2/3}).
    %     \label{eq:difference_between_mu_1_H_hat_and_lambda_1_A_hat}
    % \end{equation}
    % The same convergence rates hold if we replace $\Hc$, $\Ac$ and $p$ in \eqref{eq:difference_between_mu_1_H_and_lambda_1_A_general} and \eqref{eq:difference_between_mu_1_H_and_lambda_1_A} by their ``hat" versions \textcolor{red}{please list them instead of saying hat versions}.
\end{proposition}
As a result, when $\alpha /\log n \to \infty$, \eqref{eq:difference_between_mu_1_H_and_lambda_1_A_general} implies
\begin{equation*}
    \frac{\mu_1(\Hc)} {\sqrt{(n-1) p (1 - p)}} 
    \xrightarrow{\Prob} \frac{\lambda_1(\Ac)} {\sqrt{(n-1) p (1 - p)}}
    = 2+o_{\Prob}(1).
\end{equation*}
Moreover, when $\alpha = \Omega (n^{2/3 + \epsilon})$, \eqref{eq:difference_between_mu_1_H_and_lambda_1_A} implies that the convergence of $\mu_1(\Hc)$ to $\lambda_1(\Ac)$ is quick enough so that it also has the Tracy-Widom limit.
\begin{equation*}
     n^{2/3} \left(\frac{\mu_1(\Hc)} {\sqrt{(n-1) p (1 - p)}}-2\right) \xrightarrow{\Prob} n^{2/3} \left(\frac{\lambda_1(\Ac)} {\sqrt{(n-1) p (1 - p)}}-2 \right) \xrightarrow{d} TW_1.
\end{equation*}
Similarly, \eqref{eq:difference_between_mu_1_Hce_and_lambda_1_Ace_general} and \eqref{eq:difference_between_mu_1_Hce_and_lambda_1_Ace} imply that the same convergence holds for $\mu_1(\Hce)$ and $\lambda_1(\Ace)$.
% That means, when testing $H_0:K=1$ under the dense regime, the test statistic $\mu_1(\Hce)$ is asymptotically equivalent to the test statistic $\lambda_1(\Ace)$ proposed by \cite{bickel2016hypothesis, lei2016goodness}.

The proof of Proposition~\ref{proposition:convergence_of_H_to_TW1} is given in Appendix~\ref{section:proof_of_proposition_tracy_widom_limit} via classical perturbation theory.
% along with the fact that $v_1$ is delocalized \cite{erdHos2013spectral, he2019local}.
The main idea is to consider an approximation for $\mu_1(\widetilde{H})$, which can be decomposed into the sum of $\mu_1(\widetilde{H}_0)$ and a correction term. 
Below, we offer some insights into this approximation.

\begin{figure}
\centering
\begin{minipage}{0.49\textwidth}
\centering
\includegraphics[width=1.0\textwidth]{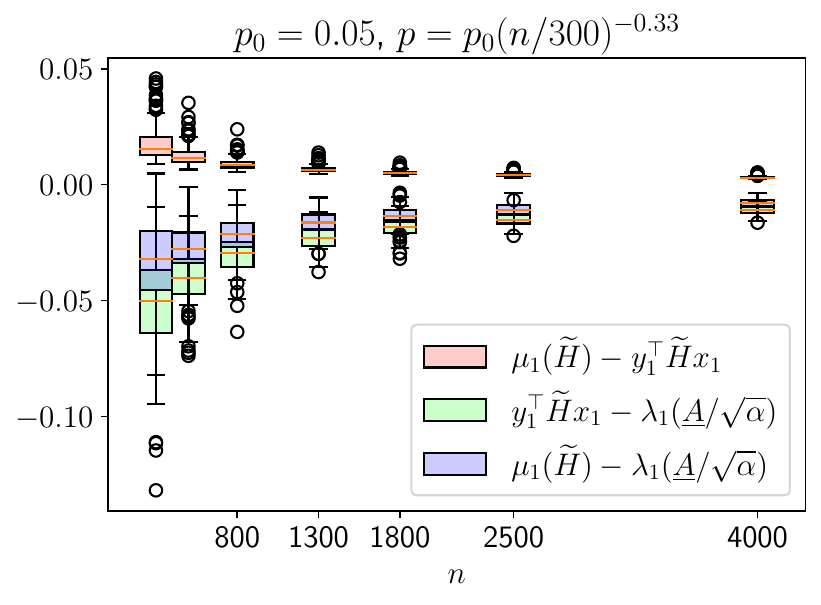}
\end{minipage}
\begin{minipage}{0.49\textwidth}
\centering
\includegraphics[width=1.0\textwidth]{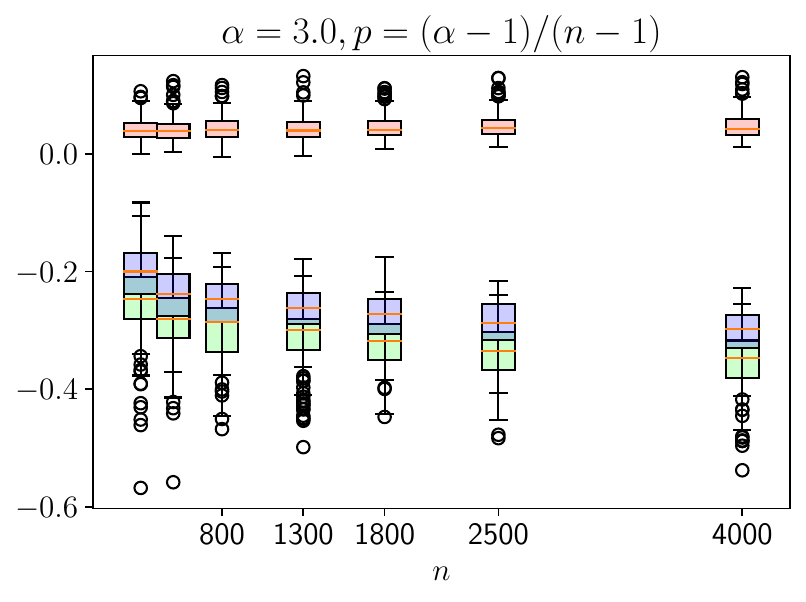}
\end{minipage}
\caption{Asymptotic order of the pairwise differences between $\mu_1(\widetilde{H})$, $y_1^{\top} \widetilde{H} x_1$, and $\lambda_1(\frac{\Ac}{\sqrt{\alpha}})$. 
We take $p\asymp n^{-1/3}$ in the left panel representing the denser regime, and $p\asymp n^{-1}$ in the right panel representing the ultra-sparse regime.
In both regimes, $y_1^{\top} \widetilde{H} x_1$ closely approximates $\mu_1(\widetilde{H})$.
The main contributor to the difference between $\mu_1(\widetilde{H})$ and $\lambda_1(\frac{\Ac}{\sqrt{\alpha}})$ comes from $y_1^{\top} \widetilde{H} x_1 - \lambda_1(\frac{\Ac}{\sqrt{\alpha}})$.
}
\label{fig:order_of_difference_among_approximations}
\end{figure}

It can be verified 
% from \eqref{eq:left_and_right_eigenvectors_of_H0}
that $\mu_1(\widetilde{H}_0)= \lambda_1(\frac{\Ac}{\sqrt{\alpha}}) = y_1^\top \widetilde{H}_0  x_1$, with left eigenvector $y_1$ and right eigenvector $x_1$ of $\widetilde{H}_0$ given by
\begin{equation*}
    y_1^\top = (\vc_1^\top,\, \bm{0}), \quad x_1 = \left( \begin{array}{c}
        \vc_1  \\
        \lambdac_1^{-1} \vc_1
    \end{array} \right), %\label{eq:left_and_right_eigenvectors_of_H0}
\end{equation*}
where $\lambdac_1$ denotes $\lambda_1(\frac{\Ac}{\sqrt{\alpha}})$, and $\vc_1$ is the corresponding eigenvector.
If we use $y_1^\top \widetilde{H}  x_1$ to approximate $\mu_1(\widetilde{H})$ (this approximation is empirically supported in Figure~\ref{fig:order_of_difference_among_approximations}, and rigorously justified by Theorem~\ref{theorem:eigenvalue_perturbation} in Appendix~\ref{section:proof_of_proposition_tracy_widom_limit} when $\|E\| = o_{\Prob}(1)$), then
\begin{eqnarray}
    y_1^\top \widetilde{H}  x_1 & = & (\vc_1^\top,\, \bm{0}) \left(\begin{array}{cc} \frac{1}{\sqrt{\alpha}}\Ac & \In - \alpha^{-1} D \\ \In & \bm{0}\end{array}\right) \left( \begin{array}{c}
        \vc_1  \\
        \lambdac_1^{-1} \vc_1
    \end{array} \right) \nonumber\\
    % & = & \frac{1}{1 +\mu_{2}^2} (v_1^\top, -\mu_2 v_1^\top) \left(\begin{array}{c} \widetilde{\lambda}_1 v_1 - \mu_2 v_1 + \mu_2 \Dc v_1 \\ v_1 \end{array}\right) \\
    & = & \lambdac_1 + \lambdac_1^{-1}(1 - \alpha^{-1} \vc_1^\top D \vc_1) \label{eq:y1_H_x1_form_1}\\
    & = & \lambdac_1 + \alpha^{-1}\lambdac_1^{-1}(1 - \vc_1^\top \Dc \ \vc_1). \label{eq:y1_H_x1_form_2}
\end{eqnarray}
It is obvious that this approximation converges to $\lambdac_1$ if $\alpha$ grows fast enough.
In particular, when $\frac{\alpha}{\log n} \to \infty$, 
the largest eigenvalue of $\Ac$ is concentrated \citep{benaych2020spectral}, 
$$\lambda_1\left(\frac{\Ac}{\sqrt{\alpha}}\right)  = 2 +  o_{\Prob}(1).$$
Also, since the graph is almost regular \citep{le2017concentration}, by Lemma 3.5 in \citet{wang2023limiting}, we have
$$|\alpha^{-1} \vc_1^\top \Dc \ \vc_1| \le \alpha^{-1}\max_i|d_i - (n-1)p| = O_{\Prob}(\alpha^{-1/2} \sqrt{\log n}) = o_{\Prob}(1),$$
which implies that the second term in \eqref{eq:y1_H_x1_form_1} goes to zero in probability.
Namely, $\mu_1(\widetilde{H}) - \lambda_1(\frac{\Ac}{\sqrt{\alpha}}) \approx  y_1^\top \widetilde{H}  x_1 - \lambda_1(\frac{\Ac}{\sqrt{\alpha}}) = o_{\Prob}(1)$. 
Nevertheless, it does require an upper bound for $\vc_1^\top \Dc \ \vc_1$ as tight as in Conjecture~\ref{conjecture:concentration_of_v1Dv1} in order to show that $\mu_1(\widetilde{H})$ converges to $\lambda_1(\frac{\Ac}{\sqrt{\alpha}})$ fast enough for it to have a Tracy-Widom limit.

\subsubsection{Ultra-sparse regime with constant degree}
\label{section:constant_degree_regime_null_distribution}
Next, we turn our focus to the ultra-sparse regime, with constant order average degree.
In this regime, $\|E\|$ does not vanish, thereby invalidating previous theoretical justification for $\mu_1(\widetilde{H}) \approx y_1^{\top} \widetilde{H} x_1$.
Nevertheless, simulation suggests that $y_1^{\top} \widetilde{H} x_1$ continues to closely approximate $\mu_1(\widetilde{H})$, as depicted in the right panel of Figure~\ref{fig:order_of_difference_among_approximations}.
The approximation error $\mu_1(\widetilde{H}) - y_1^{\top} \widetilde{H} x_1$ appears to be negligible, contributing minimally to the difference between $\mu_1(\widetilde{H})$ and $\lambda_1(\frac{\Ac}{\sqrt{\alpha}})$. The primary focus in this section will be on the major contributor to this difference, namely $y_1^{\top}\widetilde{H} x_1 - \lambda_1(\frac{\Ac}{\sqrt{\alpha}})$.

For the approximation $y_1^\top \widetilde{H}  x_1$, the second term in \eqref{eq:y1_H_x1_form_1} or \eqref{eq:y1_H_x1_form_2} no longer vanishes when the graph is ultra-sparse, under the constant degree regime where $np \to d >1$.
The semicircle law no longer holds, and $\lambda_1(\Ac)$ is no longer concentrated.
Instead, it is dominated by the largest degree in the graph \citep{krivelevich2003largest, benaych2019largest}:
$$\lambda_1(\Ac) \asymp \sqrt{d_{\max}} \asymp \sqrt{\frac{\log n}{\log \log n}}.$$
Moreover, its corresponding eigenvector is localized around the highest degree node \citep{benaych2019largest, hiesmayr2023spectral}.
% As a result, we have $\alpha^{-1} v_1^\top D v_1 \gg 1$.
% matching the order of the diverging $\lambda_1$.
We have the upper bound $ \vc_1^\top D \vc_1 \le d_{\max} = O\left( \frac{\log n}{\log \log n}\right)$. 
Together with the lower bound in the next proposition,
we have $\vc_1^\top D \vc_1 = \Theta \left(\frac{\log n}{\log \log n} \right)\gg 1$.
This very slow but diverging growth rate aligns with our observation in the bottom right panel in Figure~\ref{fig:growth_rate_of_v1Dv1}. 
\begin{proposition}[Partial cancellation]
\label{proposition:order_of_vDv_constant_degree}
    Suppose network $A$ is generated from an Erd\H{o}s-R\'{e}nyi model $G(n, \frac{d}{n})$, namely with expected degree held constant when $n\to \infty$.
    Then
    \begin{equation*}
        % \underline{v}_1^\top D \underline{v}_1 \ge \frac{\log n}{ \log \log n}(1 - o_{\Prob}(1)),
        \vc_1^\top D \vc_1 \ge  \frac{\log n}{ 2\log \log n}(1 + o_{\Prob}(1)),
    \end{equation*}
    % where $\underline{v}_1$ is the eigenvector of $\Ac = A - \mathbb{E}[A]$ corresponding to its largest eigenvalue.
    where $\vc_1$ is the eigenvector of $\Ac$ corresponding to its largest eigenvalue.
\end{proposition}
As a result, the correction term is negative when $n\to \infty$, with a diverging magnitude: 
\begin{equation*}
    \left\vert\lambdac_1^{-1}(1- \alpha^{-1} \vc_1 ^\top D \vc_1) \right \vert \ge \frac{1}{2\sqrt{\alpha}}\sqrt{\frac{\log n}{\log \log n}}(1+o_{\Prob}(1)),
\end{equation*}
matching the divergence rate of $\lambda_1 (\frac{\Ac}{\sqrt{\alpha}}) \asymp \frac{1}{\sqrt{\alpha}} \sqrt{\frac{\log n}{\log \log n}}$.
This is illustrated by the green boxes on the right panel of Figure~\ref{fig:order_of_difference_among_approximations}.
Thus, the correction term effectively reduces about half of the noise induced by high degree nodes.

The proof of Proposition \ref{proposition:order_of_vDv_constant_degree} is given in Appendix~\ref{section:proof_of_proposition_order_of_vDv_constant_degree}, which mainly uses the fact that $\vc_1$ is highly localized around some node with highest degree \citep{hiesmayr2023spectral}.

\begin{figure}
\centering
\begin{minipage}{0.49\textwidth}
\centering
\includegraphics[width=1.0\textwidth]{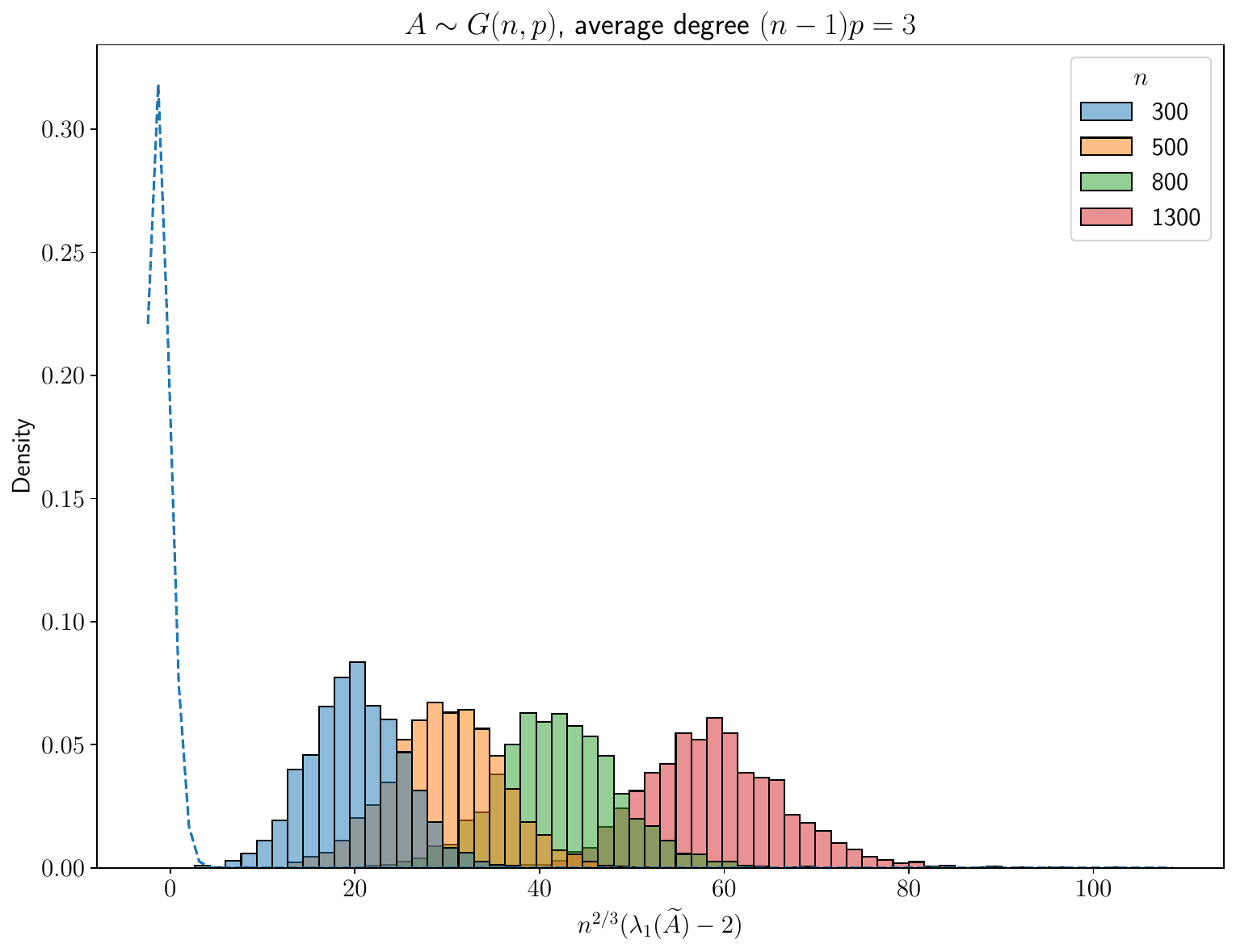}
\end{minipage}
\begin{minipage}{0.49\textwidth}
\centering
\includegraphics[width=1.0\textwidth]{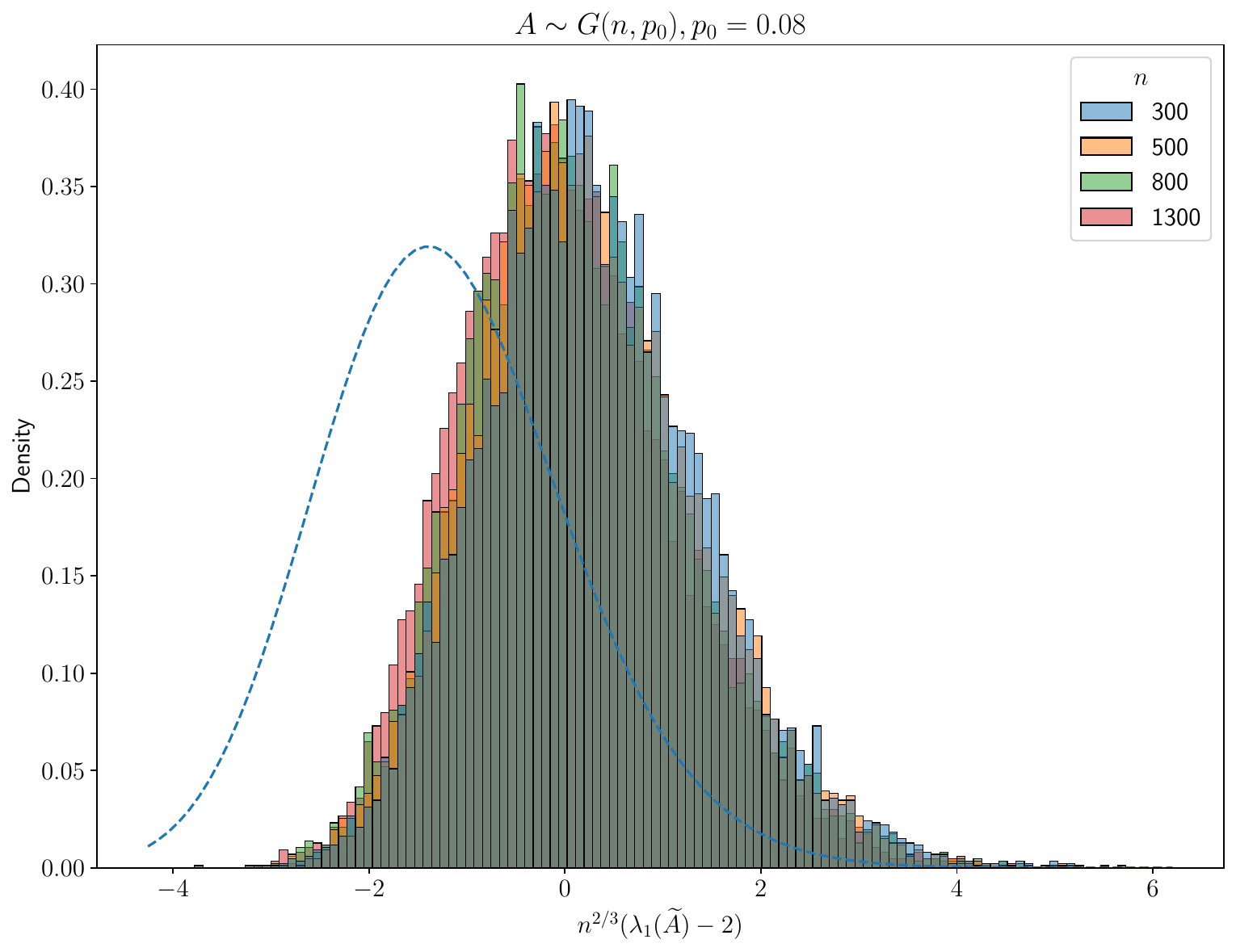}
\end{minipage}
\begin{minipage}{0.49\textwidth}
\centering
\includegraphics[width=1.0\textwidth]{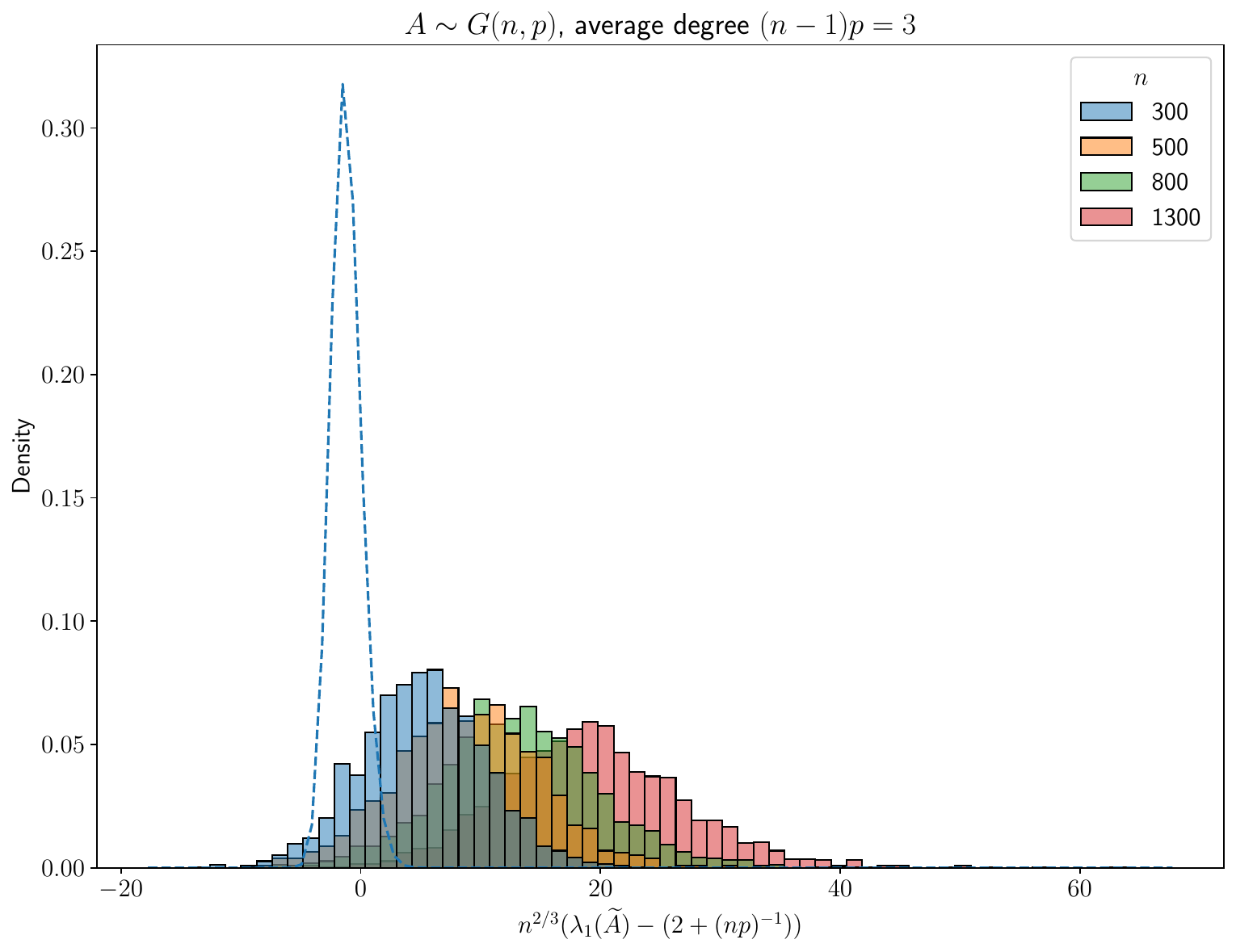}
\end{minipage}
\begin{minipage}{0.49\textwidth}
\centering
\includegraphics[width=1.0\textwidth]{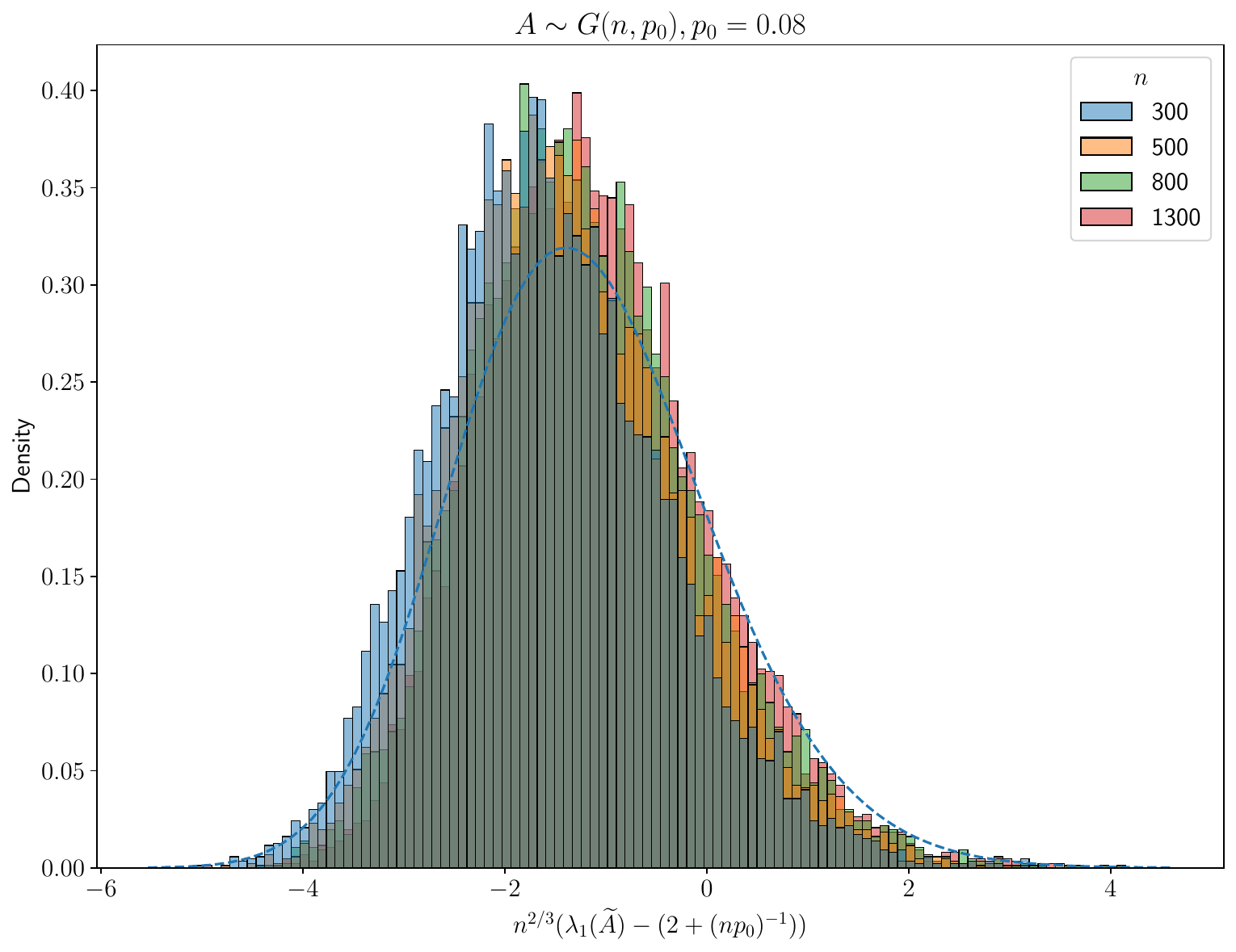}
\end{minipage}
\begin{minipage}{0.49\textwidth}
\centering
\includegraphics[width=1.0\textwidth]{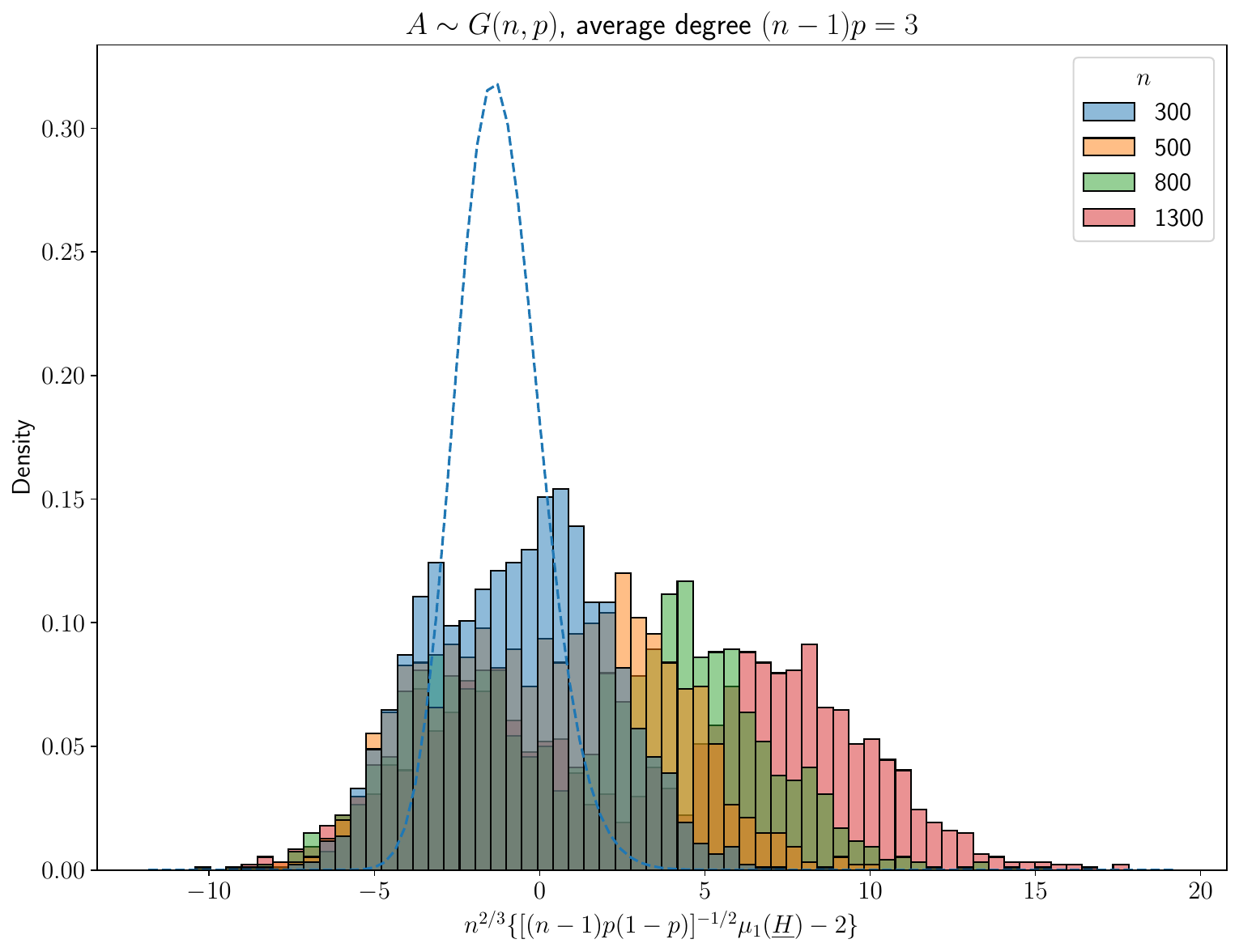}
\end{minipage}
\begin{minipage}{0.49\textwidth}
\centering
\includegraphics[width=1.0\textwidth]{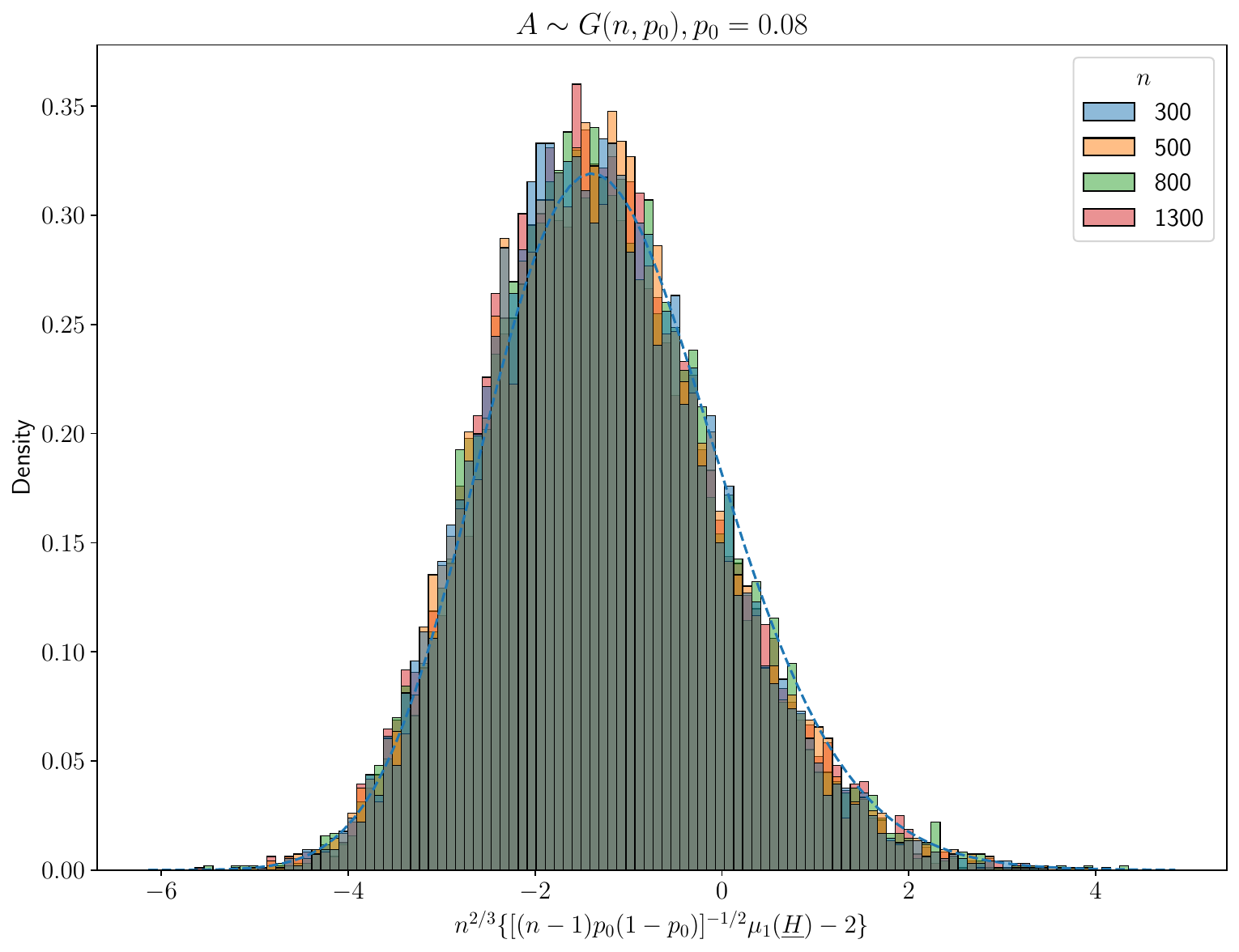}
\end{minipage}
\caption{Null distributions of $\lambda_1(\An)$ and $\mu_1(\Hc)$ with $n^{2/3}$ scaling.
We fix the average degree at 3 in the left panel, representing the ultra-sparse regime, and fix $p_0 = 0.08$ in the right panel.
We let $n$ take various values, indicated by different colors.
The blue dashed line represents the Tracy-Widom distribution with index 1.}
\label{fig:null_distributions_with_scaling}
\end{figure}

\subsubsection{Empirical evidence for the null distribution}
In Figure~\ref{fig:null_distributions_with_scaling}, we present empirical evidence that supports the theoretical results discussed in previous sections. Under the null distribution, we generate $A$ from $G(n, p)$. We simulate the distribution of $\lambda_1(\An) = [(n-1)p(1-p)]^{-1/2} \lambda_1(\Ac)$ and $[(n-1)p(1-p)]^{-1/2} \mu_1(\Hc)$ as $n$ grows, considering two sparsity levels. For the sparse regime, we maintain the average degree fixed at $(n-1)p=3$, while for the dense regime, we keep the density fixed at $p=p_0=0.08$. We vary $n$ across different values in each scenario.

The left panel of Figure~\ref{fig:null_distributions_with_scaling} demonstrates the ultra-sparse regime.
When scaled by $n^{2/3}$, $\lambda_1(\An)$ clearly deviates from the Tracy-Widom limit, with a diverging positive bias. 
This deviation diminishes somewhat after including the deterministic shift $(np)^{-1}$~\eqref{eq:Tracy-Widom_limit_with_deterministic_shift}.  
On the other hand, the distribution of $\mu_1(\Hc)$, after $n^{2/3}$ scaling, also displays a diverging bias against the Tracy-Widom limit, although much more moderate compared to the case of $\lambda_1(\An)$ with or without correction, thanks to the partial cancellation effect outlined in Section~\ref{section:constant_degree_regime_null_distribution}.

% The empirical distribution of $\lambda_1(\An)$ shows significant deviation from the Tracy-Widom distribution, but their difference diminishes considerably after incorporating the deterministic shift term $(np_0)^{-1}$ mentioned in \eqref{eq:Tracy-Widom_limit_with_deterministic_shift} \cite{lee2018local}.
% On the other hand, the empirical distribution of $\mu_1(\Hc)$, after scaling, also displays a positive bias against the Tracy-Widom limit, although much more moderate compared to the case of $\lambda_1(\An)$, thanks to the partial cancellation effect outlined in Proposition \ref{proposition:order_of_vDv_constant_degree}.

The advantage of $\mu_1(\Hc)$ is also evident in the right panel of Figure~\ref{fig:null_distributions_with_scaling}.
In this dense regime, the finite-sample distribution of $\mu_1(\Hc)$ appears nearly perfectly aligned with the Tracy-Widom limit. In contrast, the distribution of $\lambda_1(\An)$, if without the deterministic shift, still exhibits a noticeable positive bias against the Tracy-Widom limit even for $n$ as large as $1300$.

Overall, Figure~\ref{fig:null_distributions_with_scaling} suggests that, for finite $n$, the null distribution of $\mu_1(\Hc)$ is more accurately approximated by the Tracy-Widom distribution compared to how well the null distribution of $\lambda_1(\An)$ is approximated by the Tracy-Widom distribution, which brings convenience to its practical usage in hypothesis testing.

\subsection{Asymptotic power under $H_1: K>1$}
\label{section:asymptotic_power_guarantee}
In this section, we provide an asymptotic lower bound for the proposed spectral statistic $\mu_1(\Hce)$ under the alternative hypothesis $H_1: K > 1$, assuming the network is not too sparse.
This will imply that the null and alternative distributions of $\mu_1(\Hce)$ are well separated.
The main idea here is that the asymptotic order of $\mu_1(\Hce)$ will be comparable to that of $\lambda_1(\Ace)$, the latter of which has been studied in \citet{lei2016goodness}.
For the special case where we test against $H_0: K=1$, we extend their Theorem 3.3 to accommodate a general sparsity level $\alpha$, provided that $\alpha \gg \log n$.

Here, the corresponding null model is the Erd\H{o}s-R\'{e}nyi model $G(n,p)$ where $p = \frac{1}{n(n-1)}\sum_{i,j} P_{ij}$ is the average expected degree in the alternative model.
Recall the definition $\alpha = (n-1)p + 1$.
% Denote $p = \frac{1}{n(n-1)}\sum_{i,j} P_{ij}$ as the average expected degree in the alternative model
% and recall the definition $\alpha = (n-1)p + 1$.
% denote $\widehat{\alpha} \coloneqq (n-1) \widehat{p} + 1$.
\begin{theorem}[Growth rate of $\lambda_1(\Ace)$]
\label{theorem:growth_rate_of_lambda_1_residual_under_alternative}
    Let the network $A$ be generated from an SBM$(\bm{g}^{(n)}, Q^{(n)})$ with some $K>1$. Assume $\min_{1\le i \le K} \frac{n_i}{n} = \Omega(1)$. Define $\delta_n \in [0,1)$ to be the maximum absolute difference among all pairs of entries in $Q^{(n)}$.
    Let $\widehat{p}=\frac{1}{n(n-1)} \sum_{i,j} A_{ij}$ and let $\Ace$ be defined in \eqref{eq:definition_of_Ac_and_Ace}.
    Assume $\alpha / \log n \to \infty$.
    Then the test statistic against $H_0: K=1$ satisfies
    \begin{equation}
        \frac{\lambda_1(\Ace)}{\sqrt{(n-1)\widehat{p}(1-\widehat{p})}} = \Omega(\delta_n n \alpha^{-1/2}) + O_{\Prob}\left(1\right).
        \label{eq:growth_rate_of_lambda_1_residual_under_alternative}
    \end{equation}
\end{theorem}
The proof is given in Appendix~\ref{section:proof_of_theorem_growth_rate_of_lambda_1_residual_under_alternative}.
When $Q^{(n)}$ does not change with $n$ and $\alpha$ grows linearly with $n$, then $\lambda_1(\Ane) = \Omega(n^{1/2})$, which aligns with the result in \citet{bickel2016hypothesis}.

Let $\widehat{\alpha} = (n-1)\widehat{p} + 1$.
Our next proposition provides an upper bound for the difference between $\widehat{\alpha}^{-1/2} \mu_1(\Hce)$ and $\widehat{\alpha}^{-1/2} \lambda_1(\Ace)$. The proof is given in Appendix~\ref{section:proof_of_proposition_difference_between_mu_1_H_and_lambda_1_alternative}, again considering the perturbation $\widetilde{H} = \widetilde{H}_0 + E$.
\begin{proposition}[Asymptotic difference between $\mu_1(\Hc)$ and $\lambda_1(\Ac)$]
\label{proposition:difference_between_mu_1_H_and_lambda_1_alternative}
    Under the same assumptions as in Theorem~\ref{theorem:growth_rate_of_lambda_1_residual_under_alternative},
    further assume that $\min_i \E[d_i] / \log n \to \infty$. Then
    \begin{equation*}
        \alpha^{-1/2} |\mu_1(\Hc) - \lambda_1(\Ac)|=
        % O_{\Prob} (\delta_n n \alpha^{-1}) +
        % + O_{\Prob} \left(\alpha^{-1} n^{1/2} \sqrt{\log n}\right),
         O_{\Prob} (1),
        % + O_{\Prob} (\alpha^{-1/2} (\log n)^{\frac{1}{2}(1 + \epsilon')}),
    \end{equation*}
    and 
    \begin{equation*}
        \widehat{\alpha}^{-1/2} |\mu_1(\Hce) - \lambda_1(\Ace)|= 
        % O_{\Prob} (\delta_n n \alpha^{-1}) +
        % + O_{\Prob} \left(\alpha^{-1} n^{1/2} \sqrt{\log n}\right).
        O_{\Prob} (1).
        % + O_{\Prob} (\alpha^{-1/2} (\log n)^{\frac{1}{2}(1 + \epsilon')}),
    \end{equation*} 
    % where $\epsilon' \in (0, \epsilon)$.
\end{proposition}
Consequently, if $\min_i \E[d_i] / \log n \to \infty$, the difference between $\mu_1(\Hce)$ and $\lambda_1(\Ace)$ is negligible,
and $\widehat{\alpha}^{-1/2} \mu_1(\Hce)$ has the same growth rate as in \eqref{eq:growth_rate_of_lambda_1_residual_under_alternative}.
Combining the spectral equivalency under the null hypothesis in Proposition~\ref{proposition:convergence_of_H_to_TW1}, using $\mu_1(\Hce)$ or $\lambda_1(\Ace)$ will lead to asymptotically equivalent power against $H_1: K>1$ in the semi-dense regime.

\subsection{Centering enhances signal in imbalanced settings}
\label{section:centering_enhances_signal}
In this section, we offer a heuristic explanation for why centering is beneficial in situations of imbalance.
We argue that centering strengthens the signal of the informative eigenvalue,
which explains the effectiveness of $\Ac$ under imbalance.
This motivates the formulation of our proposed operator $\Hc$ \eqref{eq:definition_of_Hc}.

Recall the construction of $\Ac$ when we test $H_1: K=2$ vs. $H_0: K=1$. It can be written in the following matrix form if we ignore the differences on its diagonal:
\begin{equation*}
    \Ac = A - P^{(0)}, \quad \textrm{ where } \quad P^{(0)} = p_0 \bm{1} \bm{1}^\top,
\end{equation*}
where $p_0$ is the average expected degree, i.e. $p_0 = \frac{1}{n^2} \sum_{i,j} \E[A_{ij}]$.

Previously in Section~\ref{section:asymptotic_power_guarantee}, we have shown the equivalence of $\mu_1(\Hc)$ and $\lambda_1(\Ac)$ when the network is not too sparse.
Similarly, when $\alpha \gg \log n$, the informative eigenvalues of $H$ and $A$, namely $\mu_2(H)$ and $\lambda_2(A)$, have also been shown equivalent \citep{wang2023limiting}.
Therefore, it suffices to argue why $\lambda_1(\Ac)$ can be better detected than $\lambda_2(A)$ under imbalance, which will, in turn, explain the advantage of $\mu_1(\Hc)$ brought by centering under imbalance.

Naturally, for the centered adjacency $\Ac$, the informative eigenvalue $\lambda_1(\Ac)$ will be close to the community-related signal $\lambda_1(\E[\Ac])$ when the signal is strong enough.
Similarly for $A$, $\lambda_2(A)$ will be close to the signal $\lambda_2(\E[A])$ when it is strong enough.
On one hand, the amount of noise in the spectrum of $A$ and $\Ac$ is almost the same.
By Cauchy's interlacing theorem and the fact that $P^{(0)}$ has rank one, the adjacency $A$ and its centered version $\Ac = A - P^{(0)}$ almost have the same size of bulk in their spectrum.
On the other hand, our next proposition shows that the centered version generally leads to a stronger signal.
We always have $\lambda_1(\E[A] - P^{(0)}) \ge \lambda_2(\E[A])$, with equality holds if and only if nodes in both communities have equal expected degrees.
% Consequently, we always have $\textrm{SNR}^* \ge \textrm{SNR}$. 
\begin{proposition}[Centering enhances signal under imbalance]
\label{proposition:lambda_1_EA_centered_larger_than_lambda_2_EA}
    Suppose network $A$ is generated from a two-block SBM with community size $n_1 \ge n_2$ and block-wise connection probability $Q^{(\delta)} \in \mathbb{R}^{2\times 2}$ given by the form for some $0 < p_0 < 1$ and $k\ge 0$:
    \begin{equation*}
        Q^{(\delta)} = \left(\begin{array}{cc}
          p_0(1 + x\delta)   &  p_0(1-\delta) \\
          p_0(1-\delta)   & p_0(1+ kx\delta)
        \end{array}\right),
    \end{equation*}
    where $x=\frac{2p(1-p)}{p^2 + k(1-p)^2}$ with $p = \frac{n_1}{n}$ so that we ensure the average expected degree over all nodes is maintained at $p_0$ for any $0 < \delta <1$.
    Let $P^{(0)} = p_0 \bm{1}\bm{1}^\top$ be a constant $n \times n$ matrix.
    Then for any $p\in [\frac{1}{2}, 1)$ and $0 <\delta <1$, we always have $\lambda_1(\E[A] - P^{(0)}) \ge \lambda_2(\E[A])$, with equality holds if and only if nodes in two communities have equal expected degrees.
    That is, when $p^2 = k(1-p)^2$. 
    
    (For simplicity, we do not impose the no-self-loop restriction here. The difference caused by the diagonal elements is negligible when $n$ becomes large enough.)
\end{proposition}
The proof is given in Appendix~\ref{section:proof_of_proposition:lambda_1_EA_centered_larger_than_lambda_2_EA}, where we derive the leading eigenvalues of $\E[A]$ and $\E[A] - P^{(0)}$ in terms of $p$, $k$ and $\delta$.

Note that the signal strength is the same with or without centering, in cases where we keep the per-community expected degree equal. 
This setting is illustrated in Figure~\ref{fig:compare_nonbacktracking_normalized_adjacency_under_sparsity},
\ref{fig:compare_power_K0=1_n0=500_n1=250_n2=250_p0=0.01_sparsity}
and \ref{fig:compare_power_K0=1_n0=500_n1=400_n2=100_equaldegree}, and we see test statistics based on $A$ and $\Ac$ exhibit similar testing power.
Otherwise, whenever there is any form of imbalance that leads to unequal per-community expected degree, 
the centered adjacency $\Ac$ always leads to a stronger signal than the original adjacency $A$.
As a result, test statistics with centering can detect the community-related signal much earlier under imbalance, as shown by the power curves in Figure~\ref{fig:compare_nonbacktracking_normalized_adjacency_under_community_size_imbalance}, \ref{fig:compare_nonbacktracking_normalized_adjacency_under_P_imbalance} and \ref{fig:compare_power_K0=1_n0=500_n1=400_n2=100_p0=0.08_balanced_P}.
More details about the mentioned figures can be found in Section~\ref{section:numerical_experiment}.

\section{Numerical Experiments}
% \subsection{Hypothesis testing using simulated null distributions}
In this section, we provide numerical evidence to demonstrate the effectiveness of our proposed operator.
We first introduce different settings for hypothesis tests in Section~\ref{section:model_parameterization}.
Subsequently, in Section~\ref{section:distribution_of_leading_eigenvalues}, we compare the empirical distribution of the leading eigenvalues of matrices $\Ace$, $H$, and $\Hce$.
As its byproduct, in Section~\ref{section:spectral_clustering_numerical_results}, we show the advantage of utilizing the leading eigenvector of $\Hce$ for label estimation.
Finally, in Section~\ref{section:power_of_test_statistics}, we compare the power of our proposed test with various other test statistics in the literature.

\label{section:numerical_experiment}
\subsection{Testing $H_1: K>1$ versus $H_0: K=1$}
\label{section:model_parameterization}
We test $H_1: K>1$ versus $H_0: K=1$.
% Fix the network size $n = 500$.
Under the null hypothesis, we assume that the network $A$ is generated from an Erd\H{o}s-R\'{e}nyi model $G(n, p_0)$.
Equivalently, $\E[A] = P^{(0)}$ under $H_0$, where $P^{(0)}_{ij} = p_0, \forall i \ne j$, and $P^{(0)}_{ii} = 0$. 
The alternative model is a 2-block SBM, where $n_1$ and $n_2=n - n_1$ represent the size of each community. The block-wise edge probability is parameterized with a tuning parameter $\delta \in [0, 1)$,
\begin{equation}
    Q^{(\delta)} = \left(\begin{array}{cc}
        Q_{11}(\delta) & Q_{12}(\delta) \\
        Q_{12}(\delta) & Q_{22}(\delta)
    \end{array} \right), \textrm{ where } Q_{12}(\delta)=p_0 (1-\delta),
    \label{eq:parameterization_of_Q^delta}
\end{equation}
and we consider three distinct cases for $Q_{11}(\delta)$ and $Q_{22}(\delta)$.
In the first case we keep $Q_{11}=Q_{22}$:
\begin{equation}
    Q_{11}(\delta) = Q_{22}(\delta) = p_0\left(1+ \frac{2 n_1 n_2}{n_1 (n_1-1) + n_2 (n_2-1)} \delta\right). \label{eq:P11_P22_balanced_P}
\end{equation}
In the second case we let $Q_{22}$ grow with $\delta$ while $Q_{11}$ remains unchanged,
\begin{equation}
    Q_{11}(\delta) = p_0, \quad Q_{22}(\delta) = p_0\left(1+ \frac{2 n_1}{ (n_2-1)} \delta\right). \label{eq:P11_P22_unbalanced_P}
\end{equation}
In the last case $Q_{11}\ne Q_{22}$ unless $n_1 = n_2$, but we ensure equal per-community expected degree,
\begin{equation}
    Q_{11}(\delta) = p_0\left(1+\frac{n_2}{n_1-1}\right)\delta, \quad Q_{22}(\delta) = p_0\left(1+ \frac{n_1}{n_2-1} \delta\right). \label{eq:P11_P22_constant_degree}
\end{equation}
In all these settings, the connecting probabilities $Q_{11}$, $Q_{22}$, and $Q_{12}$ are adjusted to ensure that $\E[\widehat{p}_0] = \frac{1}{n(n-1)}\sum_{i,j} P^{(\delta)}_{ij} = p_0$ for any $\delta$.
The null model can be taken as a special case with $\delta=0$. The larger the value of $\delta$, the greater the deviation of the 2-block SBM $P^{(\delta)}$ from the Erd\H{o}s-R\'{e}nyi null model $P^{(0)}$.
% Naturally, the power of a reasonable test should increase when $\delta$ gets larger.

\subsection{Distribution of leading eigenvalues}
\label{section:distribution_of_leading_eigenvalues}

\begin{figure}
\centering
\begin{minipage}{0.49\textwidth}
\centering
\includegraphics[width=1.0\textwidth]{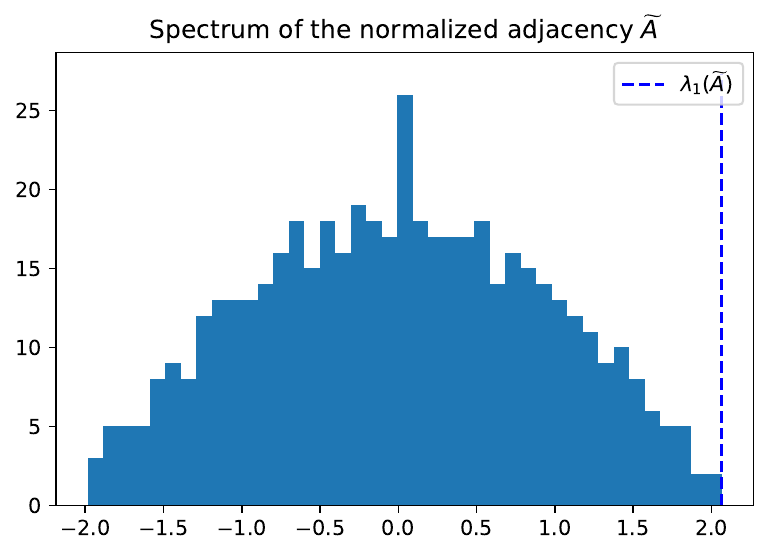}
\end{minipage}
\begin{minipage}{0.49\textwidth}
\centering
\includegraphics[width=1.0\textwidth]{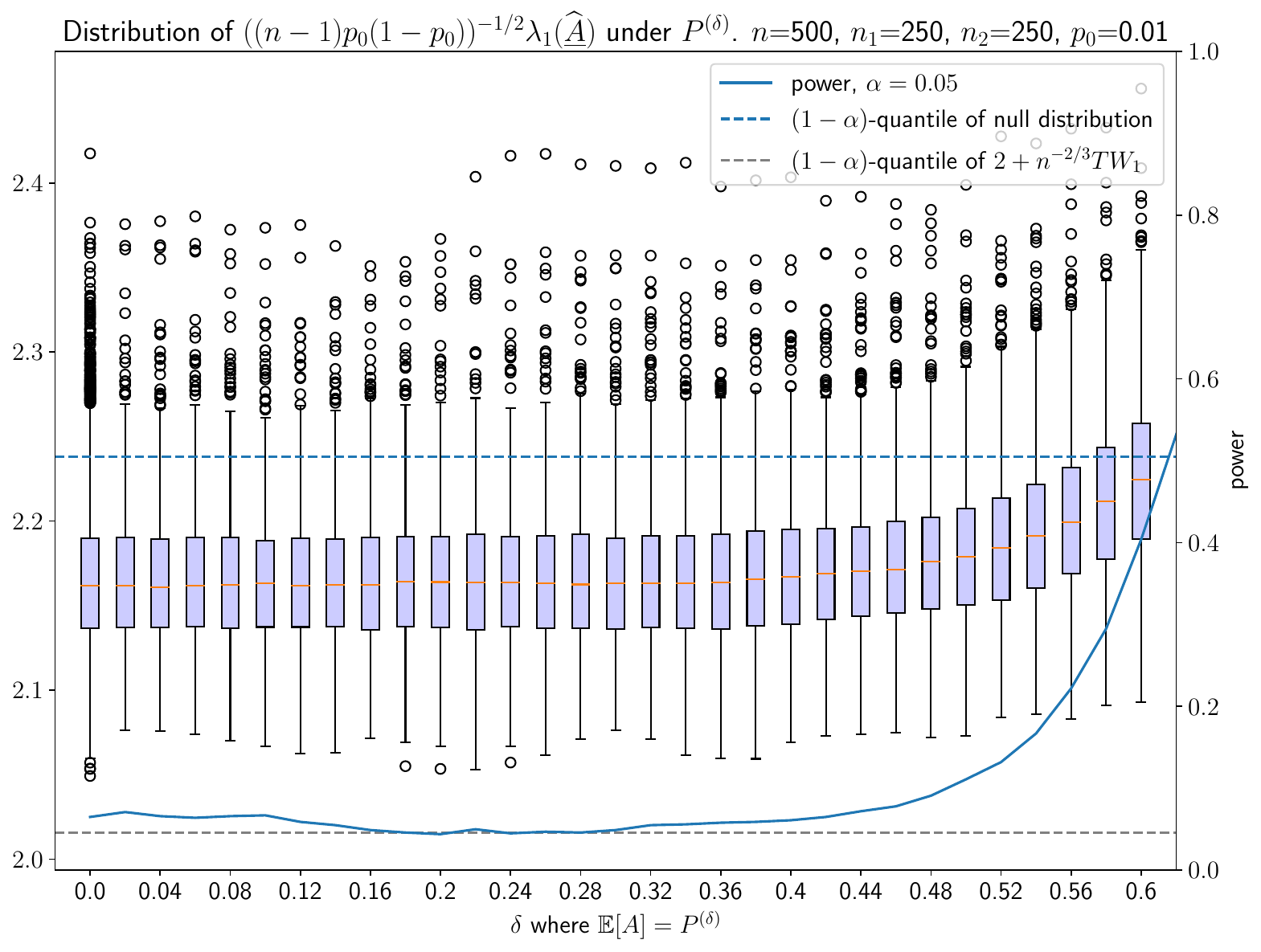}
\end{minipage}
\begin{minipage}{0.49\textwidth}
\centering
\includegraphics[width=1.0\textwidth]{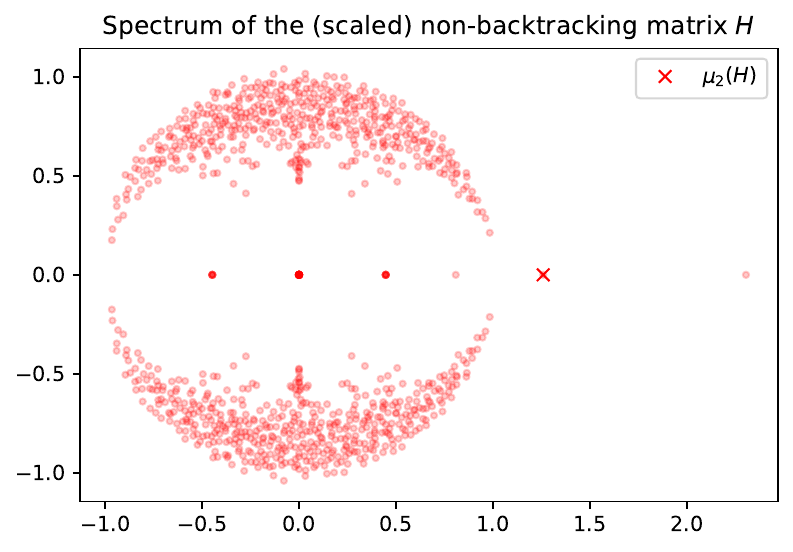}
\end{minipage}
\begin{minipage}{0.49\textwidth}
\centering
\includegraphics[width=1.0\textwidth]{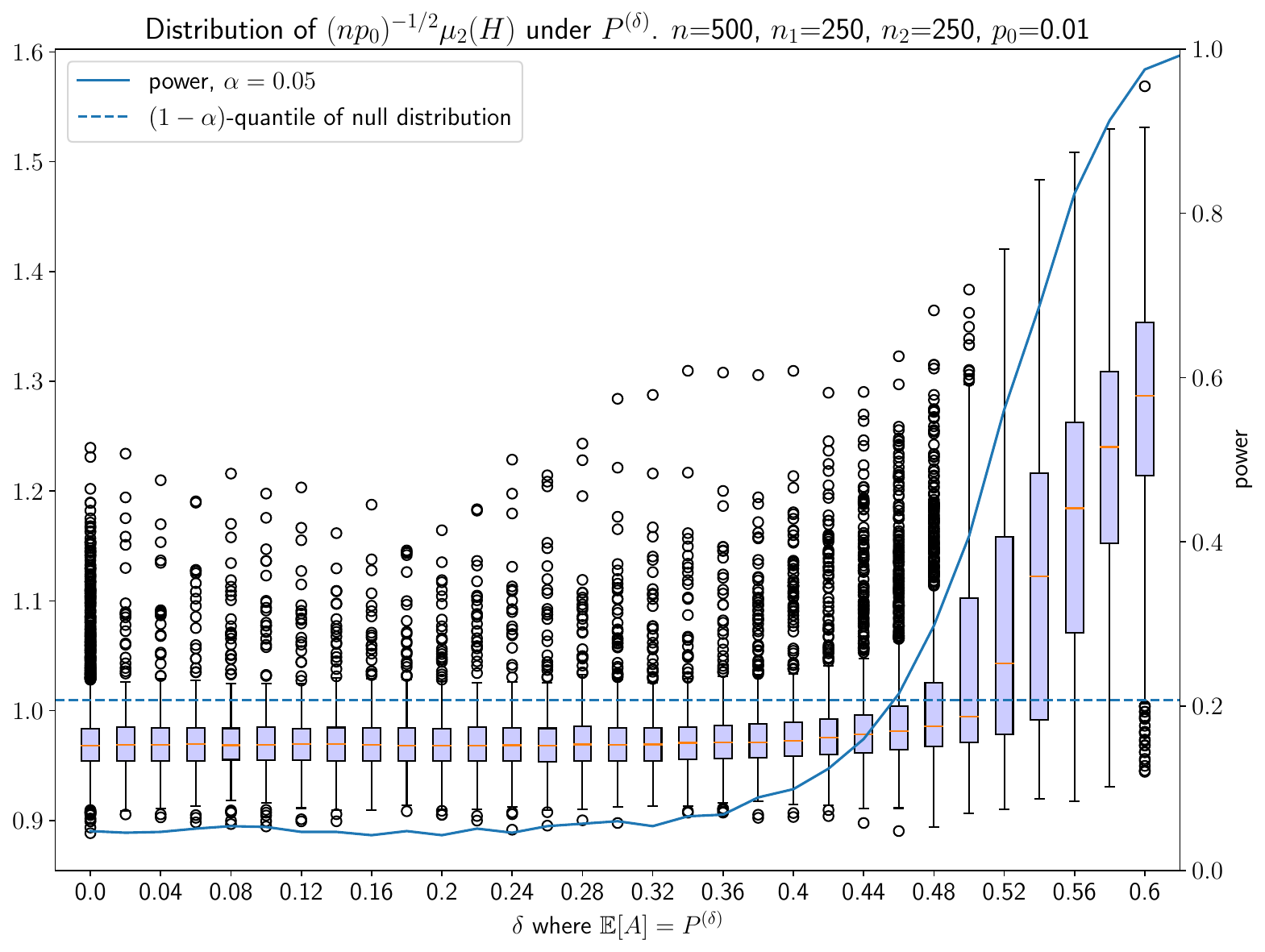}
\end{minipage}
\begin{minipage}{0.49\textwidth}
\centering
\includegraphics[width=1.0\textwidth]{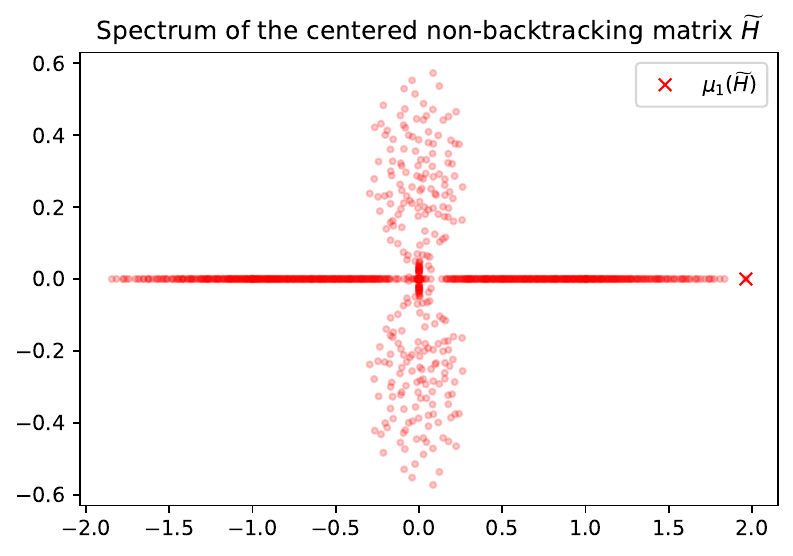}
\end{minipage}
\begin{minipage}{0.49\textwidth}
\centering
\includegraphics[width=1.0\textwidth]{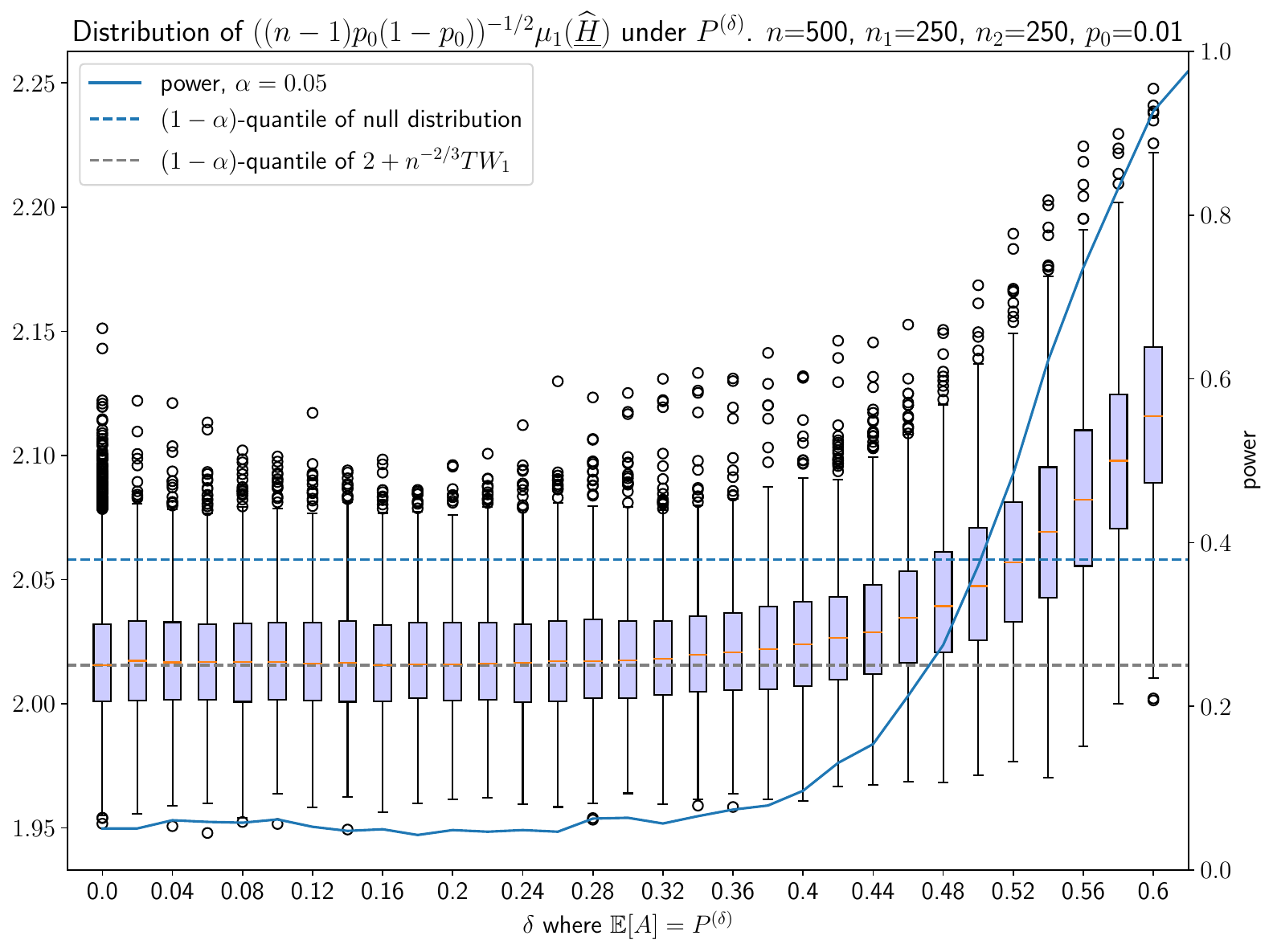}
\end{minipage}
\caption{Testing $H_1: K>1$ versus $H_0: K=1$ under sparsity. The two communities have equal sizes $n_1=n_2=250$. We set $p_0=0.01$ and let $Q_{11}$ and $Q_{22}$ grow with $\delta$ simultaneously as in \eqref{eq:P11_P22_balanced_P}.
The left panel shows the spectrum of each operator, where we fix $\delta = 0.6$. 
The right panel shows how the distribution of each test statistic changes with $\delta$.
Also shown is the power curve, where the rejection rule is based on the $(1-\alpha)$-quantile of the null distribution.
The values of the test statistic correspond to the left $y$-axis, while the power values correspond to the right $y$-axis.
% , shown by the blue dashed line.
% The KS threshold for $\delta$ is $(n p_0)^{-1/2}\approx 0.45$ shown by the red dashed vertical line.
}
\label{fig:compare_nonbacktracking_normalized_adjacency_under_sparsity}
\end{figure}

\begin{figure}
\centering
\begin{minipage}{0.49\textwidth}
\centering
\includegraphics[width=1.0\textwidth]{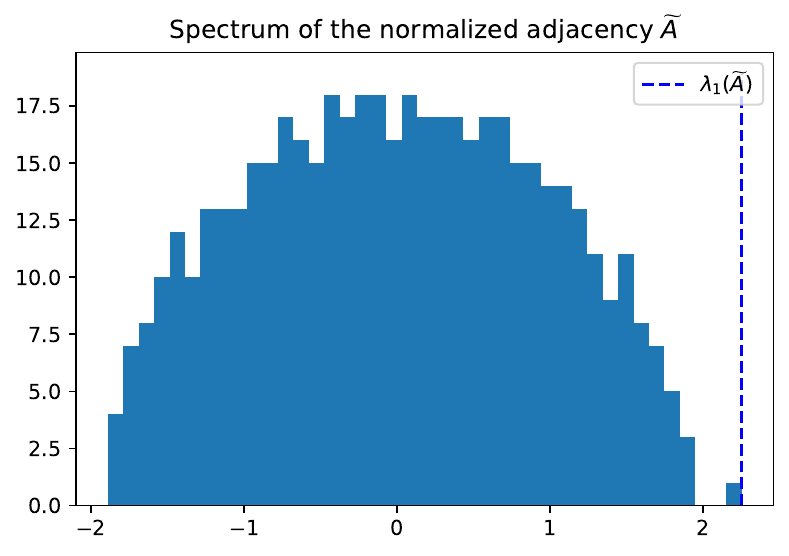}
\end{minipage}
\begin{minipage}{0.49\textwidth}
\centering
\includegraphics[width=1.0\textwidth]{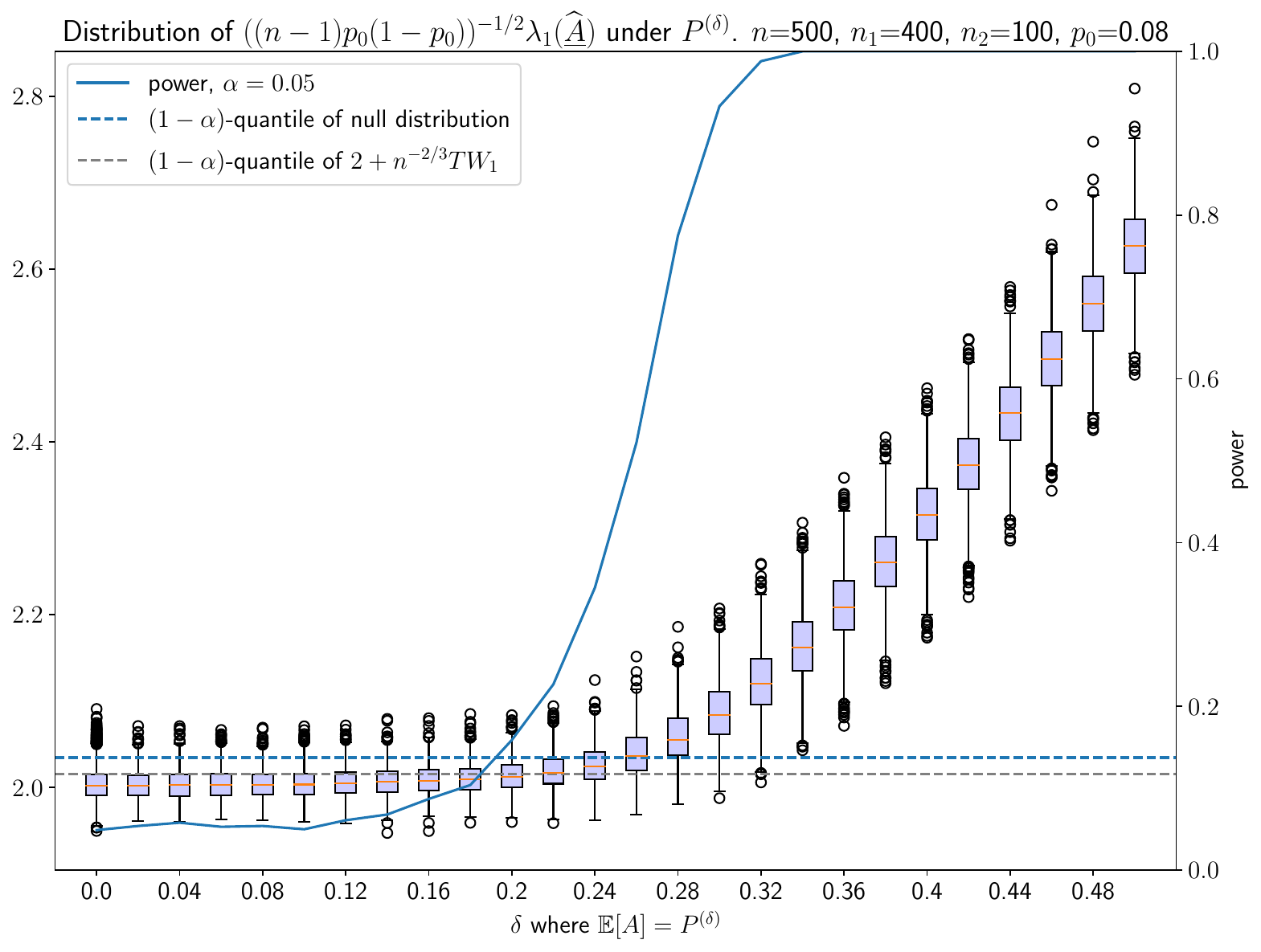}
\end{minipage}
\begin{minipage}{0.49\textwidth}
\centering
\includegraphics[width=1.0\textwidth]{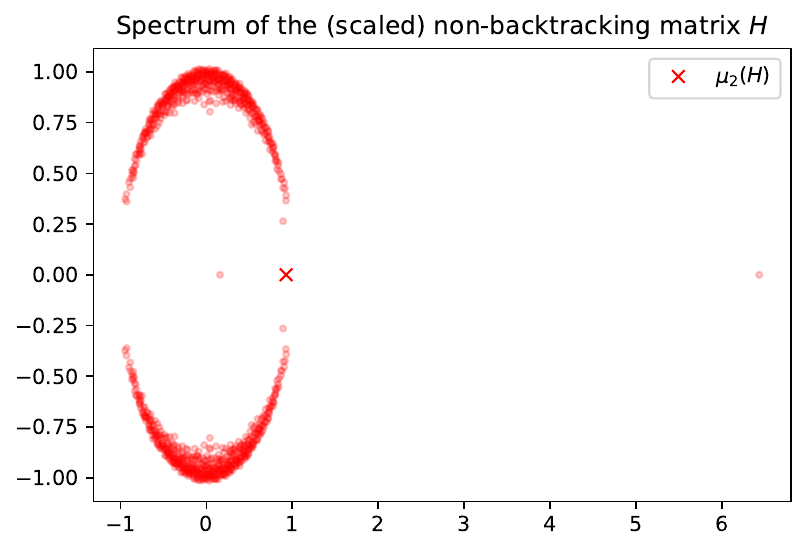}
\end{minipage}
\begin{minipage}{0.49\textwidth}
\centering
\includegraphics[width=1.0\textwidth]{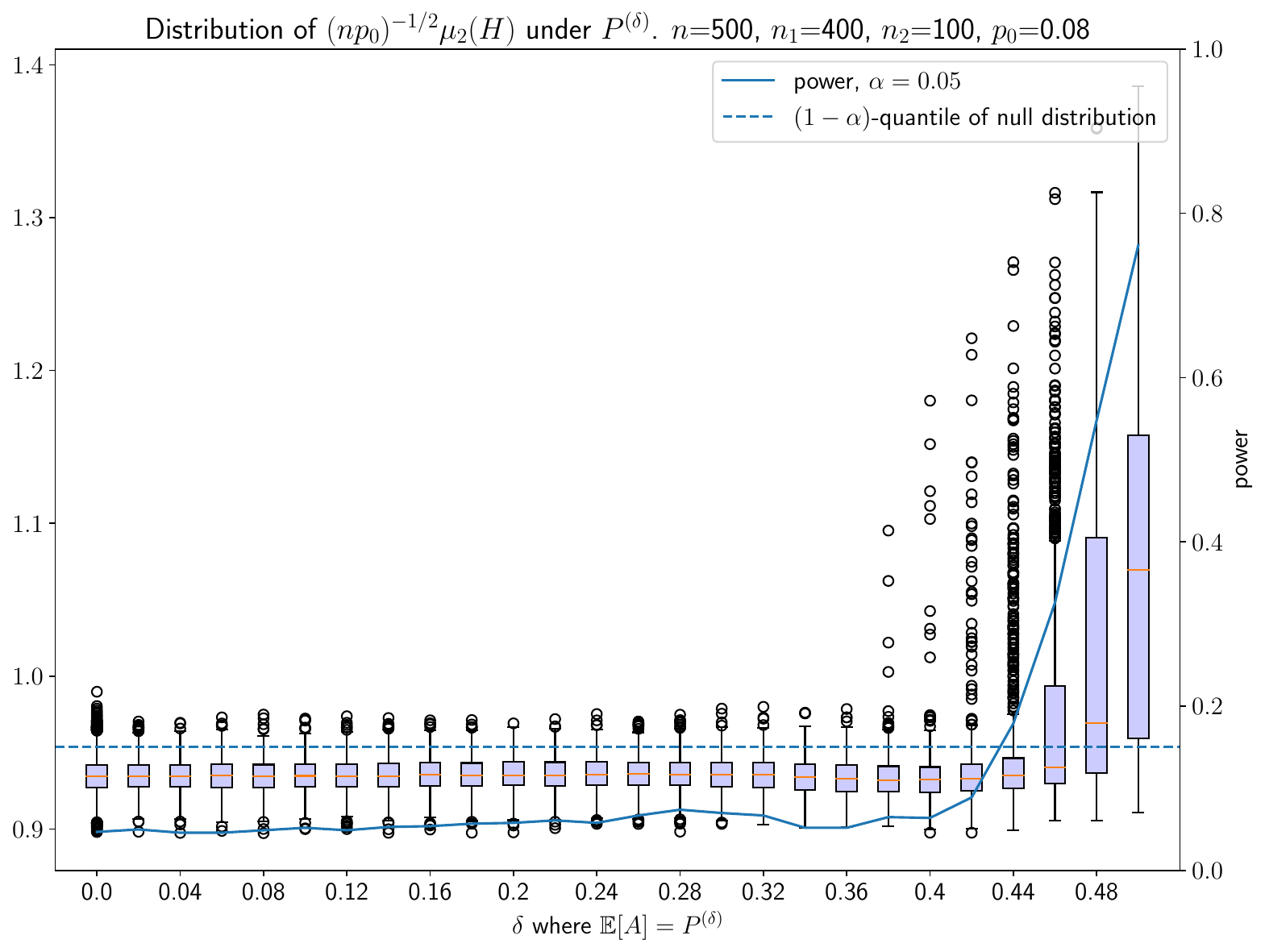}
\end{minipage}
\begin{minipage}{0.49\textwidth}
\centering
\includegraphics[width=1.0\textwidth]{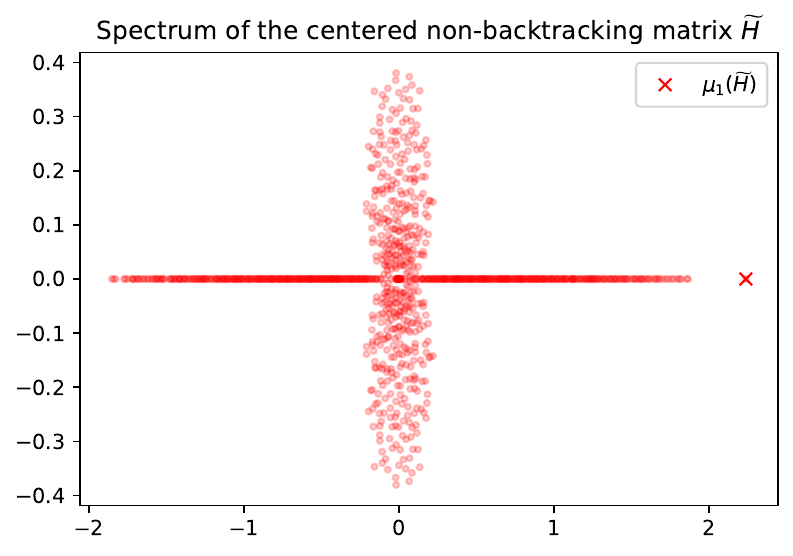}
\end{minipage}
\begin{minipage}{0.49\textwidth}
\centering
\includegraphics[width=1.0\textwidth]{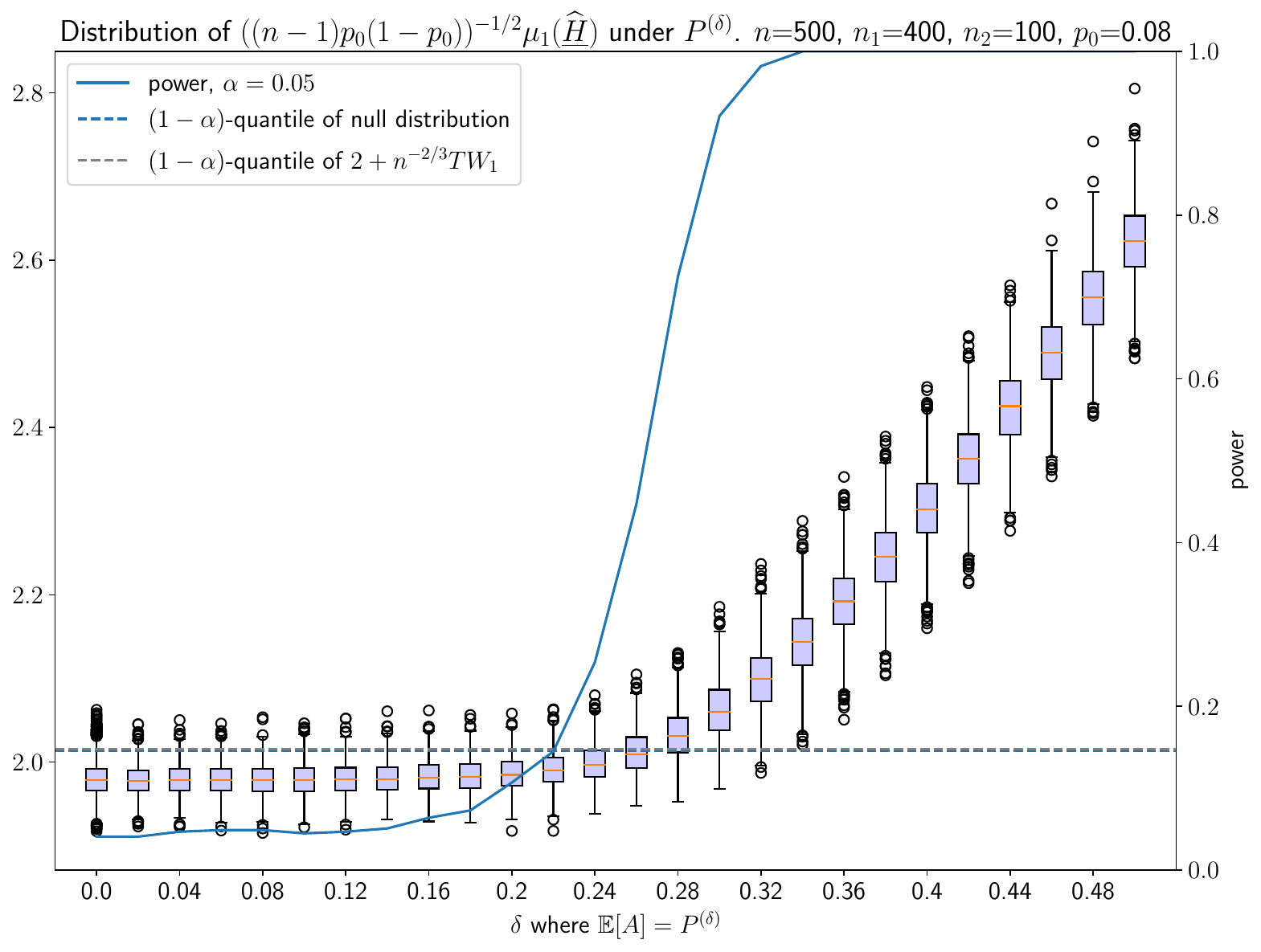}
\end{minipage}
\caption{Testing $H_1: K>1$ versus $H_0: K=1$ under community size imbalance $n_1 \ne n_2$. The larger community has $n_1=400$, and the smaller community has $n_2=100$. We set $p_0=0.08$ and let $Q_{11}$ and $Q_{22}$ grow with $\delta$ simultaneously \eqref{eq:P11_P22_balanced_P}. We fix $\delta = 0.4$ for the spectra on the left.
% In the right panel, we represent the KS thresholds for operators with and without centering. 
% Note that the KS threshold of $\E[A] - P^{(0)}$ is smaller than the KS threshold of $\E[A]$.
}
\label{fig:compare_nonbacktracking_normalized_adjacency_under_community_size_imbalance}
\end{figure}

\begin{figure}
\centering
\begin{minipage}{0.49\textwidth}
\centering
\includegraphics[width=1.0\textwidth]{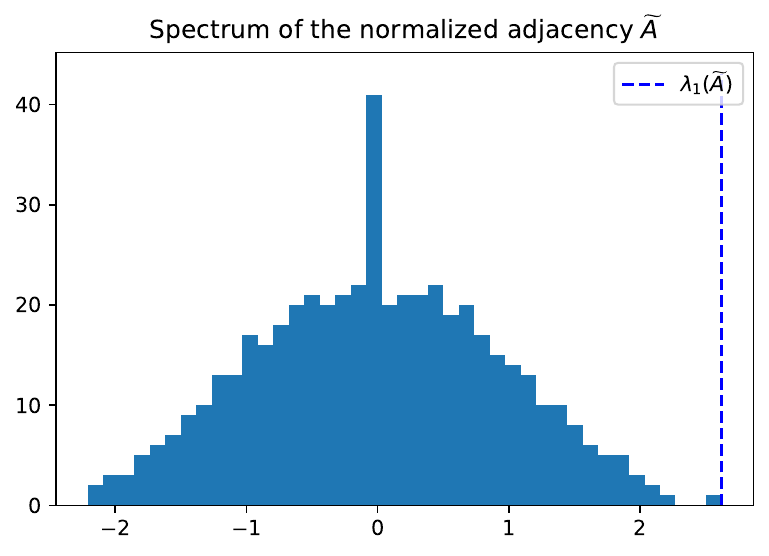}
\end{minipage}
\begin{minipage}{0.49\textwidth}
\centering
\includegraphics[width=1.0\textwidth]{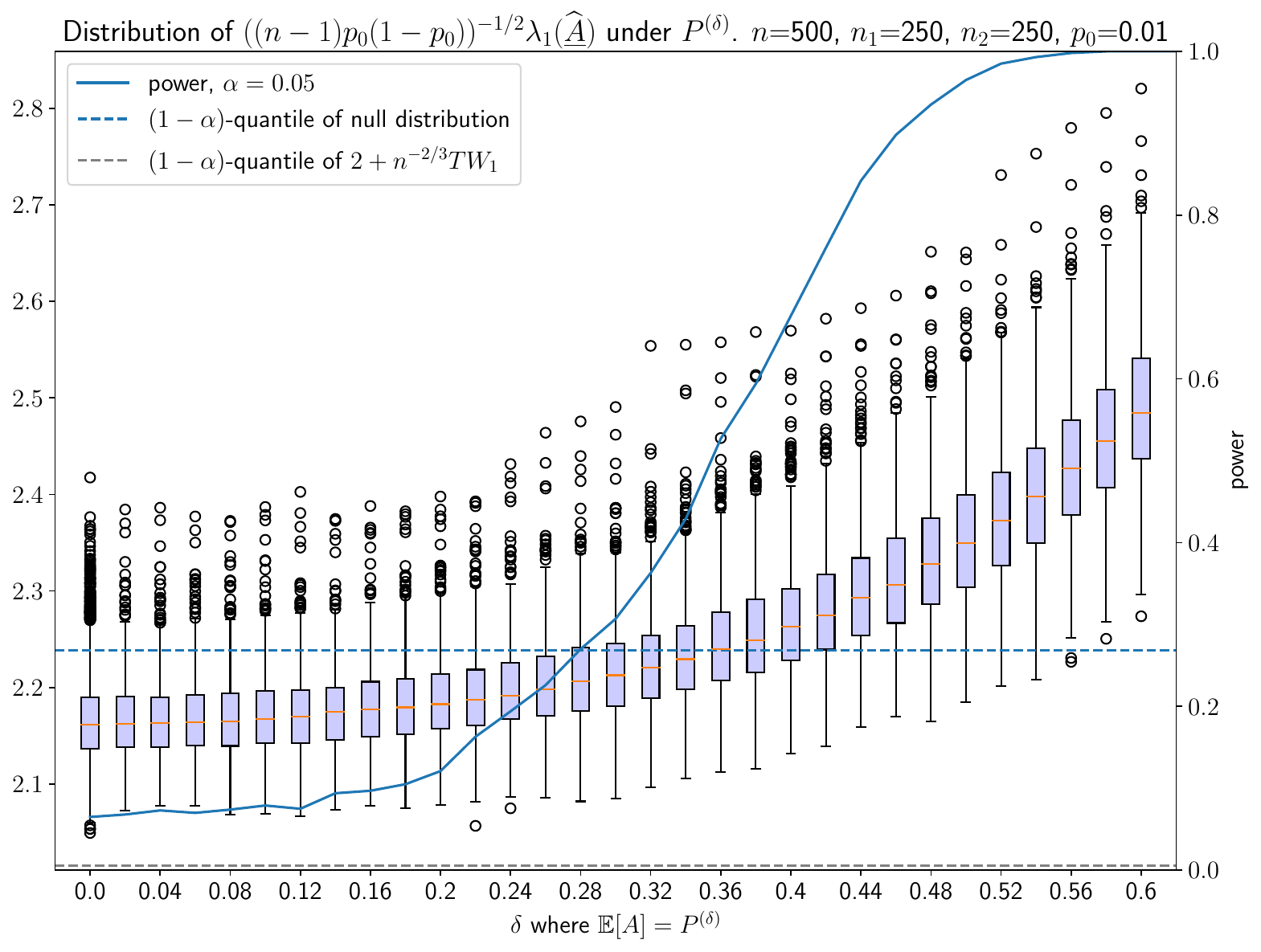}
\end{minipage}
\begin{minipage}{0.49\textwidth}
\centering
\includegraphics[width=1.0\textwidth]{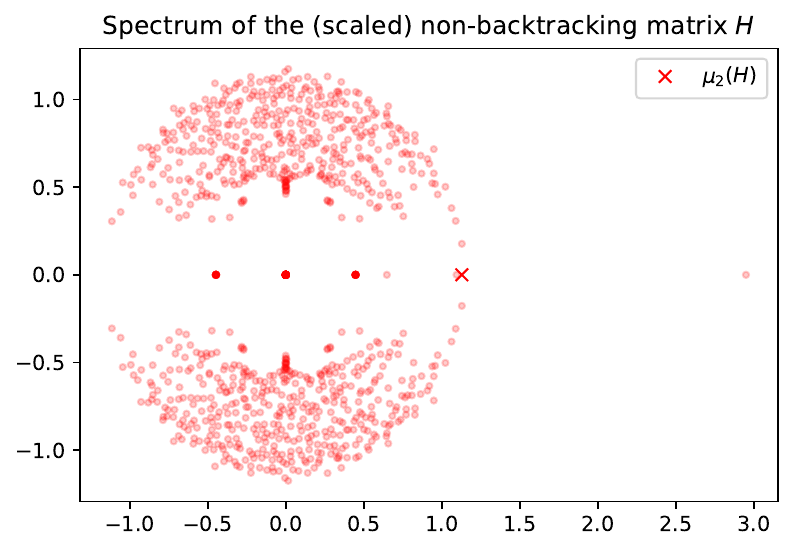}
\end{minipage}
\begin{minipage}{0.49\textwidth}
\centering
\includegraphics[width=1.0\textwidth]{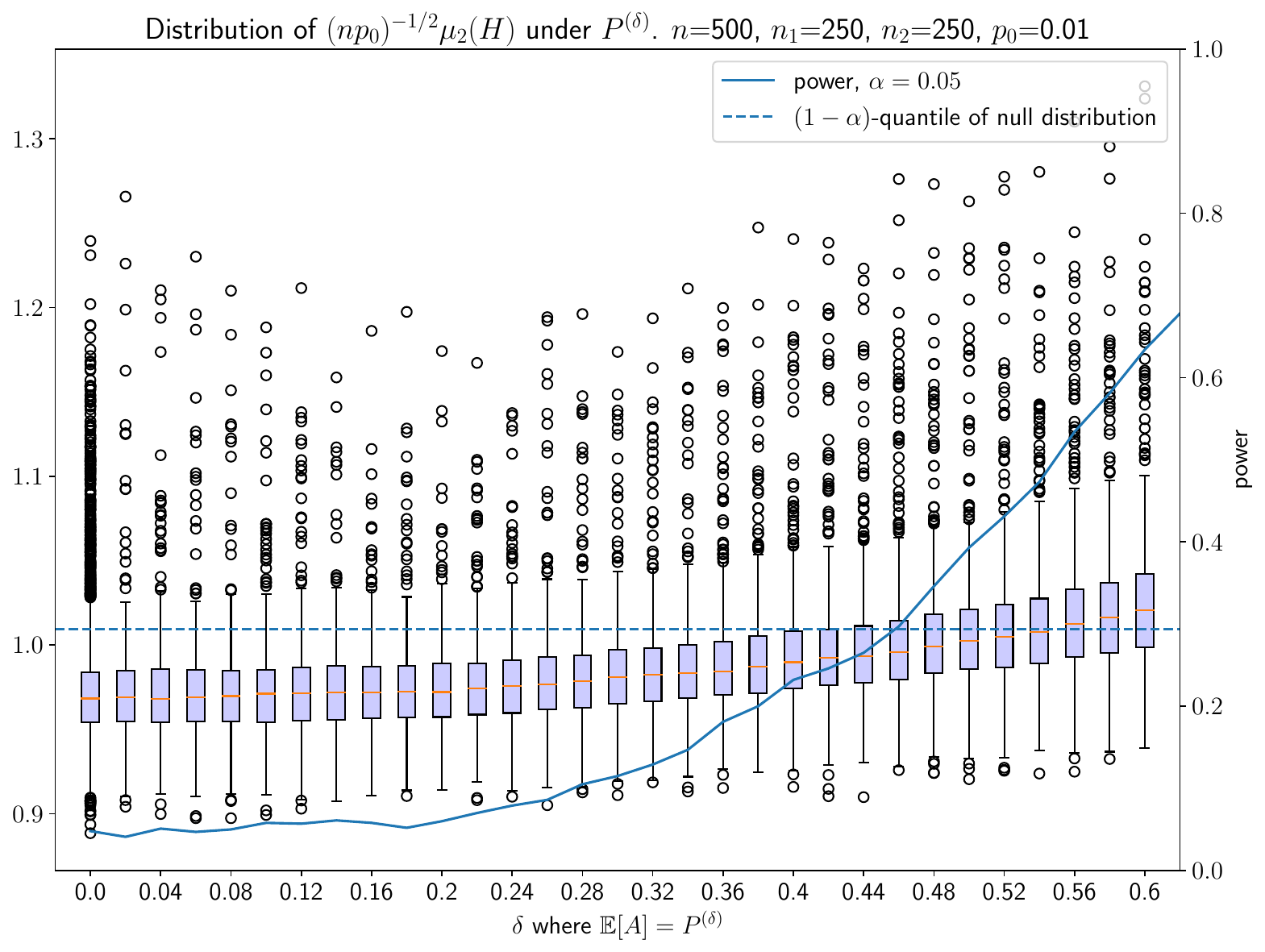}
\end{minipage}
\begin{minipage}{0.49\textwidth}
\centering
\includegraphics[width=1.0\textwidth]{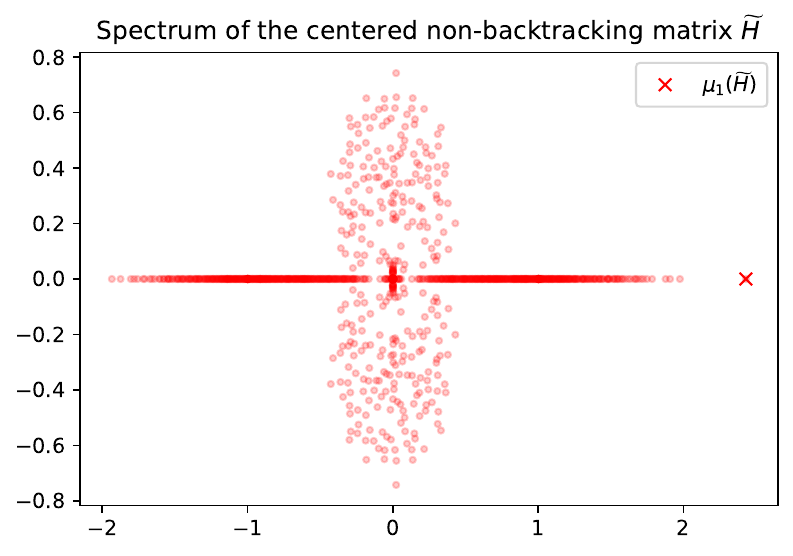}
\end{minipage}
\begin{minipage}{0.49\textwidth}
\centering
\includegraphics[width=1.0\textwidth]{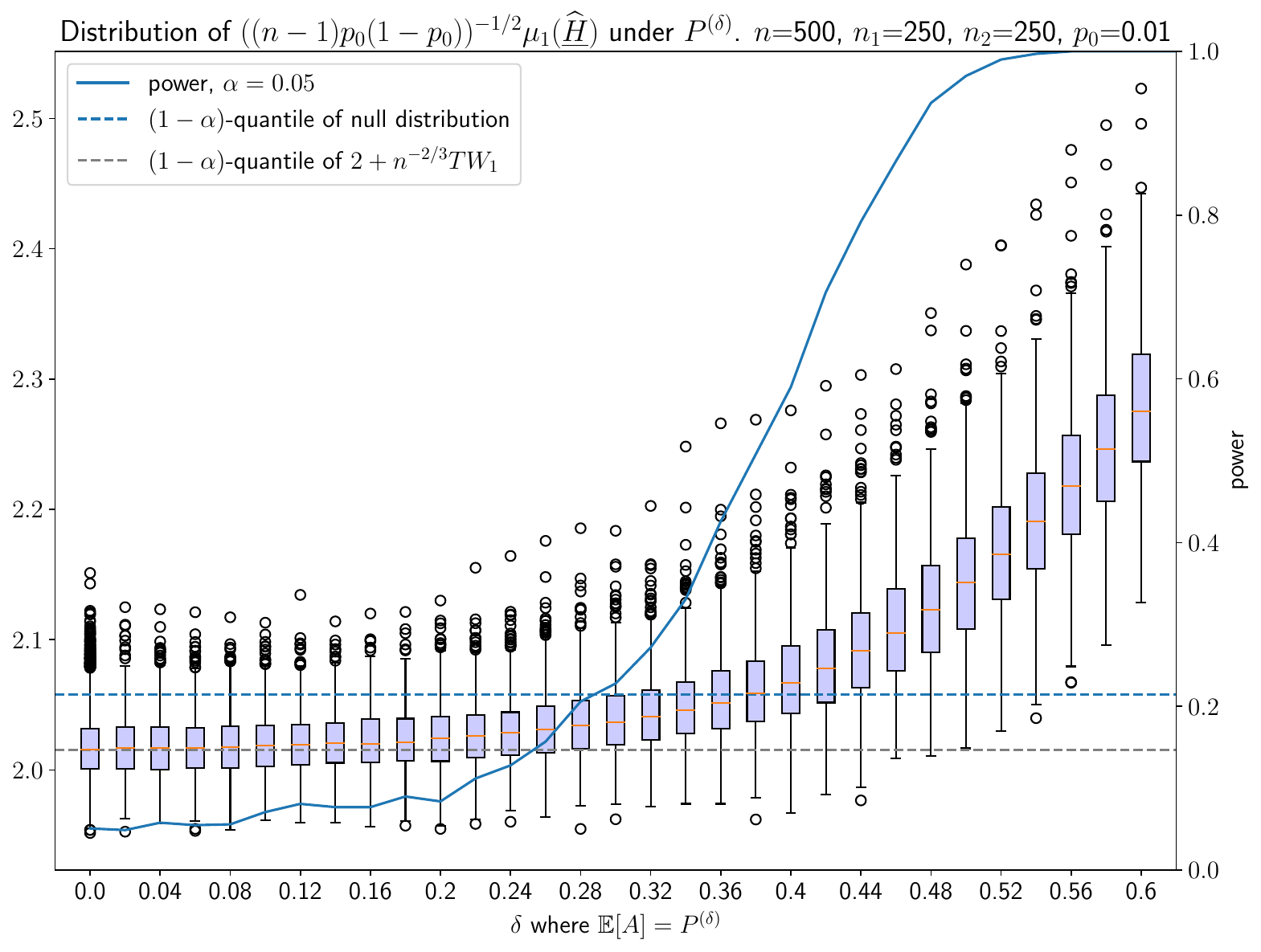}
\end{minipage}
\caption{Testing $H_1: K>1$ versus $H_0: K=1$ under imbalance caused by $Q_{11} \ne Q_{22}$. The two communities have equal sizes $n_1=n_2=250$. We set $p_0=0.01$ and let $Q_{22}$ grow with $\delta$ but $Q_{11}$ remains unchanged \eqref{eq:P11_P22_unbalanced_P}.
We fix $\delta = 0.6$ for the spectra on the left.
% There is no line representing the KS threshold of $\E[A]$, because no such $\delta$ exists for which $\textrm{SNR}$ \eqref{eq:SNR} exceeds 1.
}
\label{fig:compare_nonbacktracking_normalized_adjacency_under_P_imbalance}
\end{figure}

We first study the spectrum of $\Hce$ \eqref{eq:definition_of_Hc} along with its closely related counterparts, $\Ace$ and $H$ \eqref{eq:definition_of_H}.
Our focus primarily revolves around their leading eigenvalues, $\mu_1(\Hce)$, $\lambda_1(\Ace)$, and $\mu_2(H)$, with proper scaling so that they are of order $O(1)$.
Additionally, we use these eigenvalues as test statistics and evaluate their power against $H_1:K>1$. To illustrate, for $\mu_1(\Hce)$, we simulate its null distribution under $P^{(0)}$, determine its ($1-\alpha$)-quantile, and subsequently for each $P^{(\delta)}$, reject $H_0$ when $\mu_1(\Hce)$ exceeds that quantile.
The tests based on $\lambda_1(\Ace)$ and $\mu_2(H)$ follow a similar manner.
See also Section~\ref{section:power_of_test_statistics}.

We present three scenarios to examine the behaviors of different test statistics under various levels of sparsity and imbalance.

\subsubsection{Sparse settings}
The first scenario is characterized by sparsity yet balance, with $p_0 = 0.01$, $n_1=n_2=250$, and let $Q_{11} = Q_{22}$ grow with $\delta$ simultaneously as in \eqref{eq:P11_P22_balanced_P}.
The results are illustrated in Figure~\ref{fig:compare_nonbacktracking_normalized_adjacency_under_sparsity}.

The left panel shows the typical spectrum of each operator when we set $\delta = 0.6$ and take $\widehat{p} = p$ (recall that the difference between $\widehat{p}$ and $p$ is minimal). 
Specifically, we examine one random realization of their ``population" version, $\An$, $H$, and $\Hn$.
Both the spectrum of $H$ and $\Hn$ successfully identify two communities based on their leading eigenvalues.
However, $\lambda_1(\An)$ fails to detect the signal as it is not distinct from the bulk.

The right panel shows how the distributions of the test statistics change with $\delta$.
The blue dashed line represents the $(1-\alpha)$-quantile of the empirical null distribution simulated under $P^{(0)}$, while the gray dashed line corresponds to the $(1-\alpha)$-quantile of the asymptotic limit distribution under the dense regime, i.e., $2 + n^{-2/3} TW_1^{-1}(1-\alpha)$.
% The KS threshold is the same for \eqref{eq:SNR} and \eqref{eq:SNR*} in this scenario, indicated by the red dashed line.
% The distributions of both $\mu_2(H)$ and $\mu_1(\Hce)$ shift dramatically beyond the KS threshold, while the distribution of $\lambda_1(\Ace)$ shows minimal movement.
% As a result, the power of using $\lambda_1(\Ace)$ is significantly inferior to the test using the other two non-backtracking operators.
The power of using $\lambda_1(\Ace)$ is significantly inferior to the test using the other two non-backtracking operators, as its distribution shows minimal movement when $\delta$ increases.
Finally, it is worth noticing that under this relatively sparse setting, the empirical distribution of $\lambda_1(\Ace)$ and $\mu_1(\Hce)$ are both positively biased against the Tracy-Widom limit, although the deviance of $\lambda_1(\Ace)$ is much more severe.

\subsubsection{Community-size Imbalance}
In this scenario, we reduce the sparsity level and introduce an imbalance in community size, with $p_0=0.08$, $n_1=100$, and $n_2=400$.
Again, we let $Q_{11} = Q_{22}$ grow with $\delta$ simultaneously as in \eqref{eq:P11_P22_balanced_P}.
The results are illustrated in Figure~\ref{fig:compare_nonbacktracking_normalized_adjacency_under_community_size_imbalance}.

We fix $\delta = 0.4$ in the left panel and examine the typical spectrum of $\An$, $H$, and $\Hn$.
Both statistics based on centered adjacency, $\lambda_1(\An)$ and $\mu_1(\Hn)$, are able to identify two communities. However, $\mu_2(H)$ cannot clearly indicate the existence of a second community, as it falls inside the bulk.
In the right panel, $\lambda_1(\Ace)$ and $\mu_1(\Hce)$ detect the signal much earlier than $\mu_2(H)$ as indicated by their power curves.
This is attributed to the enhanced signal as a result of centering.

% The distinction becomes more evident when examining two different KS thresholds presented in the right panel.
% In the case of the original non-backtracking operator $H$, the KS threshold corresponds to the value of $\delta$ at which $\textrm{SNR}$ \eqref{eq:SNR} surpasses 1.
% On the other hand, for the two operators incorporating centering, $\Ace$ and $\Hce$, the KS threshold is determined by the value of $\delta$ such that $\textrm{SNR}^*$ \eqref{eq:SNR*} exceeds 1.
% Notably, this latter KS threshold, based on the informative eigenvalue of $\E[A]- P^{(0)}$, is smaller due to Proposition~\ref{proposition:lambda_1_EA_centered_larger_than_lambda_2_EA} and the imbalance in expected degrees between the two communities.
% Consequently, $\lambda_1(\Ace)$ and $\mu_1(\Hce)$ detect the signal much earlier than $\mu_2(H)$ as indicated by their power curves.

Lastly, we observe that with decreased sparsity, the discrepancy between the quantile of the null distribution and the Tracy-Widom distribution is much smaller than in scenario \romannumeral1.
Nevertheless, we still see a noticeable difference in the case of $\lambda_1(\Ace)$, while for $\mu_1(\Hce)$, the quantile of its null distribution closely resembles the Tracy-Widom quantile, aligning with our observation in Figure~\ref{fig:null_distributions_with_scaling}.
On the other hand, the growth of $\lambda_1(\Ace)$ and $\mu_1(\Hce)$ looks almost the same when $\delta$ gets larger, aligning with our asymptotic power analysis in Proposition~\ref{proposition:difference_between_mu_1_H_and_lambda_1_alternative}.

\subsubsection{Edge-probability imbalance}
The third scenario explores a different type of imbalance, where the community sizes are equal, but the expected degrees for nodes in each community are unequal.
In particular, we set $p_0=0.01$ and $n_1= n_2 = 250$, and let $Q_{22}$ grow while $Q_{11}$ remains unchanged \eqref{eq:P11_P22_unbalanced_P}.
The results are illustrated in Figure~\ref{fig:compare_nonbacktracking_normalized_adjacency_under_P_imbalance}.

We fix $\delta=0.6$ in the left panel and look at the typical spectrum of $\An$, $H$ and $\Hn$. Again, both $\lambda_1(\An)$ and $\mu_1(\Hn)$ clearly indicate more than one community with their leading eigenvalues well separated from the bulk.
The spectrum of $H$ displays two circles with distinct radii.
However, $\mu_2(H)$ fails to fall outside the larger circle.
The right panel shows the notably inferior performance of $\mu_2(H)$ under this type of imbalance.
% In the right panel, we once again mark the KS threshold of $\E[A] - P^{(0)}$ by identifying the value of $\delta$ for which the $\textrm{SNR}^*$ \eqref{eq:SNR*} exceeds 1. However, no such threshold exists for $\E[A]$, as there is no $\delta$ that yields a value of $\textrm{SNR}$ \eqref{eq:SNR} greater than 1.
% This elucidates the notably inferior performance of $\mu_2(H)$ under this type of imbalance.

\subsection{Label estimation using the leading eigenvector of $\Hce$}
\label{section:spectral_clustering_numerical_results}

\begin{figure}
    \centering
    \includegraphics[width=1.0\textwidth]{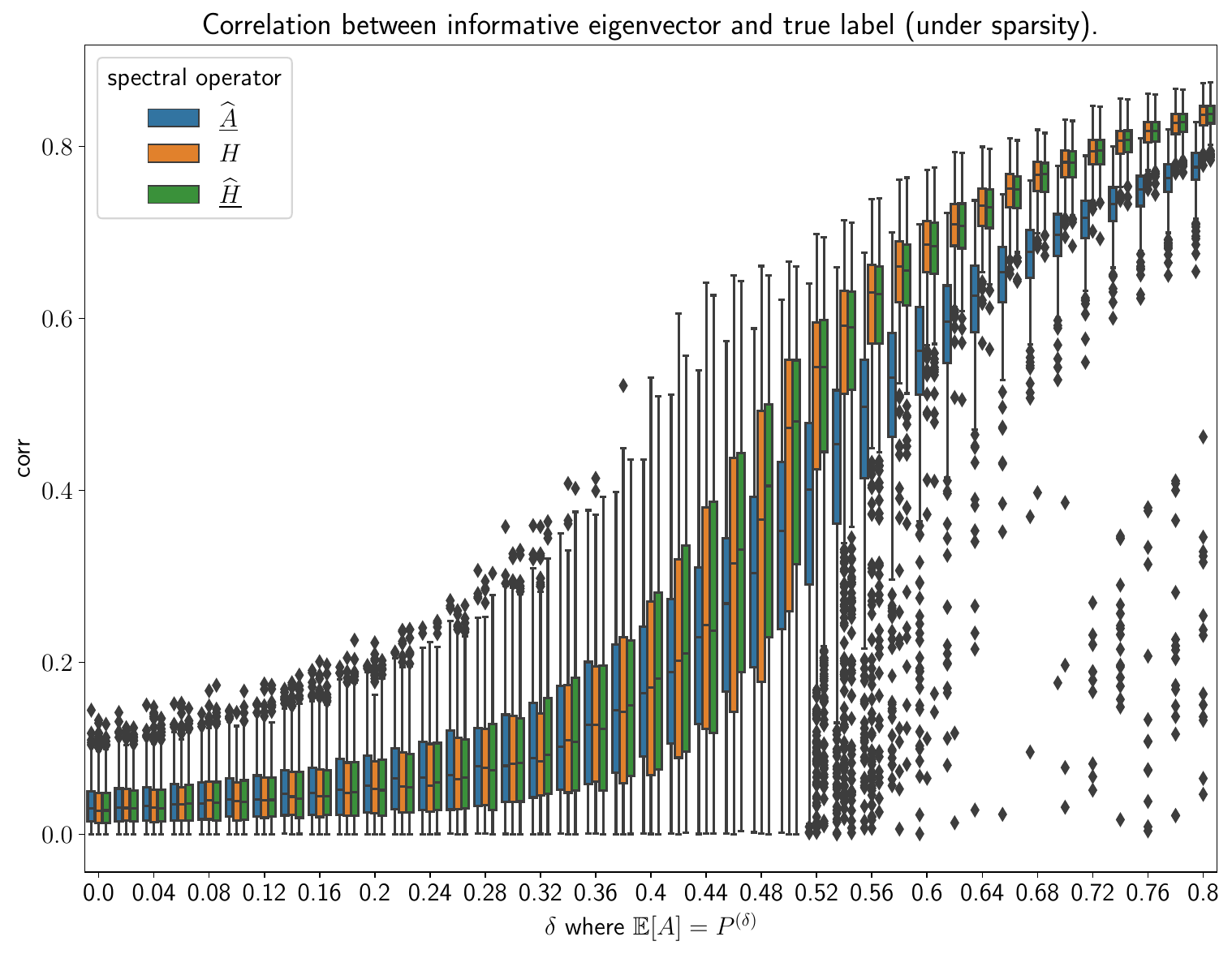}
    \caption{Correlation between informative eigenvectors and the true binary label.
    For each value of $\delta$, the model $P^{(\delta)}$ represents a two-block SBM with parameterized by \eqref{eq:parameterization_of_Q^delta}. 
    The setup mirrors that of Figure~\ref{fig:compare_nonbacktracking_normalized_adjacency_under_sparsity}, emphasizing sparsity.
    The two communities have equal sizes $n_1 = n_2 = 250$. We set $p_0 = 0.01$ and let $Q_{11}$ and $Q_{22}$ grow with $\delta$ simultaneously as in \eqref{eq:P11_P22_balanced_P}.
    }
    \label{fig:compare_spectral_clustering_under_sparsity}
\end{figure}

\begin{figure}
    \centering
    \includegraphics[width=1.0\textwidth]{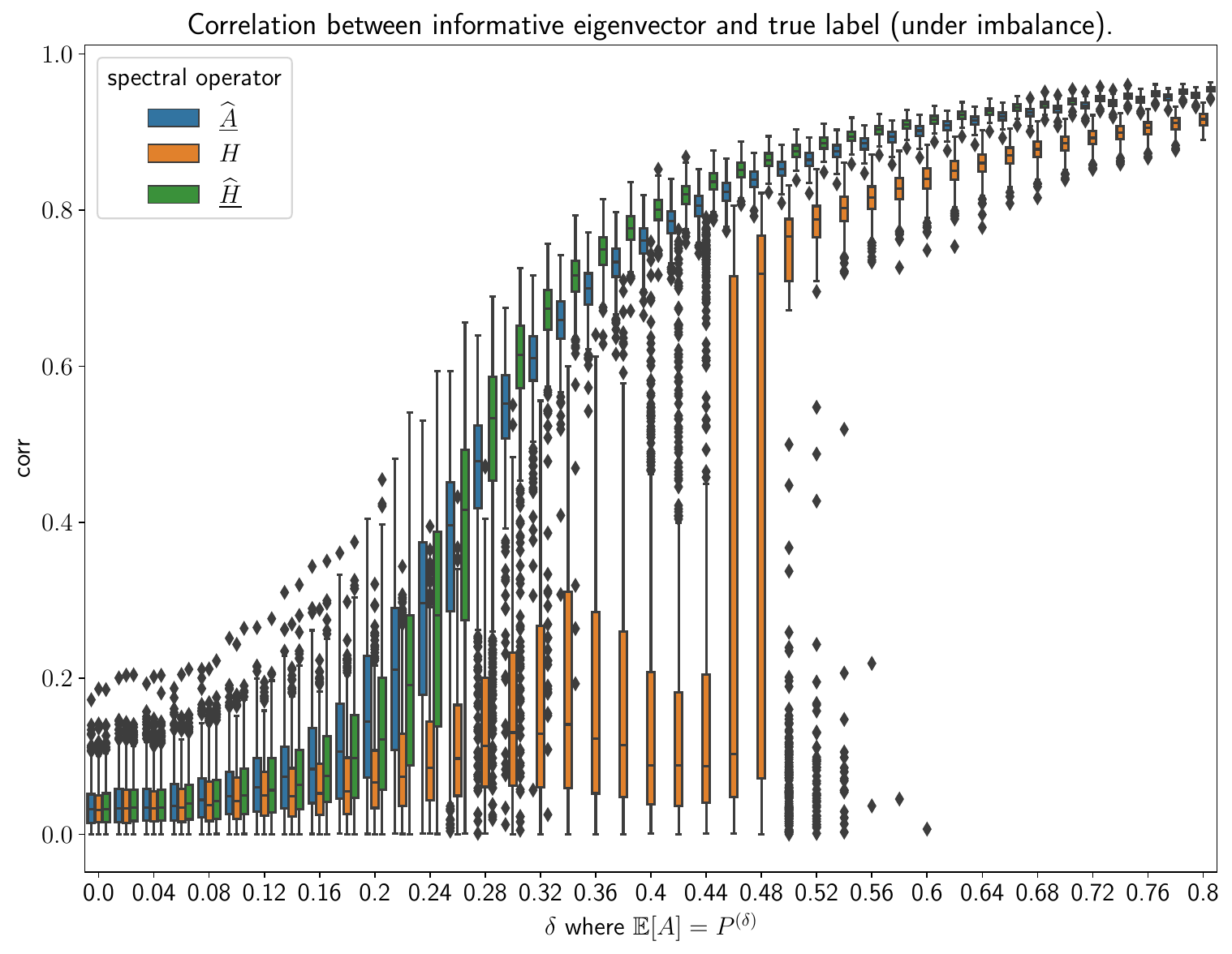}
    \caption{Correlation between informative eigenvectors and the true binary label.
    The setup mirrors that of Figure~\ref{fig:compare_nonbacktracking_normalized_adjacency_under_community_size_imbalance} and \ref{fig:compare_power_K0=1_n0=500_n1=400_n2=100_p0=0.08_balanced_P}, emphasizing community size imbalance.
    The two communities have sizes $n_1 = 400$ and $n_2=100$, respectively.
    We set $p_0 = 0.08$ and let $Q_{11}$ and $Q_{22}$ grow with $\delta$ simultaneously as in \eqref{eq:P11_P22_balanced_P}.
    }
    \label{fig:compare_spectral_clustering_under_imbalance}
\end{figure}

The eigenvector of $\Hce$ corresponding to the eigenvalue $\mu_1(\Hce)$ is of the form $\left(\begin{array}{c}
     x_1 \\
     \mu_1^{-1} x_1
\end{array}\right)$ for some $x_1 \in \mathbb{R}^n$,
so that
\begin{equation*}
    \Ace x_1 +  \mu_1 ^{-1} (\In - \Dce) x_1 = \mu_1 x_1
\end{equation*}
It is reasonable to expect a nontrivial correlation between the ``partial" eigenvector $x_1$ and the community label $\bm{g}$ under the alternative model with $K=2$.
Moreover, given that the chance of $\mu_1(\Hce)$ being informative is higher than $\lambda_1(\Ace)$ under sparsity and higher than $\mu_2(H)$ under imbalance, bi-partitioning using the eigenvector of $\Hce$ should outperform bi-partitioning using the eigenvector of $\Ace$ and $H$ under these settings, respectively.

To elaborate, we compare spectral clustering using the following informative eigenvectors.
For $\Ace$, we use its eigenvector corresponding to its largest eigenvalue.
For $H$, we use the first half of its eigenvector corresponding to its second-largest eigenvalue (in magnitude).
For $\Hce$, we use the first half of its eigenvector corresponding to its largest eigenvalue (in magnitude).

The advantage of using $x_1(\Hce)$ for spectral clustering is indeed demonstrated through simulations. 
In Figure~\ref{fig:compare_spectral_clustering_under_sparsity} and \ref{fig:compare_spectral_clustering_under_imbalance}, we assess the correlation between the informative eigenvectors and the true binary label under each $P^{(\delta)}$ defined in Section~\ref{section:model_parameterization}, with $\delta$ varies from 0 to 0.8. 
The setup of Figure~\ref{fig:compare_spectral_clustering_under_sparsity} mirrors that of Figure~\ref{fig:compare_nonbacktracking_normalized_adjacency_under_sparsity}, where we see that $H$ and $\Hce$ outperform $\Ace$ under sparsity.
The setup of Figure~\ref{fig:compare_spectral_clustering_under_imbalance} is the same as Figure~\ref{fig:compare_nonbacktracking_normalized_adjacency_under_community_size_imbalance}, and we observe that $\Ace$ and $\Hce$ outperform $H$ under imbalance.
This empirical evidence suggests that $x_1(\Hce)$ is a promising candidate for binary label estimation in challenging scenarios, thanks to the fact that its corresponding informative eigenvalues have a larger chance of being separated from the bulk.

\subsection{Powers of different test statistics}
\label{section:power_of_test_statistics}
Besides the previous statistics based on eigenvalues of $\Ace$, $H$, and $\Hce$, we additionally include some other widely-used statistics in the literature and compare their power.
The likelihood-ratio statistic was considered in \citet{wang2017likelihood}, and here we take a simplified form.
We view the ``correct" label assignment as given so that the label assignment probability is no longer included in the likelihood, and we no longer need to sum over all possible label assignments. 
For the test $H_1: K= K_0+1$ vs. $H_0: K= K_0$, given the label assignment $\bm{g}_{K_0}$ under $H_0$ and $\bm{g}_{K_0+1}$ under $H_1$, we define the likelihood ratio statistic $L_{K_0,K_0+1}$ as follows,
\begin{equation*}
    L_{K_0,K_0+1} = \log \frac{\sup_{\theta \in \Theta_{K_0+1}} f_{K_0+1} (\bm{g}_{K_0+1}, A; \theta)}{\sup_{\theta \in \Theta_{K_0}} f_{K_0} (\bm{g}_{K_0}, A; \theta)},
\end{equation*}
where the likelihood $f_k(\bm{g}, A; \theta)$ is defined as
\begin{equation*}
    f_k(\bm{g}, A; \theta) = \left(\prod_{a=1}^k \prod_{b=1}^k (Q_{a,b}^{(k)}) ^{O_{a,b}(\bm{g})} (1-Q^{(k)}_{a,b}) ^{n_{a,b}(\bm{g}) - O_{a,b}(\bm{g})} \right)^{1/2},
\end{equation*}
where $O_{a,b}(\bm{g}) = \sum_{i=1}^n \sum_{j\ne i}\bm{1}\{g_i=a, g_j=b\} A_{ij}$ and $n_{a,b}(\bm{g}) = \sum_{i=1}^n \sum_{j\ne i}\bm{1}\{g_i=a, g_j=b\}$.
When $k\ge 2$ and $\bm{g}_k$ is not given, we plug in its estimator $\widehat{\bm{g}}_k$.
Then the maximizer $\theta$ is taken as the profile MLE $\widehat{Q}^{(k)}$ given $\widehat{\bm{g}}_k$.

Another candidate we consider as a test statistic is the $(K_0+1)$th smallest eigenvalue of the Bethe-Hessian matrix, namely $\lambda_{n-K_0}(H(r))$.
% We also add the $k+1$-th largest eigenvalue of the adjacency, $\lambda_{k+1}(A)$, when testing against $H_0: K=k$ 
The Bethe-Hessian matrix $H(r)$ is a real symmetric matrix defined as
\begin{equation*}
    H(r) = (r^2 - 1)\In - rA + D
\end{equation*}
where $r\in \mathbb{R}$ is a scale parameter.
Following the choices in \citet{le2022estimating}, we consider $r_a = \sqrt{n^{-1} \sum_i d_i}$ and $r_m = \sqrt{\frac{\sum_{i=1}^n d_i^2}{\sum_{i=1}^n d_i} - 1}$.
% , and we call the corresponding matrix BHa.
We also include the choice $r = \widehat{\zeta}$ proposed by \citet{hwang2023estimation}, as defined in their Algorithm 4.1.
We set the hyperparameters $c=0.3$ and $\widehat{\epsilon}_{\delta}=0.2$ same as in \citet{hwang2023estimation}.

Additionally, we include the $(K_0+1)$th largest eigenvalue of the adjacency $A$ as a reference point.

Finally, we include the triangle count test statistic proposed in \citet{gao2017testing1, gao2017testing2}.
For SBM alternatives, they proposed the following simplified version with normalization:
\begin{equation*}
    \sqrt{\widehat{T}_{\Delta}} - \sqrt{\widehat{p}^3}, \quad \textrm{ where } \widehat{T}_{\Delta} = \left[6 {n \choose 3}\right]^{-1} \mathrm{tr}(A^3)
\end{equation*}
represents the triangle frequency. A similar test is also considered in \citet{verzelen2013community}. Intuitively, we reject the null when observing an unusually high triangle frequency.

% Different from Section~\ref{section:distribution_of_leading_eigenvalues}, for the spectral statistics with centering, here we examine the leading eigenvalues of the empirical versions $\Ane$ and $\Hce$, with $P$ replaced by its estimator $\widehat{P}$.
% The difference between $\An$ and $\Ane$ (or $\Hc$ and $\Hce$) is indeed minimal when $H_0: K=1$ since $\widehat{P}$ recovers $P$ almost perfectly.
% However, when $K>1$ under $H_0$, and labels need to be estimated, the error of $\widehat{P}$ is typically non-negligible under finite $n$, leading to noticeably different behaviors of $\An$ and $\Ane$ (or $\Hc$ and $\Hce$).

% This happens when we set $Q_{11}$ and $Q_{22}$ as in \eqref{eq:P11_P22_constant_degree}, but not when we set them as in \eqref{eq:P11_P22_balanced_P} or \eqref{eq:P11_P22_unbalanced_P}.

% We simulate the power curves given by all aforementioned tests.
For each of the aforementioned test statistics, we establish their one-sided rejection regions using ($1-\alpha$)-quantiles of their null distributions under $P^{(0)}$, simulated through Monte Carlo with 10,000 replicates.
Subsequently, we evaluate the power of each test statistic across various values of $\delta$, ranging from $0$ to $1$ with increment of $0.02$.
For each specific $\delta$ value, we conduct 1000 independent tests and empirically estimate the power by calculating the average rejection rate.
In other words, we generate 1000 independent replicates of the network under $P^{(\delta)}$ and test each statistic against its null distribution. 

Indeed, the tests we are comparing here are their ``oracle" version, as the null distribution under $P^{(0)}$---which we utilize to establish the rejection regions---is generally unknown in practice.
The primary goal here is to assess and compare the sensitivity of each statistic's distribution to the alternative hypothesis.
Nevertheless, in practice, many of these statistics can be reasonably approximated using their asymptotic distributions. For instance, this holds for the triangle count, $\lambda_1(\Ane)$, and our proposed $\mu_1(\Hce)$ in the case of $K_0=1$.
Alternatively, as a practical approach, we can resort to a bootstrap approximation.
Namely, we can use the distribution of the test statistic in networks generated from $\widehat{P}$ to approximate the null distribution under $P^{(0)}$, with the conviction that $\widehat{P}$ closely resembles $P^{(0)}$ under $H_0$. 
While this conviction is especially plausible when $K_0=1$, it is less so when $K_0>1$ and estimating labels $\bm{g}$ is needed.

% We test $H_1: K=2$ vs. $H_0: K=1$ in Figure~\ref{fig:compare_power_K0=1_n0=500_n1=400_n2=100_p0=0.08_balanced_P} and \ref{fig:compare_power_K0=1_n0=500_n1=400_n2=100_equaldegree}.
\subsubsection{Testing $H_1: K>1$ versus $H_0: K=1$}
The setup of Figure~\ref{fig:compare_power_K0=1_n0=500_n1=250_n2=250_p0=0.01_sparsity} mirrors that of Figure~\ref{fig:compare_nonbacktracking_normalized_adjacency_under_sparsity}, emphasizing sparsity. 
We set $n_1 =n_2=250$ and let $Q_{11}=Q_{22}$ grow symmetrically with $\delta$ following \eqref{eq:P11_P22_balanced_P}.
Spectral statistics with or without centering perform similarly in this case.
Meanwhile, both the Bethe-Hessian and the non-backtracking matrix slightly outperform the adjacency matrix. 
Another observation here is that it is not appropriate to use the Tracy-Widom limit in a very sparse setting like this, as doing so results in a Type-I error well above the specified level of $\alpha = 0.05$.
As is also noticed in \citet{bickel2016hypothesis, lei2016goodness}, a bootstrap correction for the Tracy-Widom distribution will be necessary.
Nevertheless, the divergence appears less severe for $\mu_1(\Hce)$, thanks to the partial cancellation effect in Section~\ref{section:constant_degree_regime_null_distribution}.
Interestingly, the triangle count test significantly underperforms the others under this symmetric setting. 

Figure~\ref{fig:compare_power_K0=1_n0=500_n1=400_n2=100_p0=0.08_balanced_P} shares the same setting as in Figure~\ref{fig:compare_nonbacktracking_normalized_adjacency_under_community_size_imbalance}, showcasing imbalance in community sizes. We set $n_1 =400$ and $n_2=100$, and let $Q_{11}=Q_{22}$ grow simultaneously with $\delta$ following \eqref{eq:P11_P22_balanced_P}.
Aligned with our reasoning in Section~\ref{section:centering_enhances_signal}, with unequal per-community expected degree,
operators constructed from $\Ace$ outperform those constructed from $A$, including the Bethe-Hessian matrix.
Under this relatively dense setting, the performance of $\mu_1(\Hce)$ is similar to $\lambda_1(\Ace)$, if using their true null distributions.
However, it is still not appropriate to use the Tracy-Widom limit to set the rejection region for $\lambda_1(\Ace)$, although the same limit fits the distribution of $\mu_1(\Hce)$ perfectly in this case.
In contrast to its lagging performance in Figure~\ref{fig:compare_power_K0=1_n0=500_n1=250_n2=250_p0=0.01_sparsity}, the triangle count test becomes the top performer in this case. The reason, we believe, is that the triangle frequency becomes more sensitive under this setting with unequal expected degrees, as a minority of nodes with higher degrees disproportionately contribute to the overall triangle frequency.
% As pointed out in \citet{verzelen2013community}, the triangle count test always has some nontrivial power in the ultra-sparse regime, which explains why its power curve begins to rise earlier than the KS threshold.

In Figure~\ref{fig:compare_power_K0=1_n0=500_n1=400_n2=100_equaldegree}, we again set $n_1 = 400$ and $n_2=100$, but we keep the expected degree equal in two communities by setting $Q_{11}$ and $Q_{22}$ as \eqref{eq:P11_P22_constant_degree}. 
Operators with or without centering now perform similarly as we expected.
Notably, we again fail to control the Type-I error if we use the Tracy-Widom limit to set the rejection region for $\lambda_1(\Ace)$.
Besides, the performance of the triangle count test once again lags behind the others, as we keep the expected degree equal for each node.

In summary, our proposed $\mu_1(\Hce)$ statistic performs comparably to that of the original non-backtracking matrix or the Bethe-Hessian matrix under sparse settings.
Conversely, its behavior resembles that of $\lambda_1(\Ace)$, in scenarios characterized by imbalance.
The effectiveness of the triangle count test varies significantly based on the level of imbalance in the alternative model.
It demonstrates strong sensitivity when a minority of nodes possess degrees higher than the average; however, it is dramatically less competitive when all nodes exhibit similar degrees.
% We perceive its unstable performance as a weakness, as in practical scenarios, one has little prior knowledge regarding the balance or imbalance of the network.

\begin{figure}
    \centering
    \includegraphics[width=1.0\textwidth]{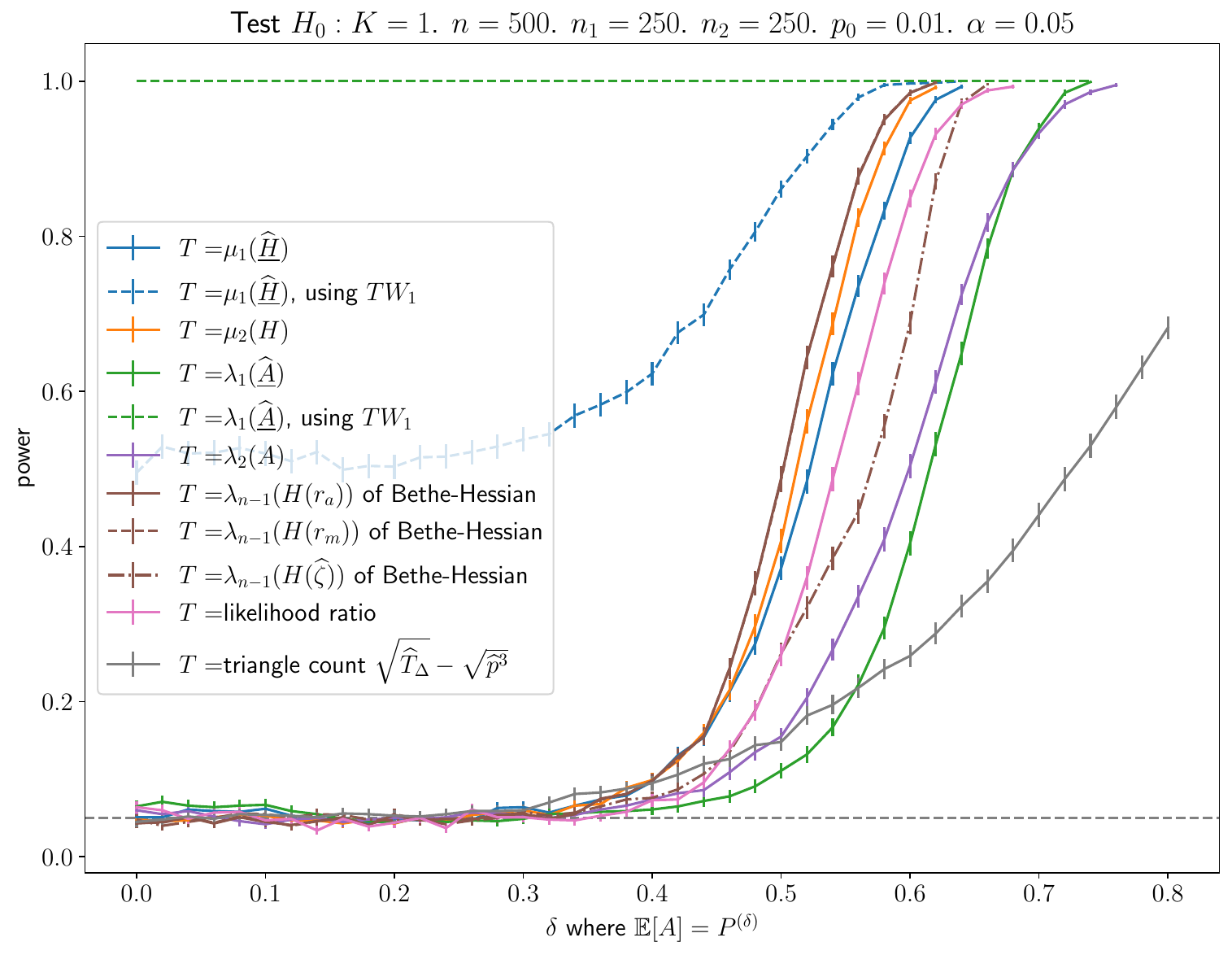}
    \caption{Power curves of different statistics testing $H_0: K=1$.
    Here $n_1 = n_2$ and $Q_{11} = Q_{22}$ grow symmetrically with $\delta$ following \eqref{eq:P11_P22_balanced_P}.
    % The KS threshold corresponding to $\E[A]$ and $\textrm{SNR}$ \eqref{eq:SNR} is larger than the KS threshold corresponding to $\E[A] - P^{(0)}$ and $\textrm{SNR}^*$ \eqref{eq:SNR*}.
    The setting is the same as Figure~\ref{fig:compare_nonbacktracking_normalized_adjacency_under_sparsity}, emphasizing sparsity.
    The error bars represent $\pm$ the standard deviation of the average rejection rate, calculated as $\sqrt{ \widehat{p} (1-\widehat{p}) / 1000}$.
    }\label{fig:compare_power_K0=1_n0=500_n1=250_n2=250_p0=0.01_sparsity}
\end{figure}

\begin{figure}
    \centering
    \includegraphics[width=1.0\textwidth]{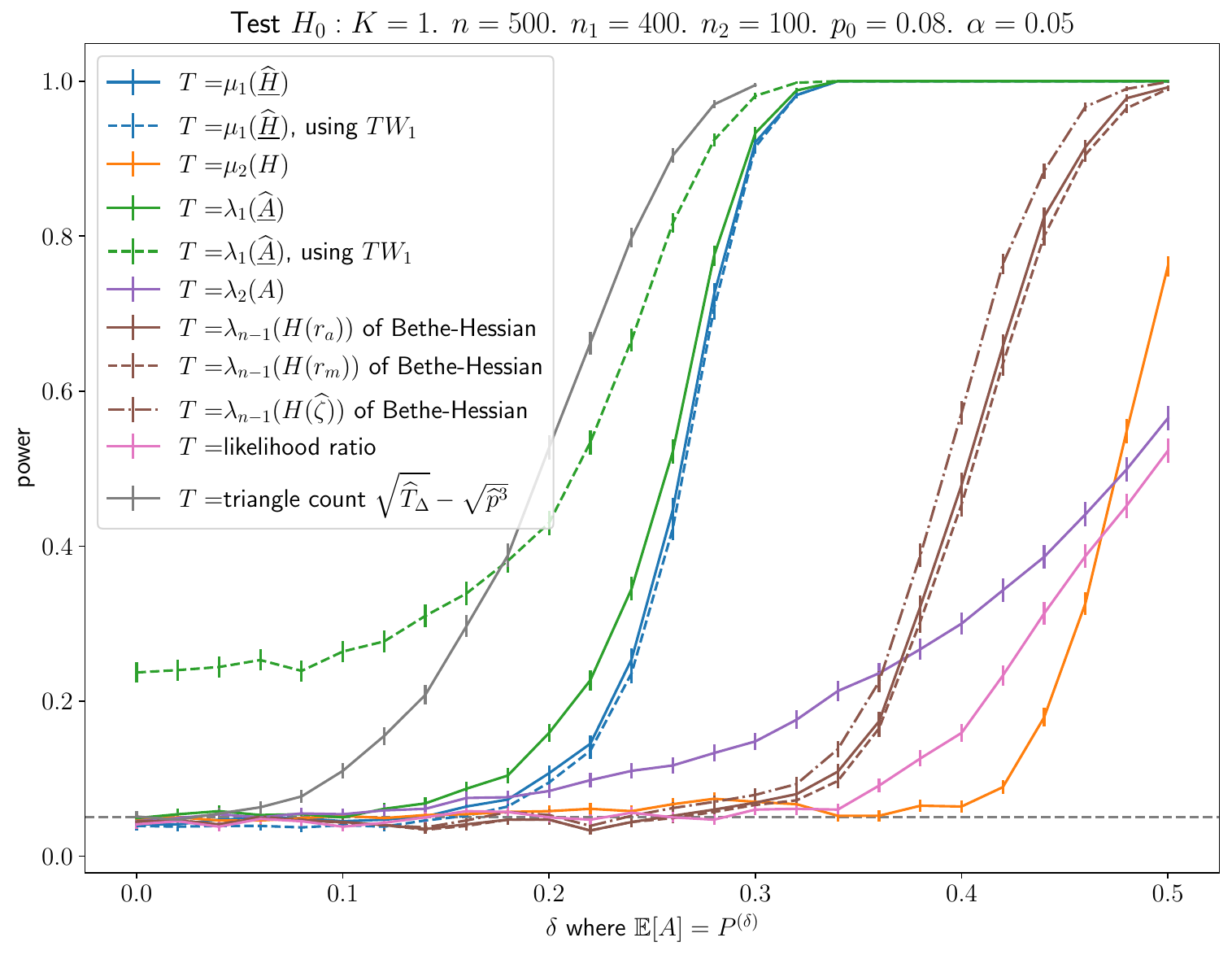}
    \caption{Power curves of different statistics testing $H_0: K=1$.
    Here $n_1\ne n_2$ but $Q_{11} = Q_{22}$ grow simultaneously with $\delta$ following \eqref{eq:P11_P22_balanced_P}.
    The setting is the same as Figure~\ref{fig:compare_nonbacktracking_normalized_adjacency_under_community_size_imbalance}, showcasing imbalance in community sizes.
    % The KS threshold corresponding to $\E[A]$ and $\textrm{SNR}$ \eqref{eq:SNR} is larger than the KS threshold corresponding to $\E[A] - P^{(0)}$ and $\textrm{SNR}^*$ \eqref{eq:SNR*}.
    % The error bars represent $\pm$ the standard deviation of the average rejection rate, calculated as $\sqrt{ \widehat{p} (1-\widehat{p}) / 1000}$.
    }\label{fig:compare_power_K0=1_n0=500_n1=400_n2=100_p0=0.08_balanced_P}
\end{figure}

\begin{figure}
    \centering
    \includegraphics[width=1.0\textwidth]{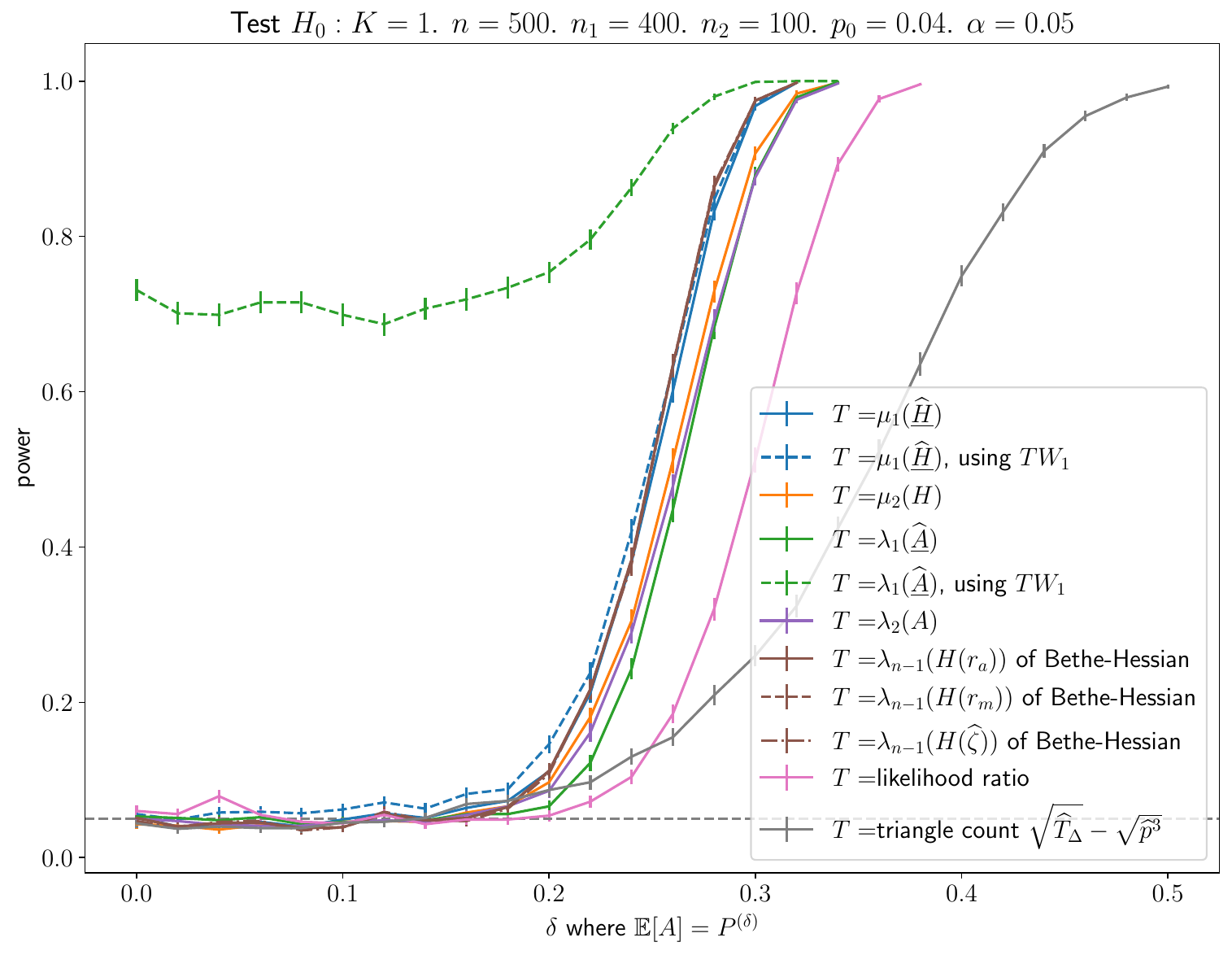}
    \caption{Power curves of different statistics testing $H_0: K=1$.
    Here $n_1\ne n_2$, but we keep the expected degree equal in two communities by setting $Q_{11}$ and $Q_{22}$ as \eqref{eq:P11_P22_constant_degree}. 
    % Both $\textrm{SNR}$ \eqref{eq:SNR} and $\textrm{SNR}^*$ \eqref{eq:SNR*} share the same KS threshold.
    }
    \label{fig:compare_power_K0=1_n0=500_n1=400_n2=100_equaldegree}
\end{figure}

\begin{figure}
    \centering
    \includegraphics[width=1.0\textwidth]{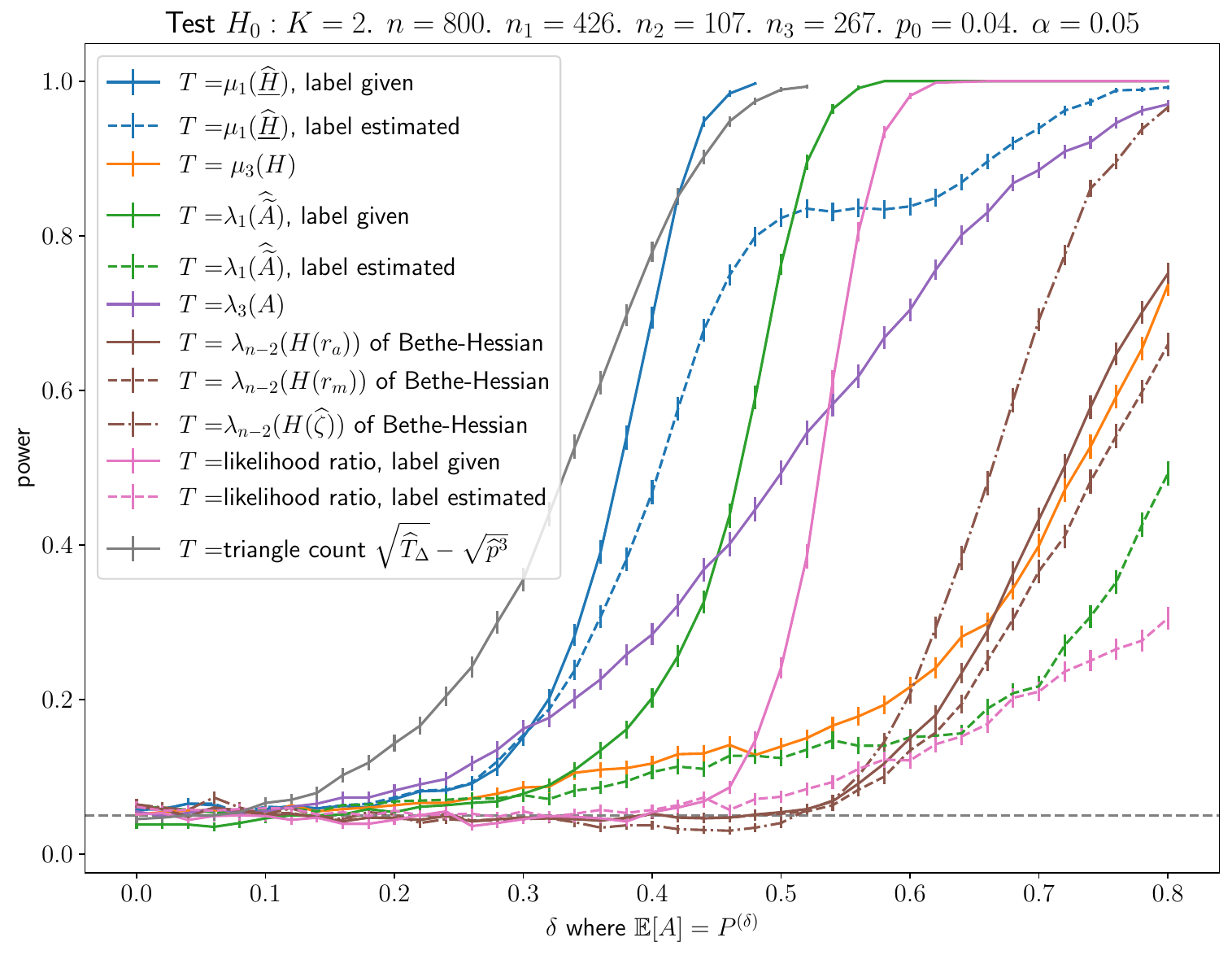}
    \caption{Power curves of different statistics testing $H_0: K=2$, with model under $H_1: K=3$ parameterized by \eqref{eq:model_K=3}.
    Here, we set $n_1=426$, $n_2=107$, $n_3=267$, and let $Q_{11} = Q_{22}$ grow simultaneously with $\delta$ following \eqref{eq:P11_P22_balanced_P}.
    % The largest eigenvalue of normalized non-backtracking, i.e., $\mu_1(\Hc)$, outperforms all other test statistics.
    }
    \label{fig:compare_multiple_statistics_and_methods_power_K0=2_n0=800_n1=267_n2=426_n_3=107_p0=0.04}
\end{figure}

\subsubsection{Testing $H_1: K>2$ versus $H_0: K=2$}
\label{section:Test_H0_K=2}
Although our theoretical analysis in Section~\ref{section:proposed_operator:nonbacktracking_with_centering} mainly focuses on the null hypothesis $H_0: K=1$, 
let us also take a look at the empirical performance of $\mu_1(\Hce)$ when testing $H_0: K=2$.
% We fix $n=800$.
Here, the network generating model under $H_1: K>2$ is an SBM with $K=3$, and block-wise edge probability parameterized as follows
\begin{equation}
    Q^{(\delta)} = \left( \begin{array}{ccc}
        Q_{11}(\delta) & Q_{12}(\delta) & 0.3 p_0 \\
        Q_{12}(\delta) & Q_{22}(\delta) & 0.3 p_0\\
        0.3 p_0 & 0.3 p_0 & p_0
    \end{array} \right), \label{eq:model_K=3}
\end{equation}
where $Q_{12}(\delta) = p_0 (1-\delta)$, and $Q_{11}=Q_{22}$ as in \eqref{eq:P11_P22_balanced_P}.
Like before, when $\delta =0$, this reduces to an SBM with $K=2$.
% The distinction between $P^{(\delta)}$ and a 2-block SBM gets larger when $\delta$ becomes larger.

For statistics that require an estimate of $P$ assuming $K=2$, we provide two versions of $\widehat{P}$: one with known label assignments $\bm{g}$ and the other with estimated labels $\widehat{\bm{g}}$. The former estimator closely approximates $P$, while the latter has a significantly higher error.
% The difference between $\An$ and $\Ane$ (or $\Hc$ and $\Hce$) is indeed minimal when $H_0: K=1$ since $\widehat{P}$ recovers $P$ almost perfectly.
% However, when $K>1$ under $H_0$, and labels need to be estimated, the error of $\widehat{P}$ is typically non-negligible under finite $n$, leading to noticeably different behaviors of $\An$ and $\Ane$ (or $\Hc$ and $\Hce$).
For a fair comparison, we use the same $\widehat{P}$ for the latter case, with $\widehat{\bm{g}}$ estimated by $K$-means on the first two eigenvectors of $A$ with $K=2$.
While this may not represent an optimal estimate of $P$, it enables us to observe the impact of an inaccurate $\widehat{P}$ on methods that rely on it.
% The null distributions are generated using these $\widehat{P}$ with or without correct labels, respectively,
% ensuring that the type-I error is maintained at level $\alpha$ in each case.

Figure~\ref{fig:compare_multiple_statistics_and_methods_power_K0=2_n0=800_n1=267_n2=426_n_3=107_p0=0.04} shows the power of different test statistics in a setting with moderate sparsity and imbalance in community sizes. We set $n_1=426$, $n_2=107$, $n_3=267$, and $p_0=0.4$.
The triangle count statistic is the top performer, again benefiting from the unbalanced expected degrees. 
Also, due to the imbalance, the two operators with accurate centering, $\mu_1(\Hce)$ and $\lambda_1(\Ane)$ using $\widehat{P}$ with given labels $\bm{g}$, significantly outperform the rest. 
Meanwhile, note that $\mu_1(\Hce)$ detects signal earlier than $\lambda_1(\Ane)$ at this level of sparsity.
% , regardless of whether we use $\widehat{P}$ with or without given labels.
For the Bethe-Hessian matrices, the choice $\widehat{\zeta}$ \citep{hwang2023estimation} does outperform $r_a$ and $r_m$ in this setting.
% Among all methods without ``cheating" using given labels, 

% Our proposed statistic $\mu_1(\Hce)$ ranks as the second-best performer.
Remarkably, for both operators $\mu_1(\Hce)$ and $\lambda_1(\Ane)$, centering using the less accurate $\widehat{P}$ with estimated labels $\widehat{\bm{g}}$ does reduce the power to some extent for finite $n$, even though we still anticipate both tests to be asymptotically powerful regardless of the estimator $\widehat{P}$.
Intuitively, a relatively strong signal is needed to recover the unknown labels, which is a practical challenge for all methods that rely on an estimator $\widehat{P}$ when $K \ge 2$.
% Another curious observation is that $\mu_1(\Hce)$ appears less affected by the error in $\widehat{P}$ compared to $\lambda_1(\Ace)$.
% , while the power of $\lambda_1(\Ace)$ 
% and the likelihood-ratio statistic, with labels $\widehat{\bm{g}}$ not accurate enough,
% drops dramatically. 
% The failure of the likelihood ratio statistic is more or less expected, as \citet{wang2017likelihood} requires a very stringent condition on the average degree---it must grow at least $n^{1/2}\log n$ for the convergence of maximum likelihood estimate.
% While the underlying reason for this phenomenon remains unclear, the robustness of $\mu_1(\Hce)$ to such error in $\widehat{P}$ may be viewed as a practical advantage.

\section{Determining $K$ in Practice}
\label{section:determine_K_in_practice}
We first discuss how to apply our goodness-of-fit test to determine $K$ in practice, and then demonstrate its application using a real-world network.

\subsection{Recursive or sequential testing}
\label{section:sequential_or_recursive_testing}
There are two main approaches for determining $K$ based on a goodness-of-fit test.
The first is through recursive bi-partitioning, which is used in \citet{bickel2016hypothesis}.
For a comprehensive overview of this type of approach, refer to \citet{li2022hierarchical} for a general framework.
Starting with the network $A$, we conduct a goodness-of-fit test for an Erd\H{o}s–R\'{e}nyi model to test $H_1: K>1$ versus $H_0: K=1$. If we reject $H_0$, we then bi-partition the network $A$ into two sub-networks, for instance, using spectral clustering with the leading eigenvector of $\Hce$, and subsequently repeat the same test recursively on these two sub-networks.
The algorithm concludes when $H_0$ is not rejected for any sub-network,
and naturally yields a hierarchical clustering structure in the end.

The second approach is similar to \citet{lei2016goodness}, where one performs the goodness-of-fit test for $K = 1, 2, \ldots$, until failing to reject $H_0$. Formally, an estimate of $K$ is given by
\begin{equation*}
    \widehat{K} = \inf\{K_0 \ge 1: \mu_1(\Hce_{(K_0)}) < t_{n, K_0}\},
\end{equation*}
where $t_{n,K_0}$ is some threshold based on the null distribution of $\mu_1(\Hce_{(K_0)})$.

Like in \citet{bickel2016hypothesis}, the recursive testing approach is based on the heuristic that each round of bi-partitioning can split the network into two subnetworks, each of which is a disjoint union of distinct sets of the true blocks with high probability.
Apparently, there is no guarantee for that.
Nevertheless, we still recommend recursive bi-partitioning over sequential testing, especially due to the challenge of determining $t_{n, K_0}$ for $K_0>1$ in practice.
Firstly, the theoretical analysis in Section~\ref{section:proposed_operator:nonbacktracking_with_centering} primarily focuses on testing $H_0:K=1$. Crucially, it is convenient and efficient to use the Tracy-Widom distribution as the null distribution, but this limiting distribution holds for $\mu_1(\Hce)$ only when $A \sim G(n,p)$ and $P_{i,j} = p$, i.e., under $H_0: K=1$.
Secondly, the task of estimating $P$ becomes more challenging for $K_0>1$ due to the need to estimate $\bm{g}$, introducing potentially large error in $\widehat{P}$. 
Remember that we have to rely on the bootstrap distribution to determine $t_{n,K_0}$ for $K_0>1$ in practice.
The large error in $\widehat{P}$ may compromise the type-I error control of the test, since the bootstrap distribution of the test statistic under $\widehat{P}$ may significantly deviate from its true null distribution.
% Although similar to \eqref{eq:asymptotic_order_lambda_1_Ane}, we still anticipate the test to be asymptotically powerful regardless of the error,
Moreover, the error in $\widehat{P}$ may also impair the power to some extent for finite $n$, as evident from our observations in Figure~\ref{fig:compare_multiple_statistics_and_methods_power_K0=2_n0=800_n1=267_n2=426_n_3=107_p0=0.04}.

Lastly, the hierarchical clustering structure, resulting from recursive bi-partitioning, is valuable in practice for its interpretability.
As mentioned in \citet{li2022hierarchical}, it can also be more accurate than $K$-way clustering in certain regimes.

\subsection{The political blog data}
In line with \citet{lei2016goodness}, we employ the political blog data \citep{adamic2005political} as a real-world example to elucidate the determination of $K$.
This dataset captures hyperlinks between web blogs shortly before the 2004 US presidential election and has been extensively utilized in the network community detection literature.
Following common practice, we consider the largest connected component,  comprising 1222 nodes categorized into two political orientations with sizes 586 and 636, respectively.

For recursive testing, we utilize $\lambda_1(\Ace)$ and $\mu_1(\Hce)$, setting the Type-I error rate at 0.001 for each recursive test.
Critical values are obtained from the quantiles of the Tracy-Widom distribution and the simulated null distribution, respectively.
We halt further recursions if subsequent subnetworks have a size smaller than 20. 
Additionally, we also provide an estimate $\widehat{K}$ using the original non-backtracking matrix $H$, by simply counting the number of real nontrivial eigenvalues outside the bulk (excluding the largest).

All the estimated $\widehat{K}$ are listed in Table~\ref{Table:estimated_K_political_blog_data}.
Counting the eigenvalues of $H$ results in the smallest $\widehat{K}$. In the case of $\lambda_1(\Ace)$ and $\mu_1(\Hce)$, utilizing simulated critical values yields a smaller $\widehat{K},$ given that these critical values are significantly larger than the Tracy-Widom quantile when $n$ is not sufficiently large.

Lastly, during recursive testing with $\mu_1(\Hce)$, we partitioned the network into 14 groups, using the leading eigenvector of $\Hce$.
The confusion matrix, which compares the final estimated label assignment with the true label (indicating political orientation), is presented in Table~\ref{Table:confusion_matrix}.
Similar to the findings in \citet{lei2016goodness}, 13 of these estimated groups predominantly comprise nodes from a single true community. The additional estimated group includes nodes with very small degrees, making the recovery of their community memberships difficult.

\begin{table}[ht]
    \centering
    \caption{Estimated $\widehat{K}$ for political blog data}
    \label{Table:estimated_K_political_blog_data}
    \begin{tabular}{|c|c|c|c|c|}
    \hline
         $\lambda_1(\Ace)$ ($TW_1$) & $\lambda_1(\Ace)$ (simulated) & $\mu_1(\Hce)$ ($TW_1$) & $\mu_1(\Hce)$ (simulated) & Counting $\mu_k(H)$\\\hline
         19 & 17 & 15 & 14 & 8 \\\hline
    \end{tabular}
\end{table}

\begin{table}[ht]
    \centering
    \caption{Confusion matrix of estimated labels by recursive spectral clustering using $x_1(\Hce)$}
    \label{Table:confusion_matrix}
    \begin{tabular}{|c|c|c|c|c|c|c|c|c|c|c|c|c|c|c|}
    \hline
        \diagbox{Truth}{Estimate} & 0 & 1 & 2 & 3 & 4 & 5 & 6 & 7 & 8 & 9 & 10 & 11 & 12 & 13\\\hline
         0 & 76 & 3 & 0 & 3 & 0 & 55 & 29 & 78 & 40 & 97 & 21 & 176 & 4 & 4\\\hline
         1 & 18 & 42 & 34 & 86 & 25 & 2 & 0 & 0 & 0 & 5 & 1 & 124 & 255 & 44 \\\hline
    \end{tabular}
\end{table}

\section{Discussions}
\label{section:discussions}
\subsection{Connection to non-backtracking walks}
Recall that for the original non-backtracking spectrum operator $H$ \eqref{eq:definition_of_H}, its spectrum is equivalent to the $2m \times 2m$ non-backtracking matrix $B$ \eqref{eq:conventional_nonbacktracking}.
Similarly, for $\Hc$, there exists a matrix $\Bc$ indexed by directed edges whose spectrum is equivalent to that of $\Hc$.

Let $\vec{E}(V) = \{(u, v): u \ne v \in V\}$ be the set of directed edges of the \textit{complete graph} on $n$ nodes. Then $|\vec{E}(V)| = n(n-1)$.
Construct an $|\vec{E}(V)| \times |\vec{E}(V)|$ matrix $\Bc$ from $\Ac$,
\begin{equation*}
    \Bc_{v \to u, w \to z} = \left \{ \begin{array}{ll}
      0   & \textrm{if } w \ne u, \\
      \Ac_{wz}   & \textrm{if } w= u \textrm{ and } z \ne v \\
      \Ac_{wz} - 1 & \textrm{if } w= u \textrm{ and } z = v 
    \end{array} \right..
    % \label{eq:definition_of_B_centered}
\end{equation*}
Similar to a conventional non-backtracking matrix, summing over incoming and outgoing edges relates $\Bc$’s spectrum to that of $\Hc$ \eqref{eq:definition_of_Hc}.
See Appendix~\ref{section:proof_of_proposition_spectrum_of_B_equivalent_to_spectrum_of_H} for details of the proof.
\begin{proposition}[Spectral equivalence of $\Bc$ and $\Hc$]
\label{proposition:spectrum_of_B_equivalent_to_spectrum_of_H}
     The spectrum of $\Bc$ is the set $\{\pm 1\} \cup \{\mu: \det(\mu^2 \In - \mu \Ac + \Dc - \In) = 0\}$, or equivalently, the set $\{\pm 1\} \cup \{\textrm{eigenvalues of }\Hc\}$.
\end{proposition}

For two given edges $e$ and $f$, the $(e,f)$ entry of $\Bc^{k}$ sums up productions of entries, along all paths from $e$ to $f$ of length $(k+1)$.
Because there are many more nonzero entries in $\Bc$ compared to $B$, powers of $\Bc$ do not only count non-backtracking paths.
Consequently, one cannot similarly show the weak Ramanujan property for $\Bc$ using combinatorial arguments from \cite{bordenave2015non}.
Nevertheless, the ``-1" term when $f_2 = e_1$ should intuitively reduce the contribution of paths that contain backtracking edges, partially offsetting the influence of high-degree nodes in a sparse network.

% \section{Conclusions and Future Directions}
\subsection{Future work directions}
% For future directions, several avenues present themselves.

Firstly, validating the Tracy-Widom limit of our proposed operator through a rigorous proof of Conjecture~\ref{conjecture:concentration_of_v1Dv1} stands as a critical future step. The significance of this conjecture may extend beyond our context, possibly attracting independent interest.

Secondly, our study has primarily focused on the SBM and the homogeneous Erd\H{o}s-R\'{e}nyi model as the null and alternative hypotheses. 
An intriguing prospect lies in extending our proposed method to DCSBMs, thus accommodating degree heterogeneity.
One major challenge here is that we will need conditions on the degree parameters, as discussed in \cite{lei2016goodness}.
For example, there must exist a community within which the degree parameters cannot be approximated by block-wise constant vectors.
Besides, the error in $\widehat{P}$ will be considerably larger for DCSBM estimators, complicating the tractability of the null distribution of $\Ace$ and $\Hce$.
While we anticipate that our testing approach will maintain asymptotic power against DCSBM null models akin to \cite{lei2016goodness}, controlling the type-I error poses a considerable challenge in this scenario.

Finally, it is interesting to see if the concept of ``centering" can be similarly integrated into the Bethe-Hessian matrix.
However, a significant hurdle remains in addressing the uncertainty around choosing the parameter $r$.

\bibliographystyle{apalike}
\bibliography{ref}

\newpage
\appendix

% \numberwithin{equation}{section}
% \numberwithin{theorem}{section}

\section{Proof of Proposition~\ref{proposition:convergence_of_H_to_TW1}: Tracy-Widom limit of $\mu_1(\widetilde{H})$}
\label{section:proof_of_proposition_tracy_widom_limit}
% Note that when $\alpha = \Omega (n^{2/3 + \epsilon})$, we have 
% \begin{equation*}
%      n^{2/3} \left( \frac{\lambda_1(\Ac)}{\sqrt{(n-1)p (1-p)}}  - 2\right) = n^{2/3} \left( \frac{\lambda_1(\Ac)}{\sqrt{\alpha (1-p)}} \sqrt{\frac{\alpha}{\alpha-1}} - 2\right) \xrightarrow{d} TW_1,
% \end{equation*}
% according to \citet{erdHos2012spectral, lee2014necessary}.
\begin{proof}[Proof of Proposition~\ref{proposition:convergence_of_H_to_TW1}]
We will only prove the convergence results \eqref{eq:difference_between_mu_1_H_and_lambda_1_A_general} and \eqref{eq:difference_between_mu_1_H_and_lambda_1_A} for $\mu_1(\Hc)$ and $\lambda_1(\Ac)$.
To show similar convergence results \eqref{eq:difference_between_mu_1_Hce_and_lambda_1_Ace_general} and \eqref{eq:difference_between_mu_1_Hce_and_lambda_1_Ace}  for their empirical counterparts $\mu_1(\Hce)$ and $\lambda_1(\Ace)$, one only needs to replace $\Ac$, $\Dc$, and $p$ by $\Ace$, $\Dce$ and $\widehat{p}$ in the following arguments.

% The proof of \eqref{eq:difference_between_mu_1_H_and_lambda_1_A_general} is based on the Bauer-Fike theorem.
% \begin{theorem}[Bauer-Fike Theorem \citep{bhatia2013matrix}]
%     If $\widetilde{H}_0$ is diagonalizable by some matrix $S$, then for any eigenvalue $\mu$ of $\widetilde{H}_0 + E$, there exists some eigenvalue $\lambda$ of $\widetilde{H}_0$ such that 
%     \begin{equation*}
%         |\mu- \lambda| \le \|S\| \|S^{-1}\| \|E\|.
%     \end{equation*}
% \end{theorem} 
% Denote the eigenvalue decomposition of $\frac{1}{\sqrt{\alpha}}\Ac$ as $\frac{1}{\sqrt{\alpha}}\Ac = V\Lambda V^\top$.
% One can verify that $\widetilde{H}_0$ is diagonalizable by the $2n \times 2n$ matrix $S$:
% \begin{equation*}
%     S = \left(\begin{array}{cc}
%         V & \bm{0}\\
%         V \Lambda^{-1} & \In
%     \end{array}\right), 
% \end{equation*}
% which satisfies
% \begin{equation}
% S^{-1} = \left(\begin{array}{cc}
%         V^\top & \bm{0}\\
%         -V\Lambda^{-1} V^\top & \In
%     \end{array}\right),
%     \quad 
%     S^{-1} \widetilde{H}_0 S = \left(\begin{array}{cc}
%         \Lambda & \bm{0}\\
%         \bm{0} & \bm{0}
%     \end{array}\right).
% \end{equation}
% Then, $\|S\|= \|S^{-1}\|= \|V\| = 1$ \textcolor{red}{Why? $\Lambda$ has eigenvalues close to zero}. Therefore, by the Bauer-Fike Theorem,

The proof relies on a perturbation analysis based on the following theorem.
\begin{theorem}[Theorem 4.4 in \citet{demmel1997applied}]
\label{theorem:eigenvalue_perturbation}
    Let $\lambda$ be a simple eigenvalue of $A$ with right eigenvector $x$ and left eigenvector $y$.
    Let $\lambda + \delta \lambda$ be the corresponding eigenvalue of $A + E$.
    Then
    \begin{equation*}
        \delta \lambda = \frac{y^* E x}{y^* x} + O(\|E\|^2 ).
    \end{equation*}
\end{theorem}

Here we consider the perturbation $\widetilde{H} = \widetilde{H}_0 + E$.
The right eigenvector $x_i$ and left eigenvector $y_i$ corresponding to the nontrivial eigenvalue $\lambdac_i$ of $\widetilde{H}_0$ are given by 
\begin{equation*}
    y_i^* = (\vc_i^\top,\, \bm{0}), \quad x_i = \left( \begin{array}{c}
        \vc_i  \\
        \lambdac_i^{-1} \vc_i
    \end{array} \right), 
    % \label{eq:left_and_right_eigenvectors_of_H0}
\end{equation*}
where $\vc_i$ is the eigenvector of $\frac{1}{\sqrt{\alpha}} \Ac$ corresponding to $\lambdac_i$.
Clearly, $y_1^* x_1 = 1$, and \begin{equation*}
    y_1^* E x_1 = \lambdac_1^{-1} (1- \alpha^{-1} \vc_1^\top D \vc_1)
    = \alpha^{-1} \lambdac_1^{-1} (1 - \vc_1^\top \Dc \ \vc_1).
\end{equation*}
Therefore, if we consider the difference of their leading eigenvalues, we have
\begin{equation}
    |\mu_1(\widetilde{H}) - \mu_1(\widetilde{H}_0)| \le |\alpha^{-1} \lambdac_1^{-1} (1 - \vc_1^\top \Dc \ \vc_1)| + O(\|E\|^2).
    \label{eq:difference_between_mu_1_H_and_mu_1_H0}
\end{equation}
In particular, note that
\begin{equation*}
     |1- \alpha^{-1} \vc_1^\top D \vc_1| \le  \max_{1\le i \le n}\{|\alpha^{-1} d_i - 1|\} = \|E\|.
 \end{equation*}
Therefore,
\begin{equation*}
    |\mu_1(\widetilde{H}) - \mu_1(\widetilde{H}_0)| \le \lambdac_1^{-1} \|E\| + O(\|E\|^2).
\end{equation*}

We first show \eqref{eq:difference_between_mu_1_H_and_lambda_1_A_general} with the assumption $\alpha / \log n \to \infty$.
By Lemma 3.5 in \citet{wang2023limiting},
\begin{equation*}
    \|E\| =  \max_{1\le i \le n}\{|\alpha^{-1} d_i - 1|\}  = O_{\Prob}(\sqrt{\log n} \alpha^{-1/2}).
\end{equation*}
Moreover, when $\alpha / \log n \to \infty$, 
the largest eigenvalue of $\Ac$ is concentrated \citep{benaych2020spectral}, i.e., $\lambdac_1 = 2 +  o_{\Prob}(1)$.
This concludes the proof of \eqref{eq:difference_between_mu_1_H_and_lambda_1_A_general}.

Now we turn to the assertion \eqref{eq:difference_between_mu_1_H_and_lambda_1_A}.
Since $\|E\|^2 = o_{\Prob}(n^{-2/3})$ and $\lambdac_1 = 2 + o_{\Prob}(1)$ when $\alpha = \Omega(n^{2/3 + \epsilon})$, by \eqref{eq:difference_between_mu_1_H_and_mu_1_H0} we are only left to show that 
$\alpha^{-1} |\vc_1^\top \Dc \ \vc_1| = o_{\Prob}(n^{-2/3})$. 
This is a consequence of Conjecture~\ref{conjecture:concentration_of_v1Dv1}.
\end{proof}

% \subsection{Difference between the spectrum of $\Hc_0$ and $\Hc$}

\section{Proof of Proposition~\ref{proposition:order_of_vDv_constant_degree}: Order of $\vc_1^\top D \vc_1$ under ultra-sparse regime}
\label{section:proof_of_proposition_order_of_vDv_constant_degree}
\begin{proof}[Proof of Proposition \ref{proposition:order_of_vDv_constant_degree}]
Recall that we denote $v_1$ as the eigenvector of the original adjacency $A$ corresponding to its largest eigenvalue, and $\underline{v}_1$ similarly as the top eigenvector of $\Ac$.
We first show that the same lower bound holds for $v_1^\top D v_1$. Then we argue that $\underline{v}_1$ is very close to $v_1$ under the constant degree regime so that $\underline{v}_1 ^\top D \underline{v}_1 \approx v_1^\top D v_1$.

% For $k \in [n]$ we set 
% \begin{equation*}
%     L_k \coloneqq \frac{\log(n /k)}{ \log ((\log n) / d)}
% \end{equation*}
Denote by $d_1^{\downarrow} \ge \ldots \ge d_n^{\downarrow}$ the decreasingly ordered degrees of $d_1, \ldots , d_n$.
The order of the largest eigenvalue of $A$ is shown in \citet{krivelevich2003largest, benaych2019largest}:
\begin{equation*}
     \lambda_1(A) \asymp \sqrt{d_1^{\downarrow}}, \textrm{ where } d_1^{\downarrow} \asymp \frac{\log n}{\log (\log n)}.
\end{equation*}
The same result holds with $A$ replaced by $\Ac$.
Meanwhile, Theorem 1.4 of \citet{hiesmayr2023spectral} shows that the eigenvector $v_1$ corresponding to the largest eigenvalue of $A$ is localized around some node $x \in [n]$ with degree $d_1^{\downarrow}$. In particular,
\begin{equation*}
    |v_{1x}| = \frac{1}{\sqrt{2}} + o_{\Prob}(1).
\end{equation*}
Therefore,
\begin{equation*}
    v_1^\top D v_1 = \sum_{i=1}^n v_{1i}^2 d_i \ge v_{1x}^2 d_1^{\downarrow} = \frac{1}{2} d_1^{\downarrow} (1 + o_{\Prob}(1)).
\end{equation*}

Let $\delta \coloneqq v_1 - \underline{v}_1 $. The difference between $v_1^\top D v_1$ and $\underline{v}_1^\top D \underline{v}_1$ is negligible if $\|\delta\| = o_{\Prob}(1)$:
\begin{equation*}
    |v_1^\top D v_1 - \underline{v}_1^\top D \underline{v}_1| = |-2 \delta^\top D v_1 + \delta^\top D \delta| = o_{\Prob}(d_1^{\downarrow}).
\end{equation*}
It only remains to show that $\underline{v}_1$ is close to $v_1$.
In particular, we show $\langle v_1, \underline{v}_1 \rangle = 1 - o_{\Prob}(1)$.
By Lemma 5.2 of \citet{demmel1997applied},
\begin{equation*}
    \underline{v}_1 = c \sum_{i=1}^n \frac{v_i^\top \bm{1}_n }{\lambda_i - \underline{\lambda}_1} v_i,
\end{equation*}
where $c$ is a normalizing constant so that $\|\underline{v}\| = 1$.
Since all coordinates of $v_1$ can be taken positive by the Perron–Frobenius theorem, we have $v_1^\top \bm{1}_n = \|v_{1}\|_1$.
% \ge \frac{1}{\sqrt{2}} + o_{\Prob}(1)$. 
We claim that $\|v_1\|_1 = O_{\Prob}(\sqrt{n}  (\log n)^{-1.9} )$, and leave its proof to the end.

Recall that $\Ac = A - \frac{d}{n} \bm{1}_n\bm{1}_n^\top$, ignoring the negligible difference of diagonal elements. We have $\underline{\lambda}_1 \le \lambda_1$, and
\begin{equation*}
    \underline{\lambda}_1 = \max_{x \in \mathbb{R}^n, \|x\|=1} x^\top ( A - \frac{d}{n} \bm{1}_n \bm{1}_n^\top) x \ge \lambda_1 - \frac{d}{n} (v_1^\top\bm{1}_n)^2
\end{equation*}
Therefore, $|\lambda_1 - \underline{\lambda}_1| \le \frac{d}{n} (v_1^\top\bm{1}_n)^2 = \frac{d}{n} \|v_1\|_1^2$, and with high probability,
\begin{equation*}
    \frac{(v_1^\top \bm{1}_n)^2}{ |\lambda_1 - \underline{\lambda}_1|^2} \ge  \frac{n^2}{d^2} \|v_1\|_1^{-2} = \Omega(n (\log n)^{3.8}).
\end{equation*}

On the other hand, we have the lower bound for the spacing of the largest eigenvalues of $A$ (Lemma 7.5 of \citet{hiesmayr2023spectral})
\begin{equation*}
    |\lambda_i - \lambda_1| = \Omega\left( \left(\frac{\log n}{\log \log n} \log^3 \left(\frac{\log n}{\log \log n}\right) 2^{3 \log \log \log n}\right)^{-1}\right)
    = \Omega ((\log n)^{-1.5}),
\end{equation*}
which holds with high probability. This implies that, for $i \ne 1$,
\begin{equation*}
    |\lambda_i - \underline{\lambda}_1| \ge ||\lambda_i - \lambda_1| - |\lambda_1 - \underline{\lambda}_1|| = \Omega ((\log n)^{-3}).
\end{equation*}
% Since $\{v_i\}_{i=1}^n$ span the entire $\mathbb{R}^n$, we have 
Recall the identity $\sum_{i=1}^n (v_i^\top \bm{1}_n)^2 = \|\bm{1}_n\|^2 = n$.
The remainder of the projection of $\underline{v}_1$ onto $v_1$ is negligible,
\begin{equation*}
    \sum_{i=2}^n \frac{(v_i^\top \bm{1}_n)^2}{(\lambda_i - \underline{\lambda}_1)^2} \le n (\log n)^3\ll \frac{(v_1^\top \bm{1}_n)^2}{ |\lambda_1 - \underline{\lambda}_1|^2}. 
\end{equation*}

Now it only remains to show our previous claim: $\|v_1\|_1 = O_{\Prob}(\sqrt{n}  (\log n)^{-1.9} )$. Intuitively, the weight of the eigenvector decays exponentially around that node $x$ with degree $d_1^{\downarrow}$.
Formally, for $i\ge 0$, we denote by $B_i(x)$ the ball of radius $i$ around $x$, rooted at $x$. Moreover, we define $S_i(x)$ to be the set of nodes $y$ such that the shortest path from $x$ to $y$ is of length $i$. Then we decompose $\|v_1\|$ as follows
\begin{eqnarray*}
    \|v_1\|_1 & = & |v_1\vert_{x}| + \|v_1 \vert_{S_1(x)}\|_1 + \|v_1 \vert_{S_2(x)}\|_1 +  \|v_1 \vert_{S_3(x)}\|_1 + \|v_1 \vert_{S_4(x)}\|_1 + \|v_1 \vert_{[n] \setminus B_4(x)}\|_1 \\
    & \le & |v_1\vert_{x}| + |S_1(x)|^{1/2} \|v_1 \vert_{S_1(x)}\| + |S_2(x)|^{1/2} \|v_1 \vert_{S_2(x)}\|+ |S_3(x)|^{1/2} \|v_1 \vert_{S_3(x)}\| \\
    && + |S_4(x)|^{1/2} \|v_1 \vert_{S_4(x)}\| + n^{1/2} \|v_1 \vert_{[n] \setminus B_4(x)}\| \\
    & \le & 1 + \sum_{i=1}^4|S_i(x)|^{1/2} + n^{1/2} \|v_1 \vert_{[n] \setminus B_4(x)}\|.
\end{eqnarray*}
By Lemma 4.2 of \citet{hiesmayr2023spectral}, the size of the neighborhood is bounded. For each $1 \le i \le 4$,
\begin{equation*}
    |S_i| = O_{\Prob}(d_1^{\downarrow}) =  O_{\Prob}\left(\frac{\log n}{\log \log n}\right).
\end{equation*}
On the other hand, we apply the bound in Theorem 1.4 of \citet{hiesmayr2023spectral} for $\|v_1 \vert_{[n] \setminus B_4(x)}\|$,
\begin{equation*}
    \|v_1 \vert_{[n] \setminus B_4(x)}\| = O_{\Prob}\left(\left(d/d_1^{\downarrow}\right)^2\right) = O_{\Prob}\left(\left(\frac{\log n}{\log \log n}\right)^{-2}\right) = O_{\Prob}\left(\left(\log n\right)^{-1.9}\right).
\end{equation*}
Combining the two bounds above gives the required bound for $\|v_1\|_1$.
\end{proof}

\section{Proof of Theorem~\ref{theorem:growth_rate_of_lambda_1_residual_under_alternative}: 
Growth rate of $\lambda_1(\Ace)$}
% extension of Theorem 3.3 of \citet{lei2016goodness}}
\label{section:proof_of_theorem_growth_rate_of_lambda_1_residual_under_alternative}
\begin{proof}[Proof of Theorem~\ref{theorem:growth_rate_of_lambda_1_residual_under_alternative}]
Note that $((n-1) \widehat{p} (1-\widehat{p}))^{-1/2}\|\Ace \|  \asymp \alpha^{-1/2} \|A -\widehat{P}^{(0)} \|$, where $\widehat{P}^{(0)}$ is a constant matrix in which all of its off-diagonal entries equal to $\widehat{p}$ and all diagonal entries being 0. 
By the triangle inequality,
\begin{eqnarray*}
    \alpha^{-1/2} \|A -\widehat{P}^{(0)} \| \ge \alpha^{-1/2} (\| \E[A] - \widehat{P}^{(0)} \| - \|A -\E[A] \|).
\end{eqnarray*}

On one hand, notice that $(\E[A] - \widehat{P}^{(0)})$ is a block-wise constant matrix with entries $Q_{ij} - \widehat{p}$.
Any of these blocks has $\Theta(n)$ rows and $\Theta(n)$ columns, 
and moreover, at least one of them has entries with absolute value greater than $\frac{\delta_n}{ 2}$.
% since $\max\{|Q_{12} -\widehat{p}|, |Q_{11} - \widehat{p}|, |Q_{22} - \widehat{p}|\} > \frac{\delta_n}{ 2}$.
Otherwise, if we have $|Q_{ij} - \widehat{p} |<\frac{\delta_n}{2}$ for all $i$, $j$, it will contradict with the assumption that $\max_{1\le i \le K, \, 1\le j \le K} |Q_{ij} - Q_{ij'}| = \delta_n$.
Therefore, since the 2-norm of a matrix is always no smaller than the 2-norm of its submatrix,
\begin{eqnarray*}
    \|\E[A] -\widehat{P}^{(0)} \| = \Omega(\delta_n n).
\end{eqnarray*}

On the other hand, we use the fact that $A$ is concentrated.
Denote $d_m \coloneqq \max_i \E[d_{i}]$.
Due to the assumption that $(\min_{1\le i \le K}\frac{n_i}{n}) \asymp 1$, we have $d_m / \alpha = O(1)$. Since $d_m /\log n\to \infty$, by \citet{benaych2020spectral}, 
\begin{equation*}
    \|A - \E[A]\| = (2+ o_{\Prob} (1)) \sqrt{d_m},
\end{equation*}
thus showing $\alpha^{-1/2} \|A - \E[A]\| = O_{\Prob} (1)$.
\end{proof}

\section{Proof of Proposition~\ref{proposition:difference_between_mu_1_H_and_lambda_1_alternative}: Difference between $\mu_1(\Hce)$ and $\lambda_1(\Ace)$ under alternative hypothesis}
\label{section:proof_of_proposition_difference_between_mu_1_H_and_lambda_1_alternative}
\begin{proof}[Proof of Proposition~\ref{proposition:difference_between_mu_1_H_and_lambda_1_alternative}]
We first show that
\begin{equation*}
    \alpha^{-1/2} |\mu_1(\Hc) - \lambda_1(\Ac)| = 
    O_{\Prob} (\delta_n n \alpha^{-1})
    % + O_{\Prob} \left(\alpha^{-1} n^{1/2} \sqrt{\log n}\right).
    + O_{\Prob} (1)
    = O_{\Prob} (1).
    % + O_{\Prob} (\alpha^{-1/2} (\log n)^{\frac{1}{2}(1 + \epsilon')}).
\end{equation*}
The empirical version counterpart can be shown in a similar manner.

Recall the definition of $\widetilde{H}$:
\begin{equation*}
    \widetilde{H} = \left(\begin{array}{cc}
       \frac{1}{\sqrt{\alpha}}  \Ac  &   \alpha^{-1} (\In - \Dc) \\
       \In  &  \bm{0}
    \end{array} \right) = \left(\begin{array}{cc}
       \frac{1}{\sqrt{\alpha}}  \Ac  &  \In - \alpha^{-1} D \\
       \In  &  \bm{0}
    \end{array} \right) , \label{eq:definition_of_tildeH}
\end{equation*}
and the signal-plus-noise decomposition $\widetilde{H} = \widetilde{H}_0 + E$. Since $\mu_1(\widetilde{H}) = \alpha^{-1/2} \mu_1(\Hc)$ and $\mu_1(\widetilde{H}_0) =  \alpha^{-1/2} \lambda_1(\Ac)$, 
% this is equivalent to showing
% \begin{equation*}
%   |\mu_1(\widetilde{H}) - \mu_1(\widetilde{H}_0)| = O_{\Prob} (\delta_n n \alpha^{-1}).
% \end{equation*}
the left-hand side equals
\begin{equation*}
    |\mu_1(\widetilde{H}) - \mu_1(\widetilde{H}_0)|.
\end{equation*}
We can similarly employ Theorem~\ref{theorem:eigenvalue_perturbation}, as demonstrated in Appendix~\ref{section:proof_of_proposition_tracy_widom_limit}, to show
\begin{equation*}
    |\mu_1(\widetilde{H}) - \mu_1(\widetilde{H}_0)| \le \lambdac_1^{-1}\|E\| + O(\|E\|^2).
\end{equation*}
where $\lambdac_1 = \lambda_1\left(\frac{\Ac}{\sqrt{\alpha}}\right) = \Omega_{\Prob}(1)$ \citep{benaych2020spectral}.
The remaining task is to bound $\|E\|$ under the alternative hypothesis.
We will next show $\|E\| = O_{\Prob}(1)$.

% We are now left to show that 
% \begin{equation*}
%     \|E\| = \|\In - \alpha^{-1} D\| = \max_{1\le i \le n} |\alpha^{-1} d_i - 1| = O_{\Prob} (\delta_n n \alpha^{-1})
% \end{equation*}
Note that 
\begin{eqnarray*}
    \|E\|  & = & \|\In - \alpha^{-1} D\| = \max_{1\le i \le n} |\alpha^{-1} d_i - 1|\\
    &\le&  \max_{1\le i \le n} |\alpha^{-1} \E[d_i] - 1| + \alpha^{-1}\max_{1\le i \le n} |d_i- \E[d_i]|.
   % \\ & = & O_{\Prob} (\delta_n n \alpha^{-1}) + O_{\Prob}(1).
\end{eqnarray*}
Because one must have $|P_{ij} - p| \le \delta_n$, for every $1\le i \le n$,
\begin{equation*}
    |\E[d_i] - (n-1)p| \le \sum_{j\ne i} |P_{ij}  - p| \le n\delta_n.
\end{equation*}
Therefore
\begin{equation*}
    \max_{1\le i \le n} |\alpha^{-1} \E[d_i] - 1| \le \alpha^{-1}  \max_{1\le i \le n} |\E[d_i] - (n-1)p - 1| \le \alpha^{-1} (n\delta_n + 1).
\end{equation*}
On the other hand, we can find a positive function $\omega(n)$ such that $\frac{\E[d_i]}{\omega(n) \log n} \to \infty$ for every $i$ while $\omega(n) \to \infty$.
The Chernoff's inequality (Theorem 2.3.1, Exercise 2.3.5 in \citet{vershynin2018high}) implies that for every $1 \le i \le n$,
\begin{equation*}
    \Prob\left(|d_i - \E[d_i]| \ge (\log n \cdot \omega(n))^{\frac{1}{2}} (\E[d_i])^{1/2}\right) \le 2 \exp(-c \log n \cdot \omega(n))
\end{equation*}
for some absolute constant $c>0$.
% and any $\epsilon' \in (0,\epsilon)$.
% Here, we can always take $(\log n )^{\frac{1}{2}(1 + \epsilon')}(\E[d_i])^{-1/2} \in (0,1]$ when $n$ is large enough, due to our assumption $\min_{1\le i \le n} \E[d_i] = \Omega((\log n)^{1+\epsilon})$. 
Here, we can always take $(\log n \cdot \omega(n) )^{\frac{1}{2}}(\E[d_i])^{-1/2} \in (0,1]$ when $n$ is large enough, due to our assumption $\frac{\E[d_i]}{\omega(n) \log n} \to \infty$. 
Then, we take the union bound over all $n$ nodes,
% and recall the assumption $(\log n)^{-1}\min_{1\le i \le n}\E[d_i] \to \infty$,
% \begin{equation*}
%     \Prob\left\{\exists 1 \le i \le n:  |d_i - \E[d_i]| \ge (\log n)^{\frac{1}{2}(1 + \epsilon')} (\E[d_i])^{1/2} \right\} \le 2 n\exp(-c (\log n)^{1 + \epsilon'}) = o(1).
% \end{equation*}
\begin{equation*}
    \Prob\left\{\exists 1 \le i \le n:  |d_i - \E[d_i]| \ge (\log n \cdot \omega(n))^{\frac{1}{2}} (\E[d_i])^{1/2} \right\} \le 2 n \exp(-c \log n \cdot \omega(n)) = o(1).
\end{equation*}
Thus, we have shown that with probability tending to 1,
\begin{equation*}
    \max_{1\le i \le n} |d_i- \E[d_i]| \le (\log n \cdot \omega(n))^{\frac{1}{2}} (\max_i\E[d_i])^{1/2} \le \left(\log n \cdot \omega(n) (\alpha - 1 + n\delta_n)\right)^{1/2}.
\end{equation*}
Combining the two bounds above, we conclude that 
\begin{eqnarray*}
    \|E\|&\le&  \alpha^{-1} (n\delta_n + 1) + \alpha^{-1}O_{\Prob} \left(\left(\log n \cdot \omega(n) (\alpha - 1 + n\delta_n)\right)^{1/2}\right)\\
    &=&  O_{\Prob} (\delta_n n \alpha^{-1}) + O_{\Prob} (1) \\
    & = & O_{\Prob}(1).
\end{eqnarray*}
% since $\delta_n\le 1 = O(1)$ by its definition.
The last equation is due to the fact that $\delta_n n = O(\alpha)$,
which can be inferred from the following inequality:
\begin{equation*}
    n \left(\min_{1\le i \le K} \frac{n_i}{n} \right)\delta_n \le n_{i} \max_{i,j} Q_{i,j} \le \max_i \E[d_i] \le \left(\min_{1\le i \le K} \frac{n_i}{n}\right)^{-1} \alpha,
\end{equation*}
and our assumption $\min_{1\le i \le K} \frac{n_i}{n} = \Omega(1)$.

Finally, the empirical counterpart 
\begin{equation*}
    \widehat{\alpha}^{-1/2} |\mu_1(\Hce) - \lambda_1(\Ace)| =  O_{\Prob} (\delta_n n \alpha^{-1}) + O_{\Prob}(1) = O_{\Prob}(1)
\end{equation*}
can be shown similarly.
In particular, one just needs to consider the empirical version of $\widetilde{H}$, denoted as $\widehat{\widetilde{H}}$:
\begin{equation*}
    \widehat{\widetilde{H}} = \left(\begin{array}{cc}
       \widehat{\alpha}^{-1/2}  \Ace  &  \In - \widehat{\alpha}^{-1} D \\
       \In  &  \bm{0}
    \end{array} \right),
\end{equation*}
and its decomposition $\widehat{\widetilde{H}} = \widehat{\widetilde{H}}_0 + \widehat{E}$,
\begin{equation*}
    \widehat{\widetilde{H}}_0
    = \left(\begin{array}{cc}
       \widehat{\alpha}^{-1/2}  \Ace  &  \bm{0} \\
       \In  &  \bm{0}
    \end{array} \right), 
    \quad 
    \widehat{E}=
    \left(\begin{array}{cc}
       \bm{0}  &  \In - \widehat{\alpha}^{-1} D \\
       \bm{0}  &  \bm{0}
    \end{array} \right).
    % \label{eq:definition_of_tildeH_hat}
\end{equation*}
Similar to $\widetilde{H}_0$, the nontrivial eigenvalues of $\widehat{\widetilde{H}}_0$ are given by the eigenvalues of $\widehat{\alpha}^{-1/2} \Ace$.
One can similarly apply Theorem~\ref{theorem:eigenvalue_perturbation} and then bound $\|\widehat{E}\|= \|\In - \widehat{\alpha}^{-1} D\|$. 
\end{proof}

\section{Eigenvalues of block matrix}
\begin{lemma}[Eigenvalues of block matrix]
\label{lemma:Eigenvalues of block matrix}
% Suppose graph $A$ is generated from an SBM model with $n$ nodes belonging to $K$ communities, where $n_i$ is the size of the $i$th community and $\sum_{i=1}^K n_i = n$
% , and $B \in \mathbb{R}^{K \times K}$ is the symmetric block-wise edge probability matrix. 
% $K$ of the eigenvalues of 
Suppose $B \in \mathbb{R}^{n\times n}$ is a block-wise real symmetric matrix with $K$ blocks each with size $n_i$, and $\sum_{i=1}^K n_i = n$. In particular, $B$ has the form
\begin{equation*}
    B = \left( \begin{array}{cccc}
        B_{11} (\bm{1}_{n_1}\bm{1}_{n_1}^{\top} - \ell_1\bm{I}_{n_1}) & B_{12} \bm{1}_{n_1}\bm{1}_{n_2}^{\top} & \ldots &  B_{1K} \bm{1}_{n_1}\bm{1}_{n_K}^{\top} \\
       B_{12} \bm{1}_{n_2}\bm{1}_{n_1}^{\top}  & B_{22} (\bm{1}_{n_2}\bm{1}_{n_2}^{\top} - \ell_2\bm{I}_{n_2}) & \ldots & B_{1K} \bm{1}_{n_2}\bm{1}_{n_K}^{\top} \\
       \ldots & \ldots & \ldots & \ldots \\
       B_{1K} \bm{1}_{n_K}\bm{1}_{n_1}^{\top}  & B_{2K} \bm{1}_{n_K}\bm{1}_{n_2}^{\top} & \ldots & B_{KK} (\bm{1}_{n_K}\bm{1}_{n_K}^{\top} -\ell_K\bm{I}_{n_K}) \\
    \end{array}\right).
\end{equation*}
Then $K$ of the eigenvalues of $B$ are given by the eigenvalues of
\begin{equation*}
    % P^{1/2} Q P^{1/2} - \left(\begin{array}{cccc}
    %    Q_{11}  &  0 & \ldots & \bm{0} \\
    %    0  & Q_{22} & \ldots & \bm{0} \\
    %    \ldots & \ldots & \ldots & \ldots \\
    %    0  & \bm{0} & \ldots & Q_{kk} \\
    % \end{array}\right) = 
    \left(\begin{array}{cccc}
       B_{11}(n_1-\ell_1)  &  \sqrt{n_1 n_2} B_{12} & \ldots & \sqrt{n_1 n_k} B_{1k} \\
       \sqrt{n_2 n_1} B_{12}  & B_{22}(n_2 -\ell_2) & \ldots & \sqrt{n_2 n_k} B_{2k} \\
       \ldots & \ldots & \ldots & \ldots \\
       \sqrt{n_kn_1} B_{1k} & \sqrt{n_2 n_k} B_{2k} & \ldots & B_{kk}(n_k-\ell_K) \\
    \end{array}\right).
\end{equation*}
% where $\ell = 1$ if we impose the no-self-loop restriction $A_{ii}=0$, and $\ell=0$ otherwise. 
The other $(n-K)$ eigenvalues are given by $\{-\ell_i B_{ii}\}_{i=1}^K$, each with multiplicity $n_i - 1$.
\end{lemma}
\begin{proof}
We give the proof for $K=2$. The other cases with $K>2$ can be easily generalized similarly. When $K=2$,
 \begin{equation*}
    B = \left( \begin{array}{cc}
        B_{11} (\bm{1}_{n_1}\bm{1}_{n_1}^{\top} - \ell_1\bm{I}_{n_1}) & B_{12} \bm{1}_{n_1}\bm{1}_{n_2}^{\top}  \\
        B_{12} \bm{1}_{n_2}\bm{1}_{n_1}^{\top} &  B_{22} (\bm{1}_{n_2}\bm{1}_{n_2}^{\top} - \ell_2\bm{I}_{n_2})
    \end{array}\right) .
\end{equation*}   
Select an orthogonal matrix $P \in \mathbb{R}^{n_1 \times n_1}$ such that $P \bm{1}_{n_1} = \sqrt{n_1} \bm{e}_{n_1}$, where $\bm{e}_{n_1} \in \mathbb{R}^{n_1}$ is the column vector $(1, 0, \ldots, 0)^{\top}$ with dimension $n_1$.
Similarly select an orthogonal matrix $Q \in \mathbb{R}^{n_2 \times n_2}$ such that $Q \bm{1}_{n_2} = \sqrt{n_2} \bm{e}_{n_2}$.
Let $M = \mathrm{diag}(P, Q)$.
Then
\begin{eqnarray*}
    MBM^{\top} & = &
    \left(\begin{array}{cc}
      P   & \bm{0} \\
      0   & Q
    \end{array}\right) 
     \left( \begin{array}{cc}
        B_{11} (\bm{1}_{n_1}\bm{1}_{n_1}^{\top} -\ell_1 \bm{I}_{n_1}) & B_{12} \bm{1}_{n_1}\bm{1}_{n_2}^{\top}  \\
        B_{12} \bm{1}_{n_2}\bm{1}_{n_1}^{\top} &  B_{22} (\bm{1}_{n_2}\bm{1}_{n_2}^{\top} -\ell_2 \bm{I}_{n_2})
    \end{array}\right)
    \left(\begin{array}{cc}
      P^\top   & \bm{0} \\
      0   & Q^\top
    \end{array}\right) \\
    & = & \left(\begin{array}{cc}
       B_{11}(n_1 \bm{e}_{n_1} \bm{e}_{n_1}^{\top} - \ell_1\bm{I}_{n_1})   & \sqrt{n_1 n_2} B_{12} \bm{e}_{n_1} \bm{e}_{n_2}^{\top} \\
      \sqrt{n_1 n_2} B_{12} \bm{e}_{n_2} \bm{e}_{n_1}^{\top}   &  B_{22}(n_2 \bm{e}_{n_2} \bm{e}_{n_2}^{\top} -\ell_2 \bm{I}_{n_2})
    \end{array}\right).
\end{eqnarray*}
There exists a permutation matrix $R$ such that
\begin{equation*}
    R (MBM^{\top}) R^{\top} = \left(\begin{array}{cccc}
       B_{11}(n_1-\ell_1)  &  \sqrt{n_1 n_2} B_{12} & \bm{0} & \bm{0} \\
       \sqrt{n_2 n_1} B_{12}  & B_{22}(n_2 -\ell_2) & \bm{0} & \bm{0} \\
       0 & \bm{0} & -B_{11}\ell_1 \bm{I}_{n_1-1} & \bm{0} \\
       0 & \bm{0} & \bm{0} & -B_{22} \ell_2\bm{I}_{n_2-1} \\
    \end{array}\right),
\end{equation*}
and the claim is proved.
\end{proof}

\section{Proof of Proposition~\ref{proposition:lambda_1_EA_centered_larger_than_lambda_2_EA}: Inequality $\lambda_1(\E[A] - P^{(0)}) \ge \lambda_2(\E[A])$}
\label{section:proof_of_proposition:lambda_1_EA_centered_larger_than_lambda_2_EA}
\begin{proof}[Proof of Proposition~\ref{proposition:lambda_1_EA_centered_larger_than_lambda_2_EA}]
Let $p = \frac{n_1}{n}$ and $q = 1-p$. By Lemma~\ref{lemma:Eigenvalues of block matrix}, the two nontrivial eigenvalues of $\E[A]$ are given by the eigenvalues of 
\begin{equation*}
    n p_0\left(\begin{array}{cc}
       p (1+x\delta)  &  \sqrt{pq} (1 - \delta)\\
       \sqrt{pq} (1 - \delta)  &  q (1+kx\delta)
    \end{array}\right),
\end{equation*}
which are
\begin{equation*}
    \lambda_1(\E[A]), \lambda_2(\E[A]) = \frac{1}{2}np_0\left(1 + \delta x(p+kq) \pm \sqrt{(p - q + (p-kq) x \delta)^2 + 4 p q (1-\delta) ^2})\right). 
    % \label{eq:leading_eigenvalues_of_EA}
\end{equation*}
Meanwhile, the two nontrivial eigenvalues of $\E[A] - P^{(0)}$ are given by the eigenvalues of 
\begin{equation*}
    n p_0\delta \left(\begin{array}{cc}
       p x  &  -\sqrt{pq} \\
       -\sqrt{pq}   &  k qx
    \end{array}\right),
\end{equation*}
which are
\begin{equation*}
    \lambda_1(\E[A] - P^{(0)}), \lambda_2(\E[A]- P^{(0)}) = \frac{1}{2}np_0\delta\left(x(p+kq) \pm \sqrt{((p-kq) x)^2 + 4 p q })\right). 
    % \label{eq:leading_eigenvalues_of_EA-P0}
\end{equation*}

% The proof of Proposition~\ref{proposition:lambda_1_EA_centered_larger_than_lambda_2_EA} relies on Lemma~\ref{lemma:lambda_1_EA_centered_larger_than_lambda_2_EA}.
% \end{proof}

% \begin{lemma}
% \label{lemma:lambda_1_EA_centered_larger_than_lambda_2_EA}
% Suppose $\frac{1}{2} \le p < 1$, $q = 1- p$, and $x=\frac{2pq}{p^2 + kq^2}$.
% We have 
We just need to show that
the following inequality holds for all $k\in [0, \infty)$:
\begin{equation}
    \delta \sqrt{x^2 (p-kq)^2 + 4pq} + \sqrt{(p - q + (p-kq) x \delta)^2 - 4 p q (1-\delta) ^2} \ge 1.
    \label{eq:inequality_for_showing_lambda_1_EA_centered_larger}
\end{equation}
% \end{lemma}
% \begin{proof}
Rewrite 
\begin{equation*}
    (p - q + (p-kq) x \delta)^2 - 4 p q (1-\delta) ^2 = (x^2(p-kq)^2+4pq)\delta^2 + 2(x(p-kq)(p-q) - 4pq) + 1.
\end{equation*}
Let $\alpha = x^2(p-kq)^2+4pq$ and $\beta =(p-kq)(p-q)x - 4pq = -pq((1+k)x + 2)$.
If we can show
\begin{equation*}
    \beta^2 \le \alpha,
\end{equation*}
then we can easily show \eqref{eq:inequality_for_showing_lambda_1_EA_centered_larger} because
\begin{equation*}
    \textrm{LHS} = \delta\sqrt{\alpha} + \sqrt{\alpha \delta^2 +1 +2\beta\delta} \ge \delta\sqrt{\alpha} + \sqrt{\alpha \delta^2 +1 -2\sqrt{\alpha}\delta} = 1.
\end{equation*}

Let $r = \frac{p}{q} \in [1, \infty)$, then $x = \frac{2 p q}{ p^2+ kq^2} = \frac{2}{r + kr^{-1}}$.
We can rewrite
\begin{equation*}
    \alpha = 4\frac{r}{(r+1)^2} \left[\frac{r}{(r^2+k)^2}(r-k)^2 + 1\right], \quad 
    \beta = \frac{r}{(r+1)^2}\left[\frac{(1+k)2r}{r^2 + k} + 2\right] = \frac{2r}{r+1}\frac{r+k}{r^2 + k}.
\end{equation*}
Therefore, 
\begin{eqnarray*}
    \alpha - \beta^2&  = &4\frac{r}{(r+1)^2}\left[\frac{r}{(r^2+k)^2}\left((r-k)^2 - (r+k)^2\right) + 1\right]\\
    & = & 4\frac{r}{(r+1)^2}\left[\frac{(r^2 +k)^2 - 4r^2k}{(r^2+k)^2}\right] \\
    & = & 4\frac{r}{(r+1)^2}\frac{(r^2 - k)^2}{(r^2+k)^2}\ge 0,
\end{eqnarray*}
with equality holds if and only if $k = r^2 = p^2/(1-p)^2$.
\end{proof}

\section{Proof of Proposition \ref{proposition:spectrum_of_B_equivalent_to_spectrum_of_H}: Spectral Equivalence of $\Bc$ and $\Hc$}
\label{section:proof_of_proposition_spectrum_of_B_equivalent_to_spectrum_of_H}
\begin{proof}[Proof of Proposition \ref{proposition:spectrum_of_B_equivalent_to_spectrum_of_H}]
    Given a vector $\tg \in \mathbb{R}^{\vec{E}(V)}$,
define $\tg^{\textrm{out}}$ and $\tg^{\textrm{in}}$ as the $n$-dimensional vectors
\begin{equation*}
    \tg^{\textrm{in}}_u = \sum_{v \ne u} \Ac_{uv} \tg_{v \to u}, 
    \quad 
    \tg^{\textrm{out}}_u = \sum_{v \ne u} \Ac_{uv} \tg_{u \to v}.
\end{equation*}
Then 
\begin{eqnarray*}
    (\Dc_u - 1) \tg^{\textrm{out}}_u & = & (\sum_{v \ne u} \Ac_{uv} - 1) \left(\sum_{v\ne u} \Ac_{uv} \tg_{u\to v}\right) \\
    & = & \sum_{v \ne u} \Ac_{uv} \sum_{w \ne u} \Ac_{uw} \tg_{u\to w} - \sum_{v \ne u} \Ac_{uv} \tg_{u\to v} \\
    & = & \sum_{v \ne u} \Ac_{uv} \left(\sum_{w \ne u}\Ac_{uw} \tg_{u\to w} -\tg_{u\to v} \right) \\
    & = & \sum_{v \ne u} \Ac_{uv} (\Bc \tg)_{v \to u} \\
    & = & (\Bc \tg)^{\textrm{in}}_u,
\end{eqnarray*}
and
\begin{eqnarray*}
    \sum_{v \ne u} \Ac_{uv} \tg^{\textrm{out}}_v - \tg^{\textrm{in}}_u
    & = & \sum_{v \ne u} \Ac_{uv} \sum_{w \ne v} \Ac_{vw} \tg_{v\to w} - \sum_{v \ne u} \Ac_{uv} \tg_{v\to u} \\
    & = & \sum_{v \ne u} \Ac_{uv} \left(\sum_{w \ne v}\Ac_{vw} \tg_{v\to w} -\tg_{v\to u} \right) \\
    & = & \sum_{v \ne u} \Ac_{uv} (\Bc \tg)_{u \to v} \\
    & = & (\Bc \tg)^{\textrm{out}}_u.
\end{eqnarray*}
Written in matrix form, we have
\begin{equation*}
    \left(\begin{array}{c}
       (\Bc \tg)^{\textrm{out}}  \\
       (\Bc \tg)^{\textrm{in}}  
    \end{array}\right) 
    = \left(\begin{array}{cc}
       \Ac & -\In \\
        \Dc - \In & \bm{0}  
    \end{array}
    \right)  \left(\begin{array}{c}
       \tg^{\textrm{out}}  \\
       \tg^{\textrm{in}}  
    \end{array}\right) 
\end{equation*}
Note that this matrix $\left(\begin{array}{cc}
       \Ac & -\In \\
        \Dc - \In & \bm{0}  
    \end{array}
    \right)$ has the same spectrum as $\Hc$ \eqref{eq:definition_of_Hc}.
\end{proof}

\end{document}